\documentclass[a4paper]{lmcs}

\usepackage{nico}
\usepackage{macros}
\usepackage{hyperref}
\usepackage{config-tikz}
\usepackage{config-tikz2}

\usepackage{multirow}

\usepackage[hyperref,xcolor,paper]{knowledge}
\knowledgeconfigure{notion,AP shape={}}
\knowledge{notion}
| marking
| markings
| Markings
\def\submarking{\withkl{\kl[\submarking]}{\mathrel{\cmdkl{\sqsubseteq}}}}
\knowledge{notion}
| \submarking
\def\mkminus{\withkl{\kl[\mkminus]}{\mathbin{\cmdkl{\setminus}}}}
\knowledge{notion}
| \mkminus
\def\mkimg#1{\withkl{\kl[\mkimg]}{\mathbin{\cmdkl{\textsf{img}(#1)}}}}
\knowledge{notion}
| \mkimg
\knowledge{notion}
| submarking
| submarkings
\knowledge{notion}
| unitary marking
| unitary markings
| unitary submarking
| unitary submarkings
| unitary
\def\mknu{\withkl{\kl[\mknu]}{\cmdkl{\nu}}\xspace}
\knowledge{notion}
| \mknu
\knowledge{notion}
| marking induced
| induced
\knowledge{notion}
| multiset
| multisets
\def\mset#1{\withkl{\kl[\mset]}{\cmdkl{\{\!\!\{}#1\cmdkl{\}\!\!\}}}}
\knowledge{notion}
| \mset

\knowledge{notion}
| empty multiset
\def\supp{\withkl{\kl[\supp]}{\cmdkl{\textsf{supp}}}}
\knowledge{notion}
| \supp
\knowledge{notion}
| support
| support of a multiset
\def\MSsize#1{\withkl{\kl[\MSsize]}{\cmdkl{|}#1\cmdkl{|}}}
\knowledge{notion}
| \MSsize
\knowledge{notion}
| size of a multiset
| size@MS
\def\submultiset{\withkl{\kl[\submultiset]}{\mathrel{\cmdkl{\sqsubseteq}}}}
\knowledge{notion}
| \submultiset
| submultiset
\def\submultisetneq{\withkl{\kl[\submultisetneq]}{\mathrel{\cmdkl{\sqsubset}}}}
\knowledge{notion}
| \submultisetneq
\def\mscup{\withkl{\kl[\mscup]}{\mathbin{\cmdkl{\uplus}}}}
\knowledge{notion}
| \mscup
\def\msminus{\withkl{\kl[\msminus]}{\mathbin{\cmdkl{\setminus}}}}
\knowledge{notion}
| \msminus
\def\proj{\withkl{\kl[\proj]}{\cmdkl{\textsf{proj}}}}
\knowledge{notion}
| \proj

\def\Aword#1{\withkl{\kl[word]}{\cmdkl{$#1$-word}}\xspace}
\knowledge{notion}
| word
| words
| Words

\def\prefix#1#2{\withkl{\kl[prefix]}{\cmdkl{\ensuremath{#1_{\langle #2 )}}}}\xspace}
\knowledge{notion}
| prefix
\def\suffix#1#2{\withkl{\kl[suffix]}{\cmdkl{\ensuremath{#1_{[ #2 \rangle}}}}\xspace}
\knowledge{notion}
| suffix

\def\EXP#1#2{\withkl{\kl[\EXP]}{\cmdkl{\ensuremath{#1\textsf{-exp}(#2)}}\xspace}}
\knowledge{notion}
| \EXP

\def\emptyw{\withkl{\kl[\emptyw]}{\cmdkl{\varepsilon}}}
\knowledge{notion}
| \emptyw
\knowledge{notion}
| empty word
\def\Wsize#1{\withkl{\kl[\Wsize]}{\cmdkl{\left|#1\right|}}}
\knowledge{notion}
| \Wsize
\knowledge{notion}
| size@W
| length
| length@W
| size of a word

\def\Dtree#1{\withkl{\kl[tree]}{\cmdkl{$#1$-tree}}\xspace}

\def\Stree#1{\withkl{\kl[Ltree]}{\cmdkl{$#1$-labelled tree}}\xspace}
\def\Strees#1{\withkl{\kl[Ltree]}{\cmdkl{$#1$-labelled trees}}\xspace}
\def\SDtree#1#2{\withkl{\kl[Ltree]}{\cmdkl{$#1$-labelled $#2$-tree}}\xspace}

\def\SBtree#1{\withkl{\kl[Btree]}{\cmdkl{$#1$-labelled  binary tree}}\xspace}

\knowledge{notion}
| tree
| trees
| Trees
| tree structure
| SDtree
\knowledge{notion}
| binary tree
| binary-tree
| binary trees
| Btree
\knowledge{notion}
| labelled tree
| Ltree
\knowledge{notion}
| arity
\knowledge{notion}
| leaf
| leaves
\knowledge{notion}
| tree accepted
| accepted
| accepts
| accept
| acceptance
| accepting
\knowledge{notion}
| rejected
\knowledge{notion}
| accepting@B
\knowledge{notion}
| accepting@ET
\knowledge{notion}
| accepting@powET
| accepted@powET
\knowledge{notion}
| accepting@powB

\knowledge{notion}
| appear
| appears
| appearing

\def\troot{\withkl{\kl[\troot]}{\cmdkl{\varepsilon}}}
\knowledge{notion}
| \troot
\knowledge{notion}
| root
| roots
\def\succ{\withkl{\kl[\succ]}{\cmdkl{\textsf{succ}}}}
\knowledge{notion}
| \succ
\knowledge{notion}
| successor
| successors
| successor node
| successor nodes

\knowledge{notion}
| branch
| Branch
| branches
\def\Bword#1{\withkl{\kl[\Bword]}{\cmdkl{w(#1)}}}
\knowledge{notion}
| \Bword
| word@B
| word of a branch

\knowledge{notion}
| Kripke structure
| Kripke structures
\knowledge{notion}
| path@KS
| paths@KS
\knowledge{notion}
| regular

\def\EUpr{\withkl{\kl[EUpair]}{\cmdkl{\EU-pair}}\xspace}
\def\EUprs{\withkl{\kl[EUpair]}{\cmdkl{\EU-pairs}}\xspace}
\def\EUconstr{\withkl{\kl[EUconstraint]}{\cmdkl{\EU-constraint}}\xspace}
\def\EUconstrs{\withkl{\kl[EUconstraint]}{\cmdkl{\EU-constraints}}\xspace}
\knowledge{notion}
| EUpair
\knowledge{notion}
| EUconstraint
\def\PBF#1{\withkl{\kl[\PBF]}{\cmdkl{\textsf{PBF}}(#1)}\xspace}
\def\DBF#1{\withkl{\kl[\DBF]}{\cmdkl{\textsf{DBF}}(#1)}\xspace}
\knowledge{notion}
| \PBF
| positive boolean combination
| positive boolean combinations
| Positive boolean formulas
| positive boolean formula
\knowledge{notion}
| \DBF
| disjunctions
\def\EUset(#1){\withkl{\kl[\EUset]}{\cmdkl{\EU}(#1)}\xspace}
\knowledge{notion}
| \EUset

\knowledge{notion}
| parity
| parity condition
| parity acceptance
| Parity acceptance
\def\prio{\omega}
\newcommand\priomax[1][\prio]{\withkl{\kl[\priomax]}{\cmdkl{|}#1\cmdkl{|}}}
\knowledge{notion}
| \priomax
\knowledge{notion}
| priority
| priority function
| priorities
\knowledge{notion}
| maximal priority

\knowledge{notion}
| computation tree

\knowledge{notion}
| execution tree
| execution subtree
| execution trees
\knowledge{notion}
| execution tree@pow
| execution subtree@pow
| execution trees@pow
\knowledge{notion}
| execution tree@pow
\knowledge{notion}
| minimal execution tree
| minimal

\knowledge{notion}
| direct subtree
| direct subtrees

\newrobustcmd\MSO{\withkl{\kl[MSO]}{\cmdkl{\textsf{MSO}}}\xspace}
\knowledge{notion}
| MSO
| Monadic Second-Order Logic
| \MSO

\makeatletter
\def\@QnumCTL#1{\expandafter\newcommand\csname #1CTL\endcsname[1][]{%
  \if\relax##1\relax
    \withkl{\kl[{\csname #1CTL\endcsname[]}]}{\cmdkl{\textsf{#1CTL}}}%
  \else
    \withkl{\kl[{\csname #1CTL\endcsname[k]}]}{\cmdkl{\textsf{#1\textsuperscript{##1}CTL}}}%
  \fi
  \xspace}
  \knowledge{\csname #1CTL\endcsname[]}{notion}
  \knowledge{\csname #1CTL\endcsname[k]}{notion}%
  \expandafter\newcommand\csname #1CTLs\endcsname[1][]{%
  \if\relax##1\relax
    \withkl{\kl[{\csname #1CTLs\endcsname[]}]}{\cmdkl{\textsf{#1CTL\textsuperscript*}}}%
  \else
    \withkl{\kl[{\csname #1CTLs\endcsname[k]}]}{\cmdkl{\textsf{#1\textsuperscript{##1}CTL\textsuperscript*}}}%
  \fi
  \xspace}
  \knowledge{\csname #1CTLs\endcsname[]}{notion}
  \knowledge{\csname #1CTLs\endcsname[k]}{notion}%
  \expandafter\newcommand\csname #1CTLp\endcsname[1][]{%
  \if\relax##1\relax
    {\textsf{#1CTL\textsuperscript+}}
  \else
    {\textsf{#1\textsuperscript{##1}CTL\textsuperscript+}}
  \fi
  \xspace}
  \knowledge{\csname #1CTLp\endcsname[]}{notion}
  \knowledge{\csname #1CTLp\endcsname[k]}{notion}}
\@QnumCTL{}
\@QnumCTL{Q}
\@QnumCTL{EQ}
\@QnumCTL{AQ}
\makeatother

\def\equivF{\withkl{\kl[\equivF]}{\mathrel{\cmdkl{\equiv}}}}
\knowledge{notion}
| equivalent@F
| \equivF

\knowledge{notion}
| tree semantics
\newrobustcmd\FATAutomata{\withkl{\kl[\FATAutomata]}{\cmdkl{Fixed-arity tree automata}}\xspace}
\newrobustcmd\faTAutomata{\withkl{\kl[\FATAutomata]}{\cmdkl{fixed-arity tree automata}}\xspace}
\newrobustcmd\faTAutomaton{\withkl{\kl[\FATAutomata]}{\cmdkl{fixed-arity tree automaton}}\xspace}
\knowledge{notion}
| \FATAutomata
\newrobustcmd\AmTAutomata{\withkl{\kl[\AmTAutomata]}{\cmdkl{Amorphous tree automata}}\xspace}
\newrobustcmd\amTAutomata{\withkl{\kl[\AmTAutomata]}{\cmdkl{amorphous tree automata}}\xspace}
\newrobustcmd\amTAutomaton{\withkl{\kl[\AmTAutomata]}{\cmdkl{amorphous tree automaton}}\xspace}
\knowledge{notion}
| \AmTAutomata

\newrobustcmd\MSOAutomata{\withkl{\kl[\MSOAutomata]}{\cmdkl{\MSOnokl-automata}}\xspace}
\newrobustcmd\MSOAutomaton{\withkl{\kl[\MSOAutomata]}{\cmdkl{\MSOnokl-automaton}}\xspace}
\knowledge{notion}
| \MSOAutomata

\knowledge{notion}
| symmetric
| Symmetric
| symmetric@TA

\newrobustcmd\SymBTAutomata{\withkl{\kl[\SymBTAutomata]}{\cmdkl{Symmetric B\"uchi tree automata}}\xspace}
\newrobustcmd\symBTAutomata{\withkl{\kl[\SymBTAutomata]}{\cmdkl{symmetric B\"uchi tree automata}}\xspace}
\newrobustcmd\symBTAutomaton{\withkl{\kl[\SymBTAutomata]}{\cmdkl{symmetric B\"uchi tree automaton}}\xspace}
\newrobustcmd\SymNBT{\withkl{\kl[\SymBTAutomata]}{\cmdkl{\textsf{SymNBT}}}\xspace}
\newrobustcmd\symNBT{\withkl{\kl[\SymBTAutomata]}{\cmdkl{\textsf{symNBT}}}\xspace}
\knowledge{notion}
| \SymBTAutomata

\newrobustcmd\BDAutomata{\withkl{\kl[\BDAutomata]}{\cmdkl{${\{\Box,\Diamond\}}$-automata}}\xspace}
\newrobustcmd\BDAutomaton{\withkl{\kl[\BDAutomata]}{\cmdkl{${\{\Box,\Diamond\}}$-automaton}}\xspace}
\knowledge{notion}
| \BDAutomata

\knowledge{notion}
| \EU tree automaton
| \EU tree automata
| \EU-automaton
| \EU-automata
\knowledge{notion}
| alternating
| Alternating
| alternation
\knowledge{notion}
| non-alternating
| Non-alternating
\newrobustcmd\AATA{\withkl{\kl[\AATA]}{\cmdkl{\textsf{A\EU{}TA}}}\xspace}
\newrobustcmd\AATAs{\withkl{\kl[\AATA]}{\cmdkl{\textsf{A\EU{}TAs}}}\xspace}
\knowledge{notion}
| \AATA
\newrobustcmd\powAATA{\withkl{\kl[\powAATA]}{\cmdkl{$\Aut$-powerset \textsf{A\EU{}TA}}}\xspace}
\newrobustcmd\powAATAs{\withkl{\kl[\powAATA]}{\cmdkl{$\Aut$-powerset \textsf{A\EU{}TAs}}}\xspace}
\newrobustcmd\powAuta{\withkl{\kl[\powAATA]}{\cmdkl{$\Aut$-powerset automata}}\xspace}
\newrobustcmd\powAuton{\withkl{\kl[\powAATA]}{\cmdkl{$\Aut$-powerset automaton}}\xspace}
\knowledge{notion}
| \powAATA
\newrobustcmd\nAATA{\withkl{\kl[\nAATA]}{\cmdkl{\textsf{\EU{}TA}}}\xspace}
\newrobustcmd\nAATAs{\withkl{\kl[\nAATA]}{\cmdkl{\textsf{\EU{}TAs}}}\xspace}
\knowledge{notion}
| \nAATA
\newrobustcmd\AAPTA{\withkl{\kl[\AAPTA]}{\cmdkl{\textsf{A\EU{}PTA}}}\xspace}
\newrobustcmd\AAPTAs{\withkl{\kl[\AAPTA]}{\cmdkl{\textsf{A\EU{}PTAs}}}\xspace}
\knowledge{notion}
| \AAPTA
\newrobustcmd\nAAPTA{\withkl{\kl[\nAAPTA]}{\cmdkl{\textsf{\EU{}PTA}}}\xspace}
\newrobustcmd\nAAPTAs{\withkl{\kl[\nAAPTA]}{\cmdkl{\textsf{\EU{}PTAs}}}\xspace}
\knowledge{notion}
| \nAAPTA
\newrobustcmd\nAPWA{\withkl{\kl[\nAPWA]}{\cmdkl{\textsf{PWA}}}\xspace}
\knowledge{notion}
| parity word automaton
| Parity automata on words
| \nAPWA

\def\Fsize#1{\withkl{\kl[\Fsize]}{\cmdkl{\left|#1\right|}}}
\knowledge{notion}
| \Fsize
| size of a formula
\def\Asize#1{\withkl{\kl[\Asize]}{\cmdkl{\left|#1\right|}}}
\knowledge{notion}
| \Asize
| size of an \AAPTA
| size@A
\def\sizeB#1{\withkl{\kl[\sizeB]}{\cmdkl{|}#1\cmdkl{|_{\textsf{bool}}}}}
\knowledge{notion}
| \sizeB
\def\sizeU#1{\withkl{\kl[\sizeU]}{\cmdkl{|}#1\cmdkl{|_{\textsf{U}}}}}
\knowledge{notion}
| \sizeU
\def\sizeE#1{\withkl{\kl[\sizeE]}{\cmdkl{|}#1\cmdkl{|_{\textsf{E}}}}}
\knowledge{notion}
| \sizeE

\makeatletter
\def\modelsbar{\withkl{\kl[\models|]}{\mathrel{\cmdkl{\parallel\joinrel=}}}}
\def\modelsplus{\withkl{\kl[\models+]}{\mathrel{\cmdkl{\mid\joinrel\equiv}}}}
\def\modelsstar{\withkl{\kl[\models*]}{\mathrel{\cmdkl{\parallel\joinrel\equiv}}}}
\let\savemodels\models
\def\models{\@ifnextchar+{\modelsplus\@gobble}{%
    \@ifnextchar*{\modelsstar\@gobble}{%
    \@ifnextchar|{\modelsbar\@gobble}{%
       \withkl{\kl[\models]}{\mathrel{\cmdkl{\savemodels}}}}}}}
\makeatother
\knowledge{notion}
| \models
\knowledge{notion}
| \models*
| set-satisfies
| set-satisfy
\knowledge{notion}
| \models+
\knowledge{notion}
| \models|

\knowledge{notion}
| blocking pair
| blocking pairs
\knowledge{notion}
| minimal@BP

\knowledge{notion}
| projection
\knowledge{notion}
| simulation
| alternation-removal
| alternation removal
| Alternation removal

\def\slabequiv#1{\withkl{\kl[\labequiv]}{\cmdkl{\ensuremath{\equiv_{#1}}}}}
\def\labequiv#1{\withkl{\kl[\labequiv]}{\cmdkl{$#1$-equivalent}}\xspace}
\def\labequivce#1{\withkl{\kl[\labequiv]}{\cmdkl{$#1$-equivalence}}\xspace}
\knowledge{notion}
| \labequiv
\knowledge{notion}
| equivalent
\knowledge{notion}
| Complementation
| complementation

\def\Lang(#1){\withkl{\kl[language]}{\cmdkl{\calL(}#1\cmdkl{)}}\xspace}
\knowledge{notion}
| language
| languages
| language accepted
| languages accepted

\knowledge{notion}
| game-based semantics
| Game-based semantics
| game semantics
| Game semantics

\knowledge{notion}
| main state
| main states
\knowledge{notion}
| auxiliary state
| auxiliary states

\def\last#1{\withkl{\kl[\last]}{\cmdkl{\textsf{last}}(#1)}\xspace}
\def\first#1{\withkl{\kl[\first]}{\cmdkl{\textsf{first}}(#1)}\xspace}
\knowledge{notion}
| \first
\knowledge{notion}
| \last

\knowledge{notion}
| two-player turn-based parity game
| Two-player turn-based parity game
| Two-player turn-based parity games
| parity game
| parity games
| Parity games
\knowledge{notion}
| path@G
| paths@G
\knowledge{notion}
| maximal@P
\knowledge{notion}
| winning@P
\knowledge{notion}
| winning@S
\knowledge{notion}
| strategy
| strategies
\knowledge{notion}
| compatible
\knowledge{notion}
| memoryless

\def\alae#1{\mathrel{\withkl{\kl[at least as expressive]}{\cmdkl{\succeq}_{#1}}}}
\knowledge{notion}
| at least as expressive
\def\eqex#1{\mathrel{\withkl{\kl[equally expressive]}{\cmdkl{\approxeq}_{#1}}}}
\knowledge{notion}
| equally expressive
| as expressive
\def\smex#1{\mathrel{\withkl{\kl[strictly more expressive]}{\cmdkl{\succstrict}_{#1}}}}
\knowledge{notion}
| strictly more expressive
| more expressive
\def\caplog{\mathop{\withkl{\kl[\caplog]}{\cmdkl{\sqcap}}}}
\knowledge{notion}
| \caplog

\begin{document}
\title{Arbitrary-arity Tree Automata and \QCTL}

\author[F. Laroussinie]{Fran\c cois Laroussinie \lmcsorcid{0009-0002-1353-7942}}[a]
\address{IRIF, Universit\'e Paris Cit\'e, France}
\email{\texttt{francoisl@irif.fr}}

\author[N. Markey]{Nicolas Markey \lmcsorcid{0000-0003-1977-7525}}[b]
\address{IRISA -- Inria, CNRS, Univ. Rennes, France}
\email{\texttt{nicolas.markey@cnrs.fr}}

\thanks{We~thank Igor Walukiewicz for answering our questions on \MSOAutomata, and the reviewers for their careful reading and constructive comments on short versions of this paper~\cite{concur2016-DLM,concur2025-LM}.}

\maketitle

\begin{abstract}

We
introduce a new class of automata (which we coin \kl{\EU-automata}) running
on infinite \kl{trees} of arbitrary (finite) arity.
We~develop and study several algorithms to perform classical
operations (union, intersection,
complement, \kl{projection}, \kl{alternation removal}) for those
automata, and precisely characterise their complexities. We~also
develop algorithms for solving membership and emptiness for the
languages of trees accepted by
\kl{\EU-automata}.

We~then use \kl{\EU-automata} to obtain several algorithmic and
expressiveness results for the temporal logic \QCTL (which
extends \CTL with quantification over atomic propositions) and
for~\MSO. On the one hand, we~obtain decision procedures with optimal
complexity for \QCTL satisfiability and model checking; on~the other
hand, we~obtain an algorithm for translating any \QCTL formula with
$k$~quantifier alternations to formulas with at most one quantifier
alternation, at the expense of a $(k+1)$-exponential blow-up in the
size of the formulas. Using the same techniques, we prove that
any \MSO formula can be translated into a formula with at most four
quantifier alternations (and only one second-order-quantifier
alternation), this time with a $(k+2)$-exponential blow-up in the size of
the formula.

\end{abstract}

\section{Introduction}

\paragraph{Logics and automata.}
The very tight links between logics and automata on infinite \kl{words} and
\kl{trees} date back to the early 1960's with the seminal works of B\"uchi,
Elgot, Trakhtenbrot, McNaughton and
Rabin~\cite{Buc62,Elg61,Tra62,McN66,Rab69}. These early results were
mainly concerned with the \kl{Monadic Second-Order Logic}~(\MSO), and have
been further extended to many other logical formalisms such as modal,
temporal and fix-point logics~\cite{SVW85,VW86b,BVW94,JW95,Wil01}.
Those tight links are embodied as translations back and forth between
various logical languages and corresponding classes of automata;
translations from logics to automata have allowed to derive efficient
algorithms for satisfiability or model checking on the one
hand~\cite{VW86a,EJ91,BVW94};
with additional translations from automata to
logics, we~get effective ways for proving expressiveness or succinctness
results for some of those logics~\cite{Wal96b,Wil99a,LMS02b,KV03b,Zan12}.
In this paper, we~investigate such links between
Quantified \CTL~(\QCTL)~\cite{Kup95a,KMTV00,Fre01,DLM12} and symmetric
tree automata~\cite{Wal96b,Wil99a,KV03b}, and derive algorithmic
and expressiveness
results for \QCTL and its fragments.

\paragraph{\QCTL}
\QCTL extends the classical temporal logic~\CTL with quantification on
atomic propositions. For~instance, formula~$\exists p. \phi$, where
$\phi$ is a \CTL formula, states that there exists a labelling of the
model under scrutiny with atomic proposition~$p$ under which
$\phi$~holds.
\QCTL~is (much) \kl{more expressive} than~\CTL: as~an example, formula
\[
\exists p.\ (\EF(\phi\et p) \et \EF(\phi\et\neg p))
\]
expresses the fact that there are at least two reachable states where~$\phi$~holds.
The~extension of~\CTL with \emph{existential} quantification was first
studied in~\cite{ES84a,Kup95a}:
contrary to~\CTL, the~resulting logic (only allowing formulas in prenex form), which
we call~\EQCTL[1]
hereafter, is sensitive to unwinding and
duplication of transitions; the~semantics thus depends on whether the
extra labelling refers to the \kl{Kripke structure} under scrutiny, or on
its \kl{computation tree}. Our~sample formula above expresses that there are
at least two \emph{different} reachable control states satisfying~$\phi$ in the former case
(which we call the \emph{structure semantics}), while it only requires that two different paths
lead to some $\phi$-states (possibly two copies of the same control state)
in the latter semantics (called the \kl{tree semantics} hereafter).

Universal quantification on atomic propositions can also be added:
\AQCTL[1] is the logic obtained from~\CTL by adding universal
quantification (in~prenex form). Mixing existential and universal quantification
defines an infinite hierarchy of temporal logics, which we~name
\EQCTL[k] and \AQCTL[k], where $k$~is the number of quantifier
alternations allowed in formulas (still assuming prenex form).
\QCTL~allows unrestricted use of both existential and universal quantifications,
and thus contains~\EQCTL[k] and~\AQCTL[k] for all~$k\geq 0$.
It~turns out that \QCTL is \kl{as expressive} as \MSO~\cite{LM14}.

In~this paper, we present several results for \QCTL with the \kl{tree
semantics}. In~particular, we~show that any \QCTL formula with
$k$~quantifier alternations can be translated in~\EQCTL[2], with a
$(k+1)$-exponential blow-up in the size of the formula. Such a result
is known to exist also in~\MSO on \kl{trees}~\cite{Rab69,Tho97b}: any~\MSO
formula can be expressed with two alternations of \emph{second-order}
quantifiers. While \MSO is known to be as expressive as~\QCTL,
this~does not directly entail our result
because \emph{first-order} quantifiers in~\MSO involve extra propositional
quantifiers when translated in~\QCTL.  The~key point of our results is
the introduction of a new class of tree automata that are particularly
well-suited for characterising models of a \QCTL formula, but also
of \QCTLs or~\MSO.

\paragraph{Tree automata.}
We develop (top-down\footnote{There are several families of tree automata:
top-down tree automata explore (finite or infinite) trees starting
from the root; bottom-up tree automata explore finite trees from the
leaves up to the root; tree-walking automata are a kind of two-way
automata for trees. We~refer to~\cite{TATA2008,Boj08b} for more
details.}) tree-automata techniques to study \QCTL. Several results
already exist on this topic~\cite{ES84a,Kup95a,LM14}, but they all
rely on \faTAutomata.

The~limitation %
has several drawbacks. When
dealing with model checking, it~implies that the compilation of the
formula being checked into a tree automaton depends on
the (size of~the) structure under scrutiny.
In~particular, it~cannot be used directly
for evaluating the \emph{program complexity}
of \QCTL model checking, as~it requires bounding the size of the
structures that the automaton can handle. An~indirect solution to this
problem is given in~\cite{LM14}, by~replacing nodes of arbitrary
(finite) arity with \kl{binary-tree} gadgets.
A~similar problem occurs when dealing with satisfiability: one~has to
use additional results to ensure that looking for a structure with
bounded size is sufficient.  More importantly, when deriving
expressiveness results, using fixed-arity tree automata again
restricts the results to \kl{trees} or structures with bounded
branching.

In~order to handle \kl{trees} of
arbitrary branching degree, tree automata must have a \emph{symbolic} way of
expressing transitions, with a finite representation that can cope
with any arity. We~highlight two existing approaches: 
\begin{itemize}
\item Janin and Walukiewicz introduce \MSOAutomata~\cite{JW95,Wal96b}, in~which
  transitions are defined as first-order formulas: quantification is
  over the \kl{successors} of the current node, and predicates indicate in
  which states of the automaton those \kl{successors} must be
  explored. These automata are shown to be as expressive as~\MSO, and
  several expressiveness results have been obtained from this
  construction~\cite{JW95,Wal96b,Wal02,Zan12}. However, to the best of our knowledge, the~exact
  complexity of the operations for manipulating those automata has not
  been studied, so that only qualitative expressiveness results can be
  obtained, and no bounds on the size and complexity of the
  translations can be derived without a more careful study.
\item Wilke introduces \BDAutomata~\cite{Wil99a},
  which are \kl{alternating} tree automata with $\Box q$ and $\Diamond q$ as
  basic blocs for expressing transitions: the~former requires that all
  \kl{successors} be explored in state~$q$, while the latter asks that some
  \kl{successor} be explored in state~$q$. Any~\CTL formula can be turned
  into an equivalent \BDAutomata of linear size; this
  is used to prove that the extension~\CTLp of~\CTL is exponentially
  more succinct that~\CTL. However, \BDAutomata are
  not expressive enough to capture \MSO or~\QCTL.
\end{itemize}

\paragraph{Our contribution.}
In~this paper, we~define a new class of arbitrary-arity \kl{alternating}
tree automata, develop effective operations for their manipulation,
and study the complexity of those operations and the size of the resulting automata. 
Instead of using pairs~$(k,q)$ in the transition function to specify
 that the~$k$-th successor of the current node has to be accepted by
 the automaton in state~$q$, transitions of our automata are defined
 with pairs~$\EUpair(E;U)$, where~$E$~is a \kl{multiset} of states
 that have to occur among the set of states involved in the
 exploration of the \kl{successors} of the current node, while~$U$ is
 a set of states indicating which states are allowed for
 exploring \kl{successor nodes} that are not explored by states
 of~$E$. For~example, $\EUpair(E=\mset{q,q,q'};U=\{q''\})$ requires
 the presence of at least three
\kl{successors} nodes; two~\kl{successors} will be explored in state~$q$, one in
state~$q'$, and the remaining ones (if~any) in state~$q''$. We~name those
automata \kl{\EU-automata}\footnote{In~\cite{KV03b}, Kupferman and Vardi
define another variant of arbitrary-arity \kl{alternating} tree automata in which
transitions are based on pairs~$(U,E)$. Those automata are equivalent
to Wilke's \BDAutomata. We~give an overview of those
automata in our Section~\ref{sec-relW} on related work.}.

\smallskip
It~is not hard to prove that such automata are closed under
conjunction and disjunction, thanks to \kl{alternation}. Closure under
negation is harder to prove: while $\Box$ and~$\Diamond$ are dual to
each other, which provides an easy \kl{complementation} procedure for
\BDAutomata, there is no obvious way of expressing
the negation of \EUprs in terms of \EUprs. We~develop such a
translation, and obtain an exponential \kl{complementation} procedure for
\kl{\EU-automata}.

\kl{Non-alternating} \kl{\EU-automata} are also closed under \kl{projection}, which
is the operation we need to encode quantification over atomic
propositions of~\QCTL,
and first- and second-order quantification in~\MSO.
Finally we~prove that any \kl{alternating} \kl{\EU-automaton} can
be turned into an equivalent \kl{non-alternating} \kl{\EU-automaton}. For~this
operation, we~adapt the \kl{simulation} procedure developed
in~\cite{Wal96b,Zan12} to our setting, and evaluate its exact complexity. 

\smallskip

Putting all the pieces together, we~prove that any \QCTL formula
$\vfi$ can be turned into an
equivalent \kl{\EU-automaton}~$\Aut_\vfi$. The~size of the automaton
is $k$-exponential in the size of~$\vfi$, where~$k$~is the number of
quantifier alternations in~$\vfi$. This~construction then yields
optimal algorithms for model-checking and satisfiability for~\QCTL.
Conversely, we~prove that acceptance by any
\kl{\EU-automaton}   can be expressed as an~\EQCTL[2] formula.
We~obtain similar results for~\MSO.
Therefore \kl{\EU-automata}, \QCTL (and~\EQCTL[2]), and \MSO (even
when restricted to two second-order quantifier alternations)
all characterise exactly the same \kl{tree} languages.

\setcounter{tocdepth}{1}
\tableofcontents

\section{Definitions}

\label{sec-defs}
\label{sec-AATA}

\subsection{Sets and \kl{multisets}}
\label{ssec-multisets}
\AP Let $\Set$ be a countable set.
A~\intro{multiset} over~$\Set$ is a mapping $\mu\colon \Set \to \bbN$.  Sets
are seen as special cases of \kl{multisets} taking values in~$\{0,1\}$.
We~use double-brace notation to distinguish between sets and \kl{multisets}:
$\{a,a,a\}$ is the same as the set~$\{a\}$ with one element, while $\intro*\mset{a,a,a}$ is the three-element \kl{multiset} $a\mapsto 3$.
The~\intro{empty multiset} is the \kl{multiset} mapping all
elements of~$\Set$ to~zero; we~denote~it with~$\emptyset$.

\AP The~\intro{support of a multiset}~$\mu$ is the set
$\intro*{\supp}(\mu)=\{s\in\Set \mid \mu(s)>0\}$.
We~write $s\in\mu$ for $s\in\supp(\mu)$.
The~\intro(MS){size}~$\intro*{\MSsize{\mu}}$
of~$\mu$ is the sum $\sum_{s\in\Set} \mu(s)$;
the~\kl{multiset}~$\mu$ is finite whenever $\MSsize\mu$~is.
\AP For~two \kl{multisets}~$\mu$ and~$\mu'$, we~write $\mu\intro*\submultiset \mu'$,
and say that $\mu$~is a \kl{submultiset}
of~$\mu'$, whenever $\mu(s)\leq \mu'(s)$ for all~$s\in\Set$.
This~defines a partial ordering over \kl{multisets}.
We~write $\mu\intro*\submultisetneq \mu'$ when $\mu\submultiset \mu'$ and $\mu\not=\mu'$.
We~define the following operations on \kl{multisets}:
\[
  \mu\intro*\mscup\mu'\colon s \in\Set \mapsto \mu(s)+\mu'(s) \quad \mbox{and} \quad
  \mu'\intro*\msminus\mu\colon s\in\Set \mapsto \max(0,\mu'(s)-\mu(s))
  \]

\AP Fix a second countable set~$\Set'$. For any~$c=(s,s')\in
\Set\times\Set'$, we~define $\intro*\proj_1(c)=s$ and $\reintro*\proj_2(c)=s'$.

\subsection{\kl{Markings}}
\label{ssec-markings}

\AP Let~$\Set$ and~$\Set'$ be two countable sets. A~\intro{marking} of~$\Set'$
by~$\Set$ is a mapping ${\nu\colon \Set' \to 2^{\Set}\setminus\{\emptyset\}}$
decorating each
element of~$\Set'$ with a (non-empty) subset of~$\Set$. A~\kl{marking}~$\nu$ is a
\intro{submarking} of a \kl{marking}~$\nu'$, denoted~$\nu\intro*\submarking\nu'$, whenever
$\nu(s')\subseteq \nu'(s')$ for all~$s'\in \Set'$.

\AP
A~\kl{marking}  $\nu$ is \intro{unitary} when $\size{\nu(s')}=1$ for
all~$s'\in \Set'$; \kl{unitary markings} can be seen as mappings from~$\Set'$
to~$\Set$.
For a \kl{unitary marking}~$\nu$ and a subset~$T$ of~$\Set'$,
we~write $\nu(T)$ for the \kl{multiset}~$\mu$
over~$\Set$ defined as $\mu(s)=\#\{t\in T \mid \nu(t)=s\}$,
which we may also write as $\mset{\nu(t) \mid t\in T}$.
We~write $\intro*\mkimg\nu$ for the \kl{multiset}~$\nu(\Set')$.

\subsection{\kl{Words} and \kl{trees}}
\label{ssec-trees}

\AP Let~$\Alp$ be a finite set.
A~\intro{word} over~$\Alp$ (or~\reintro*\Aword{\Alp}) is a sequence $w=(w_i)_{0\leq i<k}$ of elements
of~$\Alp$, with $k\in\bbN\cup\{+\infty\}$.
The~\intro(W){length} (or~\reintro(W){size}) of~$w$, denoted with~$\intro*\Wsize w$, is~$k$.
We~write $\Alp*$ for the set of finite \kl{words} over~$\Alp$, and $\Alp~$
for the set of infinite \kl{words} over~$\Alp$.  We~write $\intro*\emptyw$ for the
\intro{empty word} (the~only \kl{word} of \kl(W){size}~$0$).

\AP For a \kl{word}~$w=(w_i)_{0\leq i<k}$ of length~$k\in \bbN\cup\{+\infty\}$, and an
integer~$j\in\bbN$ such that $0\leq j\leq k$,
the~\intro{prefix} of~$w$ of length~$j$ (also~referred to as its $j$-th prefix)
is the \kl{word}~$\reintro*\prefix wj=(w_i)_{0\leq i<j}$. For~$0\leq j<k+1$, the~$j$-th \intro{suffix}
of~$w$ is the \kl{word} $\reintro*\suffix wj=(w_{j+i})_{0\leq i<k-j}$. Given a
finite \kl{word}~$w$ and a (possibly infinite) \kl{word}~$w'$, their
concatenation~$w\cdot w'$ is the \kl{word}~$x$ whose $\length
w$-th prefix is~$w$ and whose $\length w$-th suffix is~$w'$.
We~identify \kl{words} of length~$1$ with their constituent letter, and
write $\intro*\first w$ for the first prefix of~$w$, and, in~case $w$ is
finite, $\intro*\last w$ for its $(\length w-1)$-th suffix.

\medskip
\AP Let~$\Dir$ be a finite set.
A~\intro{tree structure} over~$\Dir$ (or~\reintro*\Dtree{\Dir})
is a subset $\tree\subseteq \Dir*$ that is closed
under prefix.
In~particular, any non-empty \kl{tree} contains the \kl{empty
word}~$\intro*\troot$, which is called its \intro{root} (and~sometimes
denoted with $\reintro*\troot_\tree$ when we need to distinguish
between the roots of different \kl{trees}). The~elements of a \kl{tree} are
called nodes.  A~node~$m$ in~$\tree$ is a \intro{successor} of a node~$n$ if
$m=n\cdot d$ for some $d\in \Dir$. In~that~case, $n$~is the (unique)
predecessor of~$m$.
We~write $\intro*\succ(n)$ for the set of \kl{successors} of node~$n$.
Notice that
in a \Dtree{\Dir}, any node may have at~most~$\size{\Dir}$ \kl{successors};
this~integer~$\size{\Dir}$ is the \intro{arity} of the \kl{tree};
as~a special case, a~\intro{binary tree} is a \Dtree{\Dir} with $\size{\Dir}=2$. 
Notice that not all nodes have to have $\size\Dir$ \kl{successors} in a
tree of arity~$\size{\Dir}$. In~particular, any~\kl{tree} may contain
\intro{leaves}, which are node with no \kl{successors}.
A~\intro{branch} of a \kl{tree} is a (finite or infinite)
sequence~$b=(n_i)_{0\leq i<k}$ of nodes of the \kl{tree} such that
$n_0=\troot$, $n_{i+1}$~is a \kl{successor} of~$n_i$ for all $0\leq i<k-1$,
and if $k$~is finite, $n_{k-1}$~is a~\kl{leaf} (in~other terms, a~\kl{branch} must be \emph{maximal}).
The~value of~$k\in\bbN\cup\{+\infty\}$ is the \kl{length}
of~$b$, denoted with~$\Wsize b$.
In~the sequel, we~may identify a branch~$b=(n_i)_{0\leq i<k}$, which is a sequence
of nodes where each $n_i$~is a finite word~$(d_j)_{0\leq j<i}$, with the sequence
of directions~$(d_j)_{0\leq j<k}$.
A~\kl{tree}~is finite when it contains only finitely many nodes.
By~K\H{o}nig's lemma, since~$\Dir$ is finite, a~\kl{tree} is finite if,
and only~if, all~its \kl{branches} are finite.

\AP A~\intro*\SDtree\Alp\Dir
is a pair~$\Tree=\tuple{\tree,\lab}$
where $\tree$ is a \Dtree{\Dir} and $\lab\colon \tree \to \Alp$ labels
each node of~$\tree$ with a letter in~$\Alp$.
With~any \kl{branch}~$b=(n_i)_{0\leq i<k}$ of~$\tree$ in a \SDtree\Alp\Dir~$\Tree=\tuple{\tree,\lab}$, we~associate its \intro(B){word}~$\Bword{b}$
over~$\Alp$ as the \kl{word}~$(w_i)_{0\leq i<k}$ defined as $w_i=\lab(n_i)$
for all $0\leq i<k$.

\medskip
\AP Let $\AtP$ be a finite set of atomic propositions.  
A~\intro{Kripke structure} over~$\AtP$ is a tuple $\calK=\tuple{V,E,\ell}$ where $V$~is a
finite set of vertices, $E\subseteq V\times V$ is a set of edges
(requiring that for any~$v\in V$, there exists $v'\in V$ s.t. $(v,v')\in E$)
and $\ell\colon V\to 2^{\AtP}$ is a labelling function.

\AP
A \intro(KS){path} in a \kl{Kripke structure} is a finite or infinite
\kl{word}~$w$ over~$V$ such that $(w_i,w_{i+1})\in E$ for
all~$i<\Wsize w$. We~write $\Path^*_{\calK}$
for the
sets of finite
\kl(KS){paths} of~$\calK$.
Given a vertex~$v\in V$, the~\intro{computation tree} of~$\calK$
from~$v$ is the \SDtree{2^{\AtP}}{V}
$\calT_{\calK,v}=\tuple{T_{\calK,v},\hat\ell}$ with $T_{\calK,v}=\{w\in V^* \mid
v\cdot w\in \Path^*_\calK\}$ and $\hat\ell(v\cdot w)=\ell(\last{v\cdot w})$. 
Notice that two nodes~$w$ and~$w'$ of~$\calT_{\calK,v}$ for which $\last
w=\last{w'}$ give rise to the same subtrees. A~\kl{tree} is said  \intro{regular}
when it corresponds to the \kl{computation tree} of some finite
\kl{Kripke structure}.

\subsection{Automata over \kl{trees} of arbitrary arity}
\label{ssec-automata}

\AP
In this section, we introduce our automata running over \kl{trees} of
arbitrary \kl{arity}. The~core element of their transition functions are
\EUprs and \EUconstrs.

\AP Let~$\Set$ be a countable set.
An~\intro*\EUpr
over~$\Set$ is a pair $\EUpair(E;U)\in \bbN^\Set\times
2^\Set$, where $E$ is a \kl{multiset} over~$\Set$ and $U$~is a subset of~$\Set$.
A~\kl{multiset}~$\mu$ over~$\Set$ satisfies the
\EUpr~$\EUpair(E;U)$, denoted $\mu\intro*\models| \EUpair(E;U)$, whenever
$E\submultiset \mu$ and $\supp(\mu\msminus E)\subseteq U$.  We~write
$\intro*\EUset(\Set)=\bbN^\Set\times 2^\Set$ for the set of \EUprs
over~$\Set$.

\begin{example}
Consider a set~$\Set=\{q_1,q_2,q_3,q_4\}$.  The~\EUpr
$\EUpair(q_1\mapsto 3, q_2\mapsto 1; \{q_1,q_3\})$ characterises all
\kl{multisets} containing \emph{at~least} three occurrences
of~$q_1$, \emph{exactly} one occurrence of~$q_2$, an arbitrary number
of occurrences of~$q_3$, and no occurrences of~$q_4$.
\end{example}

\smallskip

\AP
For~a finite set~$\BV$ of boolean variables, we~write $\PBF\BV$ for the
set of \intro{positive boolean combinations} over~$\BV$:
\[
\reintro*\PBF\BV \ni \phi\coloncolonequals \top \mid \bot \mid v  \mid \phi\et\phi \mid \phi\ou\phi
\]
where $v$ ranges over~$\BV$.
The~set of \intro{disjunctions} over~$\BV$ is the subset of~$\PBF\BV$
defined as
\[
\reintro*\DBF\BV \ni \phi\coloncolonequals   \top \mid \bot \mid v \mid \phi\ou\phi
\]
where again $v$ ranges over~$\BV$.

\medskip

\AP
An~\intro*{\EUconstr} is a \kl{positive boolean formula} over~\EUprs.
We can now define our class of automata:
\begin{definition}
\AP  Let~$\Alp$ be a finite alphabet.
  An \intro{alternating} \intro{\EU tree automaton} (\intro*\AATA for~short)
  over~$\Alp$
  is a
  $4$-tuple $\Aut=\tuple{\State,\initstate,\Trans,\Accept}$ with
  \begin{itemize}
  \item $\State$ is a finite set of states, and $\initstate\in\State$ is
    the initial state;
  \item $\Trans\colon \State\times\Alp \to
    \PBF{\EUset(\State)}$ is the set of transitions;
  \item $\Accept \colon \State~\to \{0,1\}$ is an
    acceptance condition.
  \end{itemize}

\AP  An \AATA is \intro{non-alternating} (and is thus an \reintro{\EU
  tree automaton}, \intro*\nAATA for~short) if~$\Trans$~takes values in
  $\DBF{\EUset(\State)}$.
\end{definition}

The following relation will be the central relation for defining the
semantics of \AATA: it~will be used to lift the satisfaction relation of
\EUconstrs to execution trees. 
\begin{definition}\label{def-models+}
Let $\Set$ and~$\Set'$ be two countable sets.
Let $\EUpair(E;U)$ be an \EUpr over~$\Set$,
and~$\nu$~be a \kl{marking} of~$\Set'$ by~$\Set$. Then $\nu$
satisfies~$\EUpair(E;U)$, denoted $\nu\intro*\models+ \EUpair(E;U)$, if there
exists a \kl{unitary marking}~$\nu'\submarking \nu$, such that
$\mkimg{\nu'}\models| \EUpair(E;U)$.

This definition extends to \EUconstrs inductively as follows:
\begin{itemize}
\item for any marking~$\nu$, $\nu\models+\top$ and $\nu\not\models+\bot$;
\item $\nu\models+ \phi_1 \ou \phi_2$ if, and only~if, $\nu\models+ \phi_1$ or $\nu\models+ \phi_2$;
\item $\nu\models+ \phi_1 \et \phi_2$ if, and only~if, $\nu\models+ \phi_1$ and $\nu\models+ \phi_2$.
\end{itemize}
\end{definition}
Notice that $\nu\models+ \phi$ is \emph{not} equivalent to having an  
\kl{unitary marking}~$\nu'  \submarking \nu$ satisfy~$\phi$, since
different \kl{submarkings}~$\nu'$ may be needed for different \EUprs.
However, the~equivalence holds if $\phi$~is a disjunction of \EUprs,
since in that case a single \EUpr has to be fulfilled.

\AP
We~can now define the notion of \kl{execution tree} of an~\AATA:
\begin{definition}\label{def-semEU}
\AP Let $\Aut=\tuple{\State,\initstate,\Trans,\Accept}$ be an \AATA
over~$\Alp$ and $\Tree=\tuple{\tree,\lab}$ be a
\SDtree{\Alp}{\Dir},
for some finite set~$\Dir$.
An~\intro{execution tree} of~$\Aut$ over~$\Tree$ is a
\SDtree{(\tree\times\State)}{(\Dir\times\State)}~$\ExTree=\tuple{\extree,\exlab}$ such that
\begin{itemize}
\item the \kl{root}~$\troot_{\extree}$ of~$\extree$ is labelled
  with~$\troot_{\tree}$ and~$\initstate$ (formally,
  $\exlab(\troot_{\extree})=(\troot_{\tree},\initstate)$);

\item any non-\kl{root} node~$n_{\extree}=(\dir_i,\state_i)_{0\leq i<\size{n_{\extree}}}$
  of~$\extree$ is labelled with $\exlab(n_{\extree}) =
  ((\dir_i)_{0\leq i<\size{n_{\extree}}}, \state_{\size{n_{\extree}}-1})$;

\itemAP for any node $n_{\extree}$ of the form $(\dir_i,\state_i)_{0\leq
  i<\size{n_{\extree}}}$ of~$\extree$
  with $\exlab(n_{\extree})=(m_{\tree},\state)$,
  letting~$\mknu_{n_\extree}$ be the \kl{marking} of~$\succ(m_\tree)$ by~$\State$
  such that $\mknu_{n_\extree}(m_\tree\cdot d)=\{\state'\in\State\mid 
   n_{\extree}\cdot (d,\state') \in\succ(n_\extree)\}$,
   we~have
   $\mknu_{n_\extree}\models+ \Trans(\state, \lab(m_{\tree}))$.
  We~name this \kl{marking}~$\intro*\mknu_{n_\extree}$ the \kl{marking} of~$\succ(m_\tree)$
  \intro{induced} by~$\ExTree$.

\end{itemize}

\AP The~tree~$\Tree$ is \intro{accepted} by~$\Aut$
  if there exists an \kl{execution tree}~$\ExTree$ of~$\Aut$
  over~$\Tree$ such that all infinite \kl{branches}
  are \intro(B){accepting}, i.e., for any
  infinite \kl{branch}~$b=(b_i)_{i\geq 0}$ in~$\ExTree$, it~holds
  ${\Accept((\proj_2(\exlab(b_i)))_{i\geq 0})=1}$.
  Such~an \kl{execution tree} is said to be \intro(ET){accepting}.
  Notice that an \kl{accepting} \kl{execution tree} may contain finite
  branches, in case those branches end in a node~$n_{\extree}$ with
  $\exlab(n_{\extree})=(m_{\tree},\state)$ such that
  $\Trans(\state, \lab(m_{\tree}))=\top$.

  The~\intro{language} of~$\Aut$, denoted by $\Lang(\Aut)$, is the set of all \kl{trees} \kl{accepted} by~$\Aut$.
The~tree~$\Tree$ is \intro{rejected} if it is not \kl{accepted}, i.e.,
if there are no \kl(ET){accepting} \kl{execution trees} of~$\Aut$ over~$\Tree$.
\end{definition}

\AP
Notice that if a \kl{marking}~$\nu$ satisfies some \EUconstr, then any
\kl{marking}~$\nu'$ containing~$\nu$ also~does. Similarly, any~\kl{execution
tree} can be extended with extra subtrees, provided that all their
\kl{branches} are \kl(B){accepting}. An~\kl{execution tree}
is said \intro{minimal} if it does not contain
dispensable subtrees.  We~may always consider that
the \kl{execution trees} we consider are \kl{minimal}.

\AP
Let us illustrate \kl{execution trees} with an example (more~examples are presented in Section~\ref{ssec-ex-aut}):
:
\begin{example}\label{ex-EUsat}
Consider a node $m$ of some input \kl{tree}~$\Tree$, with three \kl{successors}
${m\cdot d_1}$, ${m\cdot d_2}$ and ${m\cdot d_3}$ (as depicted on
Fig.~\ref{fig-sem-aut}); assume that this node~$m$ is labelled with
some letter~$\sigma$. Consider an \AATA~$\Aut$ visiting node~$m$ in
state~$q$, giving rise to a node~$n$ in an \kl{execution tree}~$\ExTree$
with $\exlab(n)=(m,q)$.  The~\kl{successors} of~$n$ in the \kl{execution tree}
give rise to the \kl{marking}~$\nu_n$ such that $\nu_n(m\cdot
d_1)=\{\state_1,\state_3\}$, $\nu_n(m\cdot d_2)=\{\state_2,\state_4\}$,
and $\nu_n(m\cdot d_3)=\{\state_1,\state_4\}$, as depicted to the left
of Fig.~\ref{fig-sem-aut}.

Assume that $\delta(\state,\sigma)$ is satisfied by the following set~$W$
of \EUprs:
\[
W=\Bigl\{
\langle\underbrace{q_1\mapsto 2; \{q_2\}}_{\EUpair(E_1;U_1)}\rangle,\
\langle\underbrace{q_1\mapsto 1, q_2\mapsto 1, q_3\mapsto 1;
   \{q_4\}}_{\EUpair(E_2;U_2)}\rangle,\
\langle\underbrace{q_3\mapsto 1; \{q_4\}}_{\EUpair(E_3;U_3)}\rangle
\Bigr\}.
\]
Figure~\ref{fig-sem-aut}
displays a possible set of \kl{successors} $\succ(n)$ of~$n$ in the
\kl{execution tree}~$\ExTree$.
Using the following three \kl{submarkings}~$\nu_i$ of~$\nu_n$,
we~are able to fulfill all three \EUprs of~$W$:
\begin{xalignat*}3
\nu_1\colon &n\cdot d_1 \mapsto q_1 &
  \nu_2\colon &n\cdot d_1 \mapsto q_3 &
  \nu_3\colon &n\cdot d_1 \mapsto q_3 \\
            &n\cdot d_2 \mapsto q_2 &
              &n\cdot d_2 \mapsto q_2 &
              &n\cdot d_2 \mapsto q_4 \\
            &n\cdot d_3 \mapsto q_1 &
              &n\cdot d_3 \mapsto q_1 &
              &n\cdot d_3 \mapsto q_4 \tagqex
\end{xalignat*}
\noqex
\end{example}

\begin{figure}[t]
\centering
\begin{tikzpicture}
  \begin{scope}
    \path (0,.6) node {input \kl{tree}};
    \draw (0,0) node (mu) {$m$}; %
    \draw (-130:2cm) node (t1q1) { $m\cdot d_1$}
      node[below=1mm,opacity=.5] {$\{q_1,q_3\}$};
    \draw (-90:1.8cm) node (t2q2) {$m\cdot d_2$}
      node[below=1mm,opacity=.5] {$\{q_2,q_4\}$};
    \draw (-50:2cm) node (t3q1) {$m\cdot d_3$}
      node[below=1mm,opacity=.5] {$\{q_1,q_4\}$};
 \draw[-latex'] (mu) -- (t1q1);
 \draw[-latex'] (mu) -- (t2q2);
 \draw[-latex'] (mu) -- (t3q1);
 \path (mu) node[below=1mm,fill=white] {\color{black!50!white}$(q)$};
  \end{scope}
  \begin{scope}[xshift=6.4cm]
    \path (0,.6) node {\kl{execution tree}};

    \draw (0,0) node (nu) {$n$} node[right=8mm] {$\ell(n)=(m,q)$};
    \draw (-160:3cm) node (t1q1) { $n\cdot(d_1,q_1)$};
    \draw (-140:2.4cm) node (t1q3) {$n\cdot(d_1,q_3)$};
    \draw (-110:2.2cm) node (t2q2) {$n\cdot(d_2,q_2)$\qquad};
    \draw (-70:2.2cm) node (t2q4) {\qquad$n\cdot(d_2,q_4)$};
    \draw (-40:2.4cm) node (t3q1) {$n\cdot(d_3,q_1)$};
    \draw (-20:3cm) node (t3q4) {$n\cdot(d_3,q_4)$};
    \draw[-latex'] (nu) -- (t1q1);
    \draw[-latex'] (nu) -- (t2q2);
    \draw[-latex'] (nu) -- (t3q1);
    \draw[-latex'] (nu) -- (t1q3);
    \draw[-latex'] (nu) -- (t2q4);
    \draw[-latex'] (nu) -- (t3q4);
  \end{scope}
\end{tikzpicture}
\caption{Example of a transition of the automaton when exploring
  node~$m$ of an input \kl{tree}~$\Tree$ in state~$q$: node~$m$ has three
  \kl{successors} $m\cdot d_1$, $m\cdot d_2$ and $m\cdot d_3$ in~$\Tree$;
  if~the automaton explores node~$m\cdot d_1$ in states~$q_1$
  and~$q_3$, node~$m\cdot d_2$ in states~$q_2$ and~$q_4$, and
  node~$m\cdot d_3$ in states~$q_1$ and~$q_4$ (as in the partial
  \kl{execution tree} on the right), we~get the \kl{marking} of the \kl{successors}
  of node~$m$ as given on the left, from which we can extract \kl{submarkings} 
  satisfying the \EUconstr~$W$ of Example~\ref{ex-EUsat}.}\label{fig-sem-aut}
\end{figure}

Notice that $\Dir$ is not constrained by the definition of~\AATAs, so
that \AATAs may accept \kl{trees} of arbitrary (finite) \kl{arity}.  However,
the~\kl{multisets} in the ``existential part'' of \EUprs can be
used to impose a lower bound on the number of \kl{successors} for
the \EUpr to be satisfied, and an upper bound can be imposed
by letting the universal part be~empty. We~develop such examples in
Section~\ref{ssec-ex-aut} below.

\begin{remark}\label{rk-nonalt}
We~do not define a notion of being \emph{deterministic} for~\nAATA:
even if the transition function~$\delta$ returns a
single \EUpr for each pair~$(q,\sigma)$, there may be many
different valid \kl{execution trees}, since there may be different ways of
satisfying a single \EUpr.

Notice that for (\kl{non-alternating}) \nAATAs, the definition of \kl{execution
trees} can be simplified:
in~\kl{non-alternating} automata,
$\delta(q,\sigma)$~is a disjunction of \EUprs, and
a single \kl{unitary marking}, hence a single
state for each \kl{successor node}, is~sufficient to fulfill~it.

It~follows that, for each input
\kl{tree}~$\Tree=\tuple{\tree,\lab}$ accepted by an~\nAATA, there~is an
\kl(ET){accepting} \kl{execution tree} of the
form~$\ExTree=\tuple{\tree', \exlab}$
where $\tree'$ is isomorphic~to a subset of~$\tree$ (some~subtrees of~$\Tree$
may be pruned when the transition function returns~$\top$).
By~identifying each node of~$\tree'$ with the corresponding node of~$\tree$,
we~may assume that any~node~$n$ of~$\tree'$ 
has $\exlab(n)=(n,\state)$ for some~$\state\in\State$.
\end{remark}

\AP
In the sequel, we mainly consider \intro{parity acceptance}:
$\Omega$~can then be defined through a mapping
$\prio\colon \State \to \bbN$.
The~integer~$\prio(\state)$ is called the \intro{priority} of
state~$\state$.  Each~\kl{branch}~$b=(b_i)_{0\leq i<\size b}$ in~$\ExTree$
thus gives rise through~$\prio$ to a sequence of priorities
$(\prio(\proj_2(\exlab(b_i))))_{0\leq i<\size b}$. \kl{Parity acceptance}
for such a sequence is defined as follows:
for~an~infinite \kl{branch}~$b$,  let~$\prio_{\min}(b)$ be
the least \kl{priority} appearing infinitely many times along~$b$;
then the~infinite \kl{branch}~$b$ is
\kl(B){accepting} if, and only~if, 
$\prio_{\min}(b)$~is even.
Notice that we impose no conditions on finite \kl{branches} (\ie, any finite \kl{branch} is \kl(B){accepting}).
\AATA (resp.~\nAATA) equipped with a \kl{parity acceptance} condition are called \intro*\AAPTA
(resp.~\intro*\nAAPTA).

\subsection{Size of an \AAPTA}

\AP
The \intro{size of an \AAPTA}, denoted with~$\reintro*\Asize{\Aut}$, is a $5$-tuple
$\tuple{\size\State, \sizeB\Trans, \sizeE\Trans,\sizeU\Trans,\priomax}$, where
$\intro*\sizeB\Trans$ is the maximum
size of the boolean formulas in~$\Trans$, 
$\intro*\sizeE\Trans$ (resp.\ $\intro*\sizeU\Trans$) is the size of the largest existential part $E$ (resp.\ universal part $U$) in some 
\EUpr in~$\Trans$,
and $\intro*\priomax= \size{\{\prio(\state) \mid \state\in\State\}}$ is the number of priorities used in the automaton.
Note that contrary to classical, \faTAutomata, we explicitly consider
the size of the transition function~$\Trans$ in the \kl(A){size} of an \AAPTA;
this is motivated by the fact that \EUconstrs may succinctly
encode very complex transitions, regardless of the size of~$\State$.

In~the sequel, we say that the \kl(A){size} of an \AAPTA is at most~$(s_Q,s_B,\allowbreak s_E,\allowbreak s_U,s_\omega)$ when $\size\State\leq s_Q$, $\sizeB\Trans\leq s_B$, $\sizeE\Trans\leq s_E$, $\sizeU\Trans\leq s_U$ and $\size\prio\leq s_\omega$.
The~fact that we use five different parameters in the size of \AAPTAs
will allow us to have more precise bounds on the complexities of
our operations for manipulating them.

\subsection{Bounding the size of the universal part of \EU-pairs}
\label{remUsingleton}

In this section, we~explain how we can modify our automata so that the
universal part of their \EUprs have size at most~$1$. While this is
not crucial for the subsequent developments, it~will be interesting to
notice that this property (of~having small universal parts) will be
preserved by all the constructions we develop for manipulating automata.

Let $\Aut = \tuple{\State,\initstate,\Trans,\prio}$ be an \AAPTA. 
From~$\Aut$, we~build an automaton $\Aut'$  in such a way that every \EUpr
it~involves
uses only a singleton or the empty set as the universal parts of its \EUconstrs
(\ie, with~${\sizeU\Trans\leq 1}$). The~construction is based on the following two~ideas. First, any~set~$U$ (with $|U|>1$) occurring in some \EUpr in~$\Trans$  is replaced by a~(new) state~$\state_U$, whose transition function is a disjunction of the transitions for every state in~$U$. In~the following, we~use~$\widetilde\State$
to denote the  set containing~$\State$ and all states~$\state_U$ for pairs~$\EUpair(E;U)$  in~$\Trans$. We~write~$\widetilde\Trans$ for the transition function obtained from~$\Trans$ by replacing each \EUpr~$\EUpair(E;U)$ with the corresponding~$\EUpair(E;\{\state_U\})$.

Replacing the states of~$U$ with a single state~$\state_U$ may impact the satisfaction of the parity conditions along infinite branches of the execution tree, since the~states in~$U$ may have different priorities. The~second part of the transformation consists in 
keeping track of the states actually used to satisfy the $U$~part:
for~this, we~use pairs $(\state,\state') \in \State\times \widetilde\State$ as
states of~$\Aut'$: a~state~$(\state,\state')$ encodes the fact that
the current state is~$\state'\in\widetilde\State$, and that the previous state
was~$\state$. If~$\state'\in\State$,
the~transition function from~$(\state,\state')$ is then
obtained from~$\widetilde\Trans(\state',\sigma)$ by replacing
every state~$r$ with~$(\state',r)$;
otherwise, if~$\state'$ is a state~$\state_U\in\widetilde\State\setminus\State$,
the~transition function from~$(\state,\state_U)$ is the disjunction
over all states~$\state''\in U$ of $\widetilde\Trans(\state'',\sigma)$
where each state~$r$ is replaced with~$(\state'',r)$.

Formally, the construction of~$\Aut'$ is as follows. 
For~any state~$\state'\in\State$,
we~define the mapping $\phi_{\state'}\colon \State \to \State^2$ as
$\phi_{\state'}(\state)=(\state',\state)$, and~extend~it to
(multi-)sets of states, \EUprs and \EUconstrs in the natural~way.
Now we can define the \AAPTA  $\Aut'=\tuple{\State\times\widetilde\State, (\initstate,\initstate), \Trans', \prio'}$ with:
\begin{itemize}
\item $\Trans'((\state,\state'),\sigma) = \begin{cases} \phi_{\state'}\bigl(\widetilde\Trans(\state',\sigma)\bigr) & \text{if } \state'\in\State \\[1.2ex]
 {\displaystyle \OU_{r\in U}  \phi_r \bigl(\widetilde\Trans(r,\alp)\bigr) } & \text{if }  \state'=\state_U \end{cases}$
\item $\prio'((\state,\state')) = \prio(\state)$.
\end{itemize}

The correctness of the construction is based on the fact that there is
direct correspondance between the execution trees of~$\Aut$
and~$\Aut'$: both have the same underlying tree structure, and the
only difference is in the labelling (by~$(t,\State)$ for~$\Aut$ and
by~$(t,\State\times\widetilde\State)$ for~$\Aut'$).  Any infinite branch
has the same least infinitely-repeated priority
in two both execution trees.

The~\kl(A){size} of the resulting automaton~$\Aut'$ is bounded by
$\tuple{{{\size\State^2(1+\sizeB\Trans\cdot\size\Sigma)}},\allowbreak
\sizeB\Trans\cdot\sizeU\Trans,\sizeE\Trans,1,\size\prio}$.
Moreover, as~we will see in the sequel, all the constructions we
develop for manipulating \AAPTAs preserve this property of having only
singleton or empty sets as the universal part of~\EUprs.
Finally note that if~$\Aut$~is non-alternating, so~is~$\Aut'$.

\subsection{Examples}
\label{ssec-ex-aut}

\AP
We~illustrate our definitions with a few examples.  Before presenting
examples of \kl{\EU-automata}, we~begin with examples~of~(special
cases~of) \EUconstrs.

\medskip

First, the \EUpr~$\EUpair(\emptyset;\emptyset)$ characterises
\kl{leaves} of the input \kl{tree}: indeed, in~order to have
$\mknu_{n_\extree} \models+\EUpair(\emptyset;\emptyset)$ (using~the
notations of Def.~\ref{def-semEU}), the~\kl{marking}~$\mknu_{n_\extree}$ must
contain a \kl{unitary} \kl{submarking} with empty image, hence its domain must be empty.

\kl{Positive boolean formulas} also allow the special cases of~$\top$
and~$\bot$. Formula~$\top$  does not impose any constraints
on the node~$n_\extree$ of the \kl{execution tree} where it is evaluated: this node can be a \kl{leaf}
even if the corresponding node~$m_\tree$ in the input \kl{tree} is not a
\kl{leaf}. The~resulting (finite) \kl{branch} in the \kl{execution tree} is \kl(B){accepting}.
On~the other hand, no~nodes of any \kl{execution tree} can satisfy~$\bot$,
and such a transition can only lead to rejection.

Finally, \kl{\EU-automata} may include a special state~$\state_\top$ for which
$\Trans(\state_\top,\alp)=\top$
for any~$\alp\in\Alp$:
any~subtree visited in this state is accepted.
Similarly, we~could have a state~$\state_\bot$ with
$\Trans(\state_\bot,\alp)=\bot$ for any~$\alp\in\Alp$,
which would reject any tree where it appears.

\medskip
We~now present some examples of \kl{\EU-automata}.

\renewcommand\dbltopfraction{.9}
\renewcommand\topfraction{.9}

\newcommand{\qbinf}{\state_{b}^{\infty}}
\newcommand{\rbinf}{r_{b}^{\infty}}
\newcommand{\qaainf}{\state_{a}^{2i}}
\newcommand{\qainf}{\state_{a}^{1i}}
\newcommand{\rainf}{r_{a}^{i}}
\newcommand{\qafin}{\state_{a}^f}
\newcommand{\rafin}{r_{a}^f}

\begin{example}\label{ex-aapta}
To illustrate the use of \AAPTA, we display an example of an automaton
accepting all trees satisfying the following two conditions: exactly
two infinite branches contain infinitely many occurrences of~$a$, and at least one branch is fully labelled with~$b$.

Let $\Sigma=\{a,b,c\}$.  We~let
$\Aut=\tuple{\State,\initstate,\Trans,\prio}$, where the set~$\State$
of states is
$\{\initstate,\qaainf,\qainf,\allowbreak\rainf,\allowbreak\qbinf,
\allowbreak\qafin,\rafin,\state_\top\}$:
the automaton uses state~$\qaainf$ to visit a prefix of a branch having
exactly two infinite sub-branches where $a$~occurs infinitely many times;
it~uses states~$\qainf$ and~$\rainf$ to visit subtrees in
which exactly one branch has infinitely many occurrences of~$a$:
state~$\rainf$ is used when letter~$a$ is read, and the acceptance
condition requires that it be visited infinitely many times along some branch;
similarly, states~$\qafin$ and~$\rafin$ are used to explore
subtrees in which all branches have finitely many occurrences of~$a$.
Finally, we~use~$\qbinf$ to explore branches fully labelled with~$b$.
Following this intuition, the transition function is defined as:
\begin{xalignat*}1
  \Trans(\initstate,\sigma) &= \begin{cases}  \bigl[\EUpair(\qaainf \mapsto 1;\{\qafin\}) \ou \EUpair( \qainf\mapsto 2;\{\qafin\})\bigr] \et \EUpair(\qbinf\mapsto 1;\{\state_\top\}) & \text{if }\sigma=b \\
\bot & \text{otherwise} \end{cases}\\[2mm]
  \Trans(\qaainf,\sigma) &= \EUpair(\qaainf \mapsto 1;\{\qafin\}) \ou \EUpair( \qainf\mapsto 2;\{\qafin\}) \quad \text{ for any } \sigma\in\Sigma\\[2mm]
  \Trans(\qainf,\sigma) &= \Trans(\rainf,\sigma) = \begin{cases} \EUpair(\rainf \mapsto 1;\{\qafin\}) & \text{ if } \sigma=a\\
  \EUpair(\qainf \mapsto 1;\{\qafin\}) & \text{ otherwise} 
\end{cases} \\[2mm]
  \Trans(\qafin,\sigma) &= \Trans(\rafin,\sigma) = \begin{cases} \EUpair(\emptyset;\{\rafin\}) & \text{if } \sigma=a\\
  \EUpair(\emptyset;\{\qafin\}) & \text{ otherwise} \end{cases} \\[2mm]
 \Trans(\qbinf,\sigma) &= \begin{cases} \EUpair(\qbinf \mapsto 1;\{\state_\top\}) & \text{ if } \sigma=b\\
  \bot & \text{ otherwise} \end{cases} 
\end{xalignat*}
Finally, $\Trans(\state_\top,\sigma)=\top$ for any~$\sigma\in\Sigma$,
so that any subtree is accepted when explored in~$\state_\top$.
Notice that $\Trans(\initstate,\sigma)$ contains a conjunction, so
that the automaton is \kl{alternating}.  Notice also how the disjunctive \EUconstr
in~$\Trans(\initstate,\sigma)$ and~$\Trans(\qaainf,\sigma)$
will either look for a single
successor from which two \kl{branches} will have infinitely many
occurrences of~$a$, or for two \kl{successors} from each of which there
will be such a  \kl{branch}. All other
\kl{branches} will be checked to have finitely many occurrences of~$a$ by
exploring them in state~$\qafin$.
The~\kl{priority function} is defined so as to check that $a$ occurs
infinitely many times along branches where it has to:
\begin{xalignat*}3
  \prio(\rainf)=\prio(\qbinf)&=0 &
  \prio(\qainf)=\prio(\qaainf)=\prio(\rafin) &= 1 &
  \prio(\qafin)&=2
\end{xalignat*}
It can be checked that the size of $\Aut$ is $\tuple{8,5,2,1,3}$.

\begin{figure}[t]
  \centering
  \def\noeud#1#2#3#4#5{\path(#1,#2) node[fill=white,draw,circle,inner sep=0pt,minimum size=5mm] (#5a) {\hbox to 0pt{\hss $#3$\hss}}
    node[below=1.5mm] (#5b) {\vphantom{fg}$#4$};}
  \def\noeudshift#1#2#3#4#5{\path(#1,#2) node[fill=white,draw,circle,inner sep=0pt,minimum size=5mm] (#5a) {\vphantom{fg}\hbox to 0pt{\hss $#3$\hss}}
    node[below=1.5mm,xshift=2mm] (#5b) {\vphantom{fg}$#4$};}
  \def\noeudoval#1#2#3#4#5{\path(#1,#2) node[fill=white,draw,rounded corners=2mm,inner sep=0pt,minimum height=5mm] (#5a) {\vphantom{fg}\hbox{\hskip-0mm $#3$\hskip-0mm}}
    node[below=1.5mm] (#5b) {\vphantom{fg}$#4$};}
  \def\noeudovalshift#1#2#3#4#5{\path(#1,#2) node[fill=white,draw,rounded corners=2mm,inner sep=0pt,minimum height=5mm] (#5a) {\vphantom{fg}\hbox{\hskip-0mm $#3$\hskip-0mm}}
    node[below=1.5mm,xshift=2mm] (#5b) {\vphantom{fg}$#4$};}

  \begin{tikzpicture}[xscale=1.4,yscale=1]

    \draw[line width=2mm,black!20] (0,0) -- (-2,-1) -- (-3,-2) -- (-3.5,-3) -- (-3.5,-4) -- (-3.5,-4.7);
    \draw[line width=2mm,black!20] (0,0) -- (0,-1) -- (1,-2) -- (1.75,-3) -- (1,-4) -- (1,-4.7);
    \draw[line width=2mm,black!50] (0,0) -- (0,-1) -- (-1,-2) -- (-.25,-3) -- (-.25,-4) -- (-.25,-4.7);

    \noeudshift 00{\epsilon}{b}{de}
    \noeud {-2}{-1}{d_1}{a}{de1}
    \noeud {0}{-1}{d_2}{b}{de2}
    \noeud {2}{-1}{d_3}{c}{de3}
    \noeudoval {-3}{-2}{d_1 d_1}{b}{de11}
    \noeudoval {-1}{-2}{d_2 d_1}{b}{de21}
    \noeudoval {1}{-2}{d_2 d_2}{a}{de22}
    \noeudoval {3}{-2}{d_3 d_1}{b}{de31}
    \noeudovalshift {-3.5}{-3}{d_1 d_1 d_1}{a}{de111}
    \noeudovalshift {-1.75}{-3}{d_2 d_1 d_1}{c}{de211}
    \noeudovalshift {-0.25}{-3}{d_2 d_1 d_2}{b}{de212}
    \noeudoval {1.75}{-3}{d_2 d_2 d_1}{c}{de221}
    \noeudoval {3.5}{-3}{d_3 d_1 d_1}{a}{de311}
    \noeudovalshift {-3.5}{-4}{d_1 d_1 d_1 d_1}{a}{de1111}
    \noeudovalshift {-1.75}{-4}{d_2 d_1 d_1 d_1}{a}{de2111}
    \noeudovalshift {-0.25}{-4}{d_2 d_1 d_2 d_1}{b}{de2121}
    \noeudovalshift {1}{-4}{d_2 d_2 d_1 d_1}{a}{de2211}
    \noeudovalshift {2.5}{-4}{d_2 d_2 d_1 d_2}{c}{de2212}
    \noeudovalshift {3.75}{-4}{d_3 d_1 d_1 d_1}{c}{de3111}

    \foreach \i/\d in {e/1,e/2,e/3,e1/1,e2/1,e2/2,e3/1,e11/1,e21/1,e21/2,e22/1,e111/1,e211/1,e212/1,e221/1,e221/2,e31/1,e311/1}{%
      \draw (d\i a) -- (d\i\d a);}
    \foreach \i in {de3111,de1111,de2111,de2121,de2211,de2212}{%
      \draw[dashed] (\i a) -- +(0,-.7);}
  \end{tikzpicture}

  \caption{An input tree with exactly two branches (light grey) having
    infinitely many occurrences of~$a$ and at least one branch (dark grey)
    fully labelled with~$b$.}\label{fig-input}
\end{figure}

\begin{figure}[t]
  \centering
  \def\noeud#1#2#3#4#5#6{\path(#1,#2) node[fill=white,draw,circle,inner sep=0pt,minimum size=5mm] (#5a) {\hbox to 0pt{\hss $#3$\hss}}
    node[below=1.5mm] (#5b) {\vphantom{fg}$#4$}
    node[above=1.5mm] (#5q) {\vphantom{fg}$#6$}
    ;}
  \def\noeudshift#1#2#3#4#5#6{\path(#1,#2) node[fill=white,draw,circle,inner sep=0pt,minimum size=5mm] (#5a) {\vphantom{fg}\hbox to 0pt{\hss $#3$\hss}}
    node[below=1.5mm,xshift=2mm] (#5b) {\vphantom{fg}$#4$}
    node[above=1.5mm,xshift=2mm] (#5q) {\vphantom{fg}$#6$}
    ;}
  \def\noeudoval#1#2#3#4#5#6{\path(#1,#2) node[fill=white,draw,rounded corners=2mm,inner sep=0pt,minimum height=5mm] (#5a) {\vphantom{fg}\hbox{\hskip-0mm $#3$\hskip-0mm}}
    node[below=1.5mm] (#5b) {\vphantom{fg}$#4$}
    node[above=1.5mm] (#5q) {\vphantom{fg}$#6$}
    ;}
  \def\noeudovalshift#1#2#3#4#5#6{\path(#1,#2) node[fill=white,draw,rounded corners=2mm,inner sep=0pt,minimum height=5mm] (#5a) {\vphantom{fg}\hbox{\hskip-0mm $#3$\hskip-0mm}}
    node[below=1.5mm,xshift=2mm] (#5b) {\vphantom{fg}$#4$}
    node[above=1.5mm,xshift=2mm] (#5q) {\vphantom{fg}$#6$}
    ;}
  \def\noeudovalshiftb#1#2#3#4#5#6{\path(#1,#2) node[fill=white,draw,rounded corners=2mm,inner sep=0pt,minimum height=5mm] (#5a) {\vphantom{fg}\hbox{\hskip-0mm $#3$\hskip-0mm}}
    node[below=1.5mm,xshift=2mm] (#5b) {\vphantom{fg}$#4$}
    node[above=1.5mm] (#5q) {\vphantom{fg}$#6$}
    ;}
  \def\noeudovalshiftq#1#2#3#4#5#6{\path(#1,#2) node[fill=white,draw,rounded corners=2mm,inner sep=0pt,minimum height=5mm] (#5a) {\vphantom{fg}\hbox{\hskip-0mm $#3$\hskip-0mm}}
    node[below=1.5mm] (#5b) {\vphantom{fg}$#4$}
    node[above=1.5mm,xshift=2mm] (#5q) {\vphantom{fg}$#6$}
    ;}
  \begin{tikzpicture}[xscale=1.6,yscale=1.3]

    \noeud 0{0}{\epsilon}{b}{de}{\initstate}
    \noeud {-3.5}{-1}{d_1}{a}{ae1}{\qainf}
    \noeud {-2}{-1}{d_2}{b}{ae2}{\qainf}
    \noeud {-.5}{-1}{d_3}{c}{ae3}{\qafin}
    \noeud {1.5}{-1}{d_1}{a}{be1}{\state_\top}
    \noeud {2.5}{-1}{d_2}{b}{be2}{\qbinf}
    \noeud {3.5}{-1}{d_3}{c}{be3}{\state_\top}
    \noeudoval {2}{-2}{d_2 d_1}{b}{be21}{\qbinf}
    \noeudoval {3}{-2}{d_2 d_2}{a}{be22}{\state_\top}
    \noeudoval {-4}{-2}{d_1 d_1}{b}{ae11}{\qainf}
    \noeudoval {-3}{-2}{d_2 d_1}{b}{ae21}{\qafin}
    \noeudovalshiftb {-1}{-2}{d_2 d_2}{a}{ae22}{\qainf}
    \noeudoval {0}{-2}{d_3 d_1}{b}{ae31}{\qafin}
    \noeudoval {1.5}{-3}{d_2 d_1 d_1}{c}{be211}{\state_\top}
    \noeudovalshift {2.5}{-3}{d_2 d_1 d_2}{b}{be212}{\qbinf}
    \noeudovalshiftb {-4.5}{-3}{d_1 d_1 d_1}{a}{ae111}{\qainf}
    \noeudovalshiftb {-3.5}{-3}{d_2 d_1 d_1}{c}{ae211}{\qafin}
    \noeudovalshiftb {-2.5}{-3}{d_2 d_1 d_2}{b}{ae212}{\qafin}
    \noeudovalshiftq {-1}{-3}{d_2 d_2 d_1}{c}{ae221}{\qainf}
    \noeudovalshiftb {.5}{-3}{d_3 d_1 d_1}{a}{ae311}{\qafin}
    \noeudovalshift {2.5}{-4}{d_2 d_1 d_2 d_1}{b}{be2121}{\qbinf}
    \noeudovalshift {-4.5}{-4}{d_1 d_1 d_1 d_1}{a}{ae1111}{\qainf}
    \noeudovalshift {-3.5}{-4}{d_2 d_1 d_1 d_1}{a}{ae2111}{\qafin}
    \noeudovalshift {-2.5}{-4}{d_2 d_1 d_2 d_1}{b}{ae2121}{\qafin}
    \noeudovalshiftb {-1.5}{-4}{d_2 d_2 d_1 d_1}{a}{ae2211}{\qainf}
    \noeudovalshiftb {-.5}{-4}{d_2 d_2 d_1 d_2}{c}{ae2212}{\qafin}
    \noeudovalshift {.5}{-4}{d_3 d_1 d_1 d_1}{c}{ae3111}{\qafin}

    \foreach \i/\d in {e/1,e/2,e/3}{%
      \draw (d\i a) -- (b\i\d a);
      \draw (d\i a) -- (a\i\d a);
    }
    \foreach \i/\d in {ae1/1,ae2/1,ae2/2,ae3/1,be2/1,be2/2,be21/1,be21/2,ae11/1,ae21/1,ae21/2,ae22/1,ae31/1,be212/1,ae111/1,ae211/1,ae212/1,ae221/1,ae221/2,ae311/1}{%
      \draw (\i a) -- (\i\d a);
    }

    \foreach \i in {ae1111,ae2111,ae2121,ae2211,ae2212,ae3111,be2121}{%
      \draw[dashed] (\i a) -- +(0,-.5);}
  \end{tikzpicture}
  \caption{An execution tree for our automaton~$\Aut$ of
    Example~\ref{ex-aapta} on the input tree of Fig.~\ref{fig-input}
    (notice that, for the sake of readability, we keep the names of
    the nodes as in the input tree, using words over~$\protect\Dir$ instead of
    words over $\protect\Dir\times\protect\State$).  Because of the conjunction
    in~$\Trans(\initstate,b)$, the~automaton explores the input tree
    twice, so~that the execution tree contains two copies of the input
    tree: the~subtree to the left of the root corresponds to the part
    looking for two branches with infinitely many occurrences of~$a$,
    while the subtree to the right looks for a branch fully labelled
    with~$b$.}\label{fig-exec}
\end{figure}

\smallskip

Figures~\ref{fig-input} and~\ref{fig-exec} display (part~of) an input
tree and a corresponding \kl{execution tree} for the automaton built above
The \kl{root} of the \kl{execution tree} is labelled with
$(\varepsilon,\initstate)$ and the \kl{marking}~$\nu$ \kl{induced} by the
\kl{execution tree} for the successors of the root of the input tree
is $\{d_1 \mapsto \{\state_\top,\qainf\}, d_2 \mapsto
\{\qainf,\qbinf\}, d_3 \mapsto \{\state_\top,\qafin\}\}$, which satisfies
$\Trans(\initstate,b)$:
indeed, 
the \kl{unitary marking} $\{d_1 \mapsto \{\qainf\}, d_2 \mapsto
\{\qainf\}, d_3 \mapsto \{\qafin\}\}$ fulfills $\EUpair( \qainf\mapsto
2;\{\qafin\})$, and the \kl{unitary marking} $\{d_1 \mapsto
\{\state_\top\}, d_2 \mapsto \{\qbinf\},$ $d_3 \mapsto
\{\state_\top\}\}$ fulfills $\EUpair(\qbinf\mapsto
1;\{\state_\top\})$.
\end{example}

\begin{example}\label{example-binary}
  Consider the \nAATA $\Aut$ over $\Alp=\{a\}$ with
  $\State=\{\initstate\}$,  and
\begin{xalignat*}2
  \Trans(\initstate,a) &= \EUpair(\initstate\mapsto 2; \emptyset)
  \end{xalignat*}
  and where $\Accept$ accepts all \kl{branches} (for~example, $\Accept$~is a \kl{parity condition} with
  $\prio(\initstate)=0$).
  This~automaton accepts a single \kl{tree}, namely the \SBtree{\{a\}}
  in which every node has exactly two successors.
  Letting
  \begin{xalignat*}2
  \Trans(\initstate,a) &= \EUpair(\initstate\mapsto 2; \emptyset)  \ou
     \EUpair(\emptyset; \emptyset)
  \end{xalignat*}
  would accept all \kl{binary trees} possibly containing finite \kl{branches} (i.e.,
  each node has either~$0$ or~$2$ \kl{successors}) if  $\prio(\initstate)=0$, but it would accept only finite \kl{binary trees} if~${\prio(\initstate)=1}$.
\end{example}

\begin{example}\label{ex-WtoT}
\AP
  \kl{Parity automata on words} can be seen as \AAPTA running on \kl{trees} of
  \kl{arity}~$1$. Formally, a~\kl{non-alternating} \intro{parity word automaton}
  (\nAPWA for~short) over~$\Alp$ is a $4$-tuple
  $\calB=\tuple{\State,\initstate,\Trans,\prio}$ where
  $\Trans\colon \State\times\Alp \to \DBF{\State}$; an~execution
  of~$\calB$ over a \Aword{\Alp} $w=(w_i)_{0\leq i<\size w}$ is a
  \Aword{\State}~$s=(s_i)_{0\leq i<\size w+1}$ such that $s_0=\initstate$
  and $s_{i+1}\in \delta(s_i,w_i)$ for all~$0\leq i<\size
  w$. The~\kl{word}~$w$ is accepted by~$\calB$ if some execution of~$\calB$
  on~$w$ is accepted by~$\prio$.%

  Given a \nAPWA
  $\calB=\tuple{\State,\initstate,\Trans,\prio}$, we~can easily build
  an \nAAPTA accepting all \kl{trees} having at least one \kl{branch}~$b$
  whose \kl(B){word}~$\Bword b$
  is accepted by~$\calB$~\cite{KSV06}: we~let $\calA=\tuple{\State\cup\{\state_{\top}\},
    \initstate, \Trans', \prio}$, where
  for all $\alp\in\Alp$, we~define
      $\delta(\state_\top,\alp)=\top$,
      and
      for all~$\state\not=\state_\top$,
    we~let $\Trans'(\state,\alp)=
    \OU_{\state' \in \Trans(\state,\alp)}
    \EUpair(\state'\mapsto 1; \{\state_{\top}\})$;
  It~is easily seen that $\calA$ precisely
  accepts those \kl{trees} containing at least one \kl{branch} whose \kl(B){word} is
  accepted by~$\calB$:
  automaton~$\calA$ can mimic the behaviour of~$\calB$ along that
  \kl{branch}, and accept the rest of the \kl{tree}.
  This~entails the following result:
  \begin{proposition}\label{prop-autEx}
    Let~$\calB=\tuple{\State,\initstate,\Trans,\prio}$ be a \nAPWA. There exists a \nAAPTA~$\calA$ that
    accepts exactly all \kl{trees} containing at least one \kl{branch}~$b$ whose
    \kl(B){word}~$\Bword b$ is accepted by~$\calB$.
    The~\kl(A){size} of~$\calA$ is $(\size\State+1, \size\Trans, 1,1, \priomax)$.
  \end{proposition}
  
  Now assume that we want to build an \AAPTA accepting all \kl{trees} in
  which \emph{all} \kl{branches} are accepted by~$\calB$. The~construction
  above cannot easily be adapted: a~natural attempt consists in
  letting %
  $\Trans''(\state,\alp)= \EUpair(\emptyset; \Trans(\state,\alp))$,
  thereby allowing to choose a different state with which to explore
  each \kl{successor node}; however, this is not correct, because the
  automaton would have to make the same non-deterministic choices on
  the common prefix of two different \kl{branches}. This approach works
  if~$\calB$ is required to be deterministic
  (\ie,~if~$\Trans(q,\sigma)$~is a singleton for all~$\state\in\State$
  and all~$\alp\in\Alp$), and more generally, if it
  is \emph{history-deterministic}~\cite{KSV06,BL23}.
  In the deterministic case:
  \begin{proposition}\label{prop-autAll}
    Let~$\calB$ be a \emph{deterministic} \nAPWA.
    There exists a \nAAPTA~$\calA$ that
    accepts exactly all \kl{trees} in which the \kl(B){word}~$\Bword b$ of any \kl{branch}~$b$
    is accepted by~$\calB$.
    The~\kl(A){size} of~$\calA$ is $(\size\State, \size\Trans, 0,1, \priomax)$.
\forceqex
\end{proposition}
\end{example}

\subsection{Related formalisms}
\label{sec-relW}
We briefly review several related classes of tree automata that have been defined in the literature.
\paragraph{\intro*\FATAutomata~\cite{Rab69,Tho90,Lod21}.}
\AP
Many papers on tree automata assume \kl{trees} of fixed
\kl{arity}. The~transition functions are then defined as positive boolean
combinations of atoms of the form~$(d,q)$, such an atom specifying
that the $d$-\kl{successor} of the current node has to be visited by
the automaton in the state~$q$. This in particular requires to have
transition functions that depend on the \kl{arity} in the
input~tree. Such \kl{tree} automata clearly have a different
expressive power compared to~\AATAs, since one can distinguish the
first and second \kl{successors} in a \kl{binary tree}, by~using directions
in the transition functions; on~the other~hand,
\AATA~can
accept \kl{trees} of arbitrary \kl{arity}. Note also that \AATA are often
much more succinct, for example by allowing constraints of the form
$\EUpair(q \mapsto k;\{q_\top\})$ that require to enumerate all
possible subset of~$k$~\kl{successors} (i.e.,~directions) in the fixed-\kl{arity}
setting.

Finally, note also that with \kl{alternating} \faTAutomata, one can
use formula $(d,q) \et (d,q')$ in transitions to have the
$d$-\kl{successor} visited by both~$q$ and~$q'$. Such a formula is not
directly possible in the syntax of \AATA transitions, because
two \EUprs cannot assign some specific state to some
specific \kl{successor}; however, exploring a single \kl{successor} in
two different states can be achieved by using an extra
state\footnote{This also requires some technicalities to take care of
the acceptance condition; we omit those details here.}~$r_{q\et
q'}$, whose transition function is defined as the conjunction of those
of~$q$ and~$q'$.

\paragraph{\intro*\AmTAutomata~\cite{BG93}.}
\AP
\AmTAutomata are, to our knowledge, the first class of tree
automata that can handle \kl{trees} of arbitrary, varying branching
degree. Transitions in an \amTAutomaton are defined through
a \emph{stretch} function, which takes as input an integer~$d$ (the
\kl{arity} of the node being visited) and an identifier and returning a
$d$-tuple of states indicating, for each \kl{successor node}, the~state of
the automaton in which it will be explored.
\AmTAutomata also have a kind of
\kl{alternation} mechanism, which allows to explore (copies~of) the~same
node of the input \kl{tree} in different states. 

When used for encoding \CTL, \amTAutomata only rely on two
stretch functions: (roughly) one~that amounts to visiting all
\kl{successors} in the same state~$q$, and one that amounts to visiting all
\kl{successors} in the same state~$q$ \emph{but} one of them, which is
explored in state~$q'$. The~stretch functions can thus be assumed to
be fixed, so that any \CTL formula~$\phi$ can be turned into an equivalent
\amTAutomata of size linear in~$\size\phi$.

\paragraph{\intro*\BDAutomata~\cite{Wil99a}.}
\AP
In~\cite{Wil99a}, a different class of \kl{alternating} automata running on
tree of arbitrary branching degree is introduced: there, the
transition function is defined as a positive boolean combination
over~$\{\Box,\Diamond\}\times\State$, where $(\Box, \state)$ requires
that the execution explores all \kl{successors} of the current node in
state~$\state$, and $(\Diamond,\state)$ requires the execution to
explore one \kl{successor node} in state~$q$.  These automata run over
Kripke structures in~\cite{Wil99a}, but their semantics could
equivalently be defined over \kl{trees} of arbitrary \kl{arity} (which is a
special case), as in our setting.
Then $(\Box, \state)$ would correspond to constraint $(\emptyset;
\{\state\})$ in our formalism (``\emph{explore all \kl{successors} of
current node in state~$\state$}''), while $(\Diamond,\state)$ would
correspond to~$(\{\state\mapsto 1\}, \{\state_\top\})$ (``\emph{select
one \kl{successor} and explore it in state~$\state$}'').  Clearly the
class of \BDAutomata corresponds exactly to the
subclass of \AATA where the transition function uses only
\EUprs of the form $\EUpair(\state\mapsto 1;\{\state_\top\})$
or $\EUpair(\emptyset;\{\state\})$.
\begin{lemma}
  \BDAutomata are less expressive than \kl{\EU-automata}.
\end{lemma}

\begin{proof} Clearly,
  \BDAutomata can be turned into \kl{\EU-automata}. The
  converse is not true, in particular because
  \BDAutomata cannot impose an upper bound on the
  number of \kl{successors} of a node: if~a \kl{tree}~$\Tree$ is accepted by
  some \BDAutomaton, then the \kl{tree} obtained
  from~$\Tree$ by duplicating one \kl{branch} is also accepted.
\end{proof}

\AP
\BDAutomata (on~\kl{trees}) impose no restrictions on
the \kl{arity} of the \kl{trees} they take as input, and do not distinguish
between the different \kl{successors} of any node of the input
tree. As~such, they are often named \intro{symmetric} tree
automata. Other variants of \kl{symmetric} tree automata have been studied
in the literature.

\paragraph{\intro*\SymBTAutomata~\cite{KV03b}.}
\AP
In~\cite{KV03b}, a class of \emph{symmetric}
non-deterministic\footnote{With our terminology, we would name them
\kl{non-alternating}.} tree automata is
defined. The~general aim of this class is to have a \kl{non-alternating}
equivalent to \BDAutomata of~\cite{Wil99a}. \SymBTAutomata
(\reintro*\symNBT for short)
are automata obtained from \BDAutomata by applying
a powerset construction: the~set~$Q$ of states of a \symNBT is of the
form~$2^S$ (where $S$ is the set of states of the \BDAutomata, whose elements are coined \emph{micro-states} to distinguish them with the \emph{macro-states} of~$Q$), and transitions return sets of pairs~$[U;E]\in
2^S\times 2^S$, where $[U;E]$ intuitively means $\Box U \et
\Diamond E$: all micro-states in~$U$ must be present in \emph{all}
macro-states visiting the \kl{successors} of the current node, and all
micro-states in~$E$ must be present in \emph{some} macro-state
visiting the \kl{successors} of the current node.

Any \BDAutomaton can be turned into 
a \symNBT~\cite{KV03b}.  In~a sense, our~Theorem~\ref{thm-simu} is an
extension of this result to a richer class of tree automata.

\paragraph{\intro*\MSOAutomata~\cite{JW95,Wal96b,Wal02,BB02,JL04,Zan12}}
\AP
      In~\cite{JW95,Wal96b,Wal02,BB02,JL04,Zan12},
      arbitrary-\kl{arity} tree automata are defined where the transition
      functions are defined using first-order formulas over the
      \kl{successor} nodes; these automata are named \MSOAutomata.
      In~\MSOAutomata, transitions are given as first-order formulas,
      with quantification over the \kl{successors} of the current node and
      predicates corresponding to states of the automaton.  For
      instance, formula $\exists x.\ q(x)$ corresponds to $\Diamond q$
      in~\BDAutomata: some \kl{successor} must be
      visited by the automaton in state~$q$.

   Using first-order logic in transition function provides great
      flexibility. However, in order to facilitate the manipulation of those
      automata, the first-order formulas defining the
      transitions can be turned into a disjunction of formulas in
      \emph{basic form}:
      \[
      \exists (x_i)_{1\leq i\leq k} \diff((x_i)_{1\leq i\leq k})
      \et
      \ET_{1\leq i\leq k} q_i(x_i)
      \et
      \forall y.\ y\notin (x_i)_{1\leq i\leq k} \impl \OU_{q\in U} q(y).
      \]
      Such formula can be seen to correspond to our $\EUpair(E;U)$
      formulas, so that \MSOAutomata have the same expressive power
      as our \kl{\EU-automata}.  However, the transformation of a
      first-order formula into such a disjunction is based on
      Ehrenfeucht-Fra\"\i ss\'e games, and is not explicited
      in~\cite{Wal02,BB02,Zan12}. In~the sequel, we~develop operations
      (union, intersection, projection, complementation and
      simulation) with explicit constructions and precise evaluation
      of the \kl(A){size} of the resulting automata, which cannot be directly
      obtained from the current results about \MSOAutomata.

\section{\kl{Game-based semantics}}
\label{sec-gamesem}

The acceptance of a \kl{tree} by an \AAPTA can be expressed as the existence
of a winning strategy in a \kl{two-player turn-based parity
game}.  We first briefly recall the definition of \kl{parity
games}, and then explain how they can be used to encode the semantics
of \kl{alternating} tree automata.

\subsection{\kl{Parity games}}

A \intro{two-player turn-based parity game} is a $4$-tuple
$\Game=\tuple{\GState_0, \GState_1,R,\gprio}$ where
$\GState_0$ and~$\GState_1$ are disjoint  sets of states and,
writing~$\GState=\GState_0\cup\GState_1$, 
$R\subseteq \GState^2$ is a set of transitions, and $\gprio\colon
\GState \to \bbN$ assigns a \kl{priority} to each state of
the~game.

In such a game, two~players (which we name \Pl0 and~\Pl1) select
transitions so as to form a \intro(G){path} in the
graph~$\tuple{\GState,R}$: from some~$\gst\in \GState_i$
(with~$i\in\{0,1\}$), \Pl i selects a transition~$(\gst,\gst')\in R$,
and the game proceeds to~$\gst'$.  A~\kl(G){path}
is \intro(P){maximal} if it is infinite, or if its last state has no
outgoing transitions.  A~finite \kl(P){maximal} \kl(G){path}
is \intro(P){winning} for \Pl0 if, and only if, its last state belongs
to $\GState_1$ (the~blocked player loses).  For an
infinite \kl(G){path}~$\pi$, we~write $\gprio_{\min}(\pi)$ for the
least integer $k$ s.t.\ there are infinitely many $i\geq 0$ with
$\gprio(\pi(i))=k$.  The infinite \kl(G){path}~$\pi$
is \reintro(P){winning} for \Pl0 if, and only~if, $\gprio_{\min}(\pi)$
is even; otherwise it is \reintro(P){winning} for~\Pl1.

A~\intro{strategy} for~\Pl i in a \kl{parity game} is a partial
function $\strat i\colon \GState^*\times \GState_i \to R$ such that
for any finite \kl(G){path}~$\pi\cdot
\gst$ with~$\gst\in \GState_i$, the~function~$\strat i$ is defined at $\pi\cdot \gst$ if, and only~if,
there exists~$\gst'$ such that $(\gst,\gst')\in R$; in that case,
we~must have $(\gst,\strat i(\pi\cdot \gst)) \in R$. A~strategy
is \intro{memoryless} if for any two paths~$\pi\cdot\gst$
and~$\pi'\cdot\gst$, it~holds $\strat i(\pi\cdot \gst)=\strat
i(\pi'\cdot \gst)$.

A~\kl(G){path}~$\pi=(\gst_j)_{0\leq j\leq\size\pi}$ is
\intro{compatible} with a \kl{strategy}~$\strat i$ of \Pl i if for any $0\leq
j<\size\pi-1$, if $\gst_j\in \GState_i$, then $\gst_{j+1}=\strat i((\gst_k)_{0\leq
k\leq j})$. A~\kl{strategy}~$\strat i$ is \intro(S){winning} for \Pl i from~$\gst$
if all \kl(P){maximal}
\kl(G){paths} starting from~$\gst$ that are \kl{compatible} with~$\strat i$ are \kl(P){winning} for \Pl i.

The following classical property about infinite \kl{parity games} will be useful in the sequel:

\begin{proposition}[\cite{tcs200(1-2)-Zie}]
\label{prop-determined}
\kl{Two-player turn-based parity games} are positionally determined: from any state of such games,
one of the two players has a \kl{memoryless} \kl(S){winning} strategy.
\end{proposition}

\subsection{\intro{Game semantics} for tree automata}
\label{sec:sem-game}

We now explain how to define a \kl{two-player turn-based parity
game}~$\GAT=\tuple{\GState_0, \GState_1,R,\gprio}$ encoding
the \kl{acceptance} of a \kl{tree}~$\Tree=\tuple{\tree,\lab}$ by
an \AAPTA~$\Aut=\tuple{\State,\initstate,\Trans,\prio}$.

States of~$\GAT$ are of three kinds: the~\intro{main states} of the
game are of the form $(n,\state)$ where $n \in t$ and
$\state\in\State$; the~\intro{auxiliary states} are of the form
$(n,q,\vfi)$ where $\vfi$ is a subformula of~$\Trans(\state,l(n))$,
and $(n,q,\nu_n)$ where $\nu_n$~is a \kl{unitary marking} of~$\succ(n)$
with states of~$\Aut$.  It~remains to partition them into~$Y_0$
and~$Y_1$. The~set~$Y_0$ contains the states of the following form:
\begin{itemize}
\item the main states $(n,\state)$:
\item the state $(n,\state,\bot)$;%
\item the states of the form $(n,\state,\OU_{1\leq i\leq k} \phi_i)$ (with $k>1$); %
\item the states of the form $(n,\state,\EUpair(E;U))$,
  which correspond to positions where \Pl0 has to assign states to the
  \kl{successors} of~$n$ so as to satisfy~$\EUpair(E;U)$; 
\end{itemize}
All other states belong to~$\GState_1$ (\ie, states of the form
$(n,\state,\top)$, $(n,\state,\ET_{1\leq i\leq k} \phi_i)$, and~$(n,q,
\nu_n)$).

Transitions are defined as follows:
\begin{itemize}
\item
  for any state~$(n,\state)$ with $n\in\tree$ and $\state\in\State$,
  there is a single transition from
  $(n,\state)$ to~$(n,\Trans(\state,\lab(n)))$;
\item from each of $(n,\state,\OU_{1\leq i\leq k} \phi_i)$ and $(n,\state,\ET_{1\leq i\leq
  k} \phi_i)$, there are transitions to $(n, \state,\phi_i)$, for each~$1\leq
  i\leq k$;
\item from $(n,\state,\EUpair(E;U))$, for each \kl{unitary marking}~$\nu_n$
  of~$\succ(n)$ 
  with states in~$\State$ such that $\nu_n\models+\EUpair(E;U)$, there
  is a transition to~$(n,\state,\nu_n)$. Notice that for~$\nu_n$ to
  fulfill~$\EUpair(E;U)$, we~must have
  $\supp(\mkimg{\nu_n})\subseteq \supp(E)\cup U$.
\item from $(n,\state,\nu_n)$ where $\nu_n$ is a \kl{unitary marking}
  of~$\succ(n)$
  with states of~$\State$, for each~$n_i\in\succ(n)$, there is a
  transition to~$(n_i,\nu_n(n_i))$.
\end{itemize}

Finally, \kl{priorities} are defined as $\gprio(n,\state)=\prio(\state)$
for the \kl{main states}, and to~$\priomax+1$
for all other states of the game. Since there are infinitely many main states
along infinite runs, only \kl{priorities} of the \kl{main states} are useful.

\begin{example}
Consider an automaton with $\Trans(q_0,a) = \EUpair(q_0\mapsto
1; \{q_1\}) \et \EUpair(q_0 \mapsto 2;\{q_\top\})$,
$\Trans(q_0,b)=\top$, and $\Trans(q_1,a) =\Trans(q_1,b)
= \EUpair(q_\top \mapsto 1; \emptyset)$.
Figure~\ref{fig-ex-game} shows the first levels of an input \kl{tree}
and the beginning of the corresponding \kl{parity game} for this
transition function (dotted nodes belong to Player~1). To simplify the
figure, we have omitted $q_0$ in states of the form
$(\varepsilon,q_0,\vfi)$ and $(\varepsilon,q_0,\nu_n)$.
\end{example}

\begin{figure}[t]
\centering
\begin{tikzpicture}
  \begin{scope}
    \path (0,.6) node {input \kl{tree}};
    \draw (0,0) node (eps) {$\varepsilon (a)$}; %
    \draw (-10mm,-16mm) node (t0) {$0 (a)$};
    \draw (0,-16mm) node (t1) {$ 1(b)$};
    \draw (10mm,-16mm) node (t2) {$2 (a)$};
 \draw[-latex'] (eps) -- (t0);
 \draw[-latex'] (eps) -- (t1);
 \draw[-latex'] (eps) -- (t2);
  \foreach \n/\a in {t0/-110,t0/-70,t1/-90,t2/-110,t2/-90,t2/-70}
    {\draw[dashed,-latex'] (\n) -- +(\a:7mm);}

  \end{scope}
  \begin{scope}[xshift=6.5cm]
    \path (0,.6) node {partial view of~$\GAT$};

    \draw (0,0) node[draw] (nu) {$\varepsilon,q_0$} ;
    \draw (0,-1cm) node[draw,dotted,rounded corners=2mm] (nu2) {$\varepsilon,\EUpair(q_0\mapsto 1; \{q_1\}) \et \EUpair(q_0 \mapsto 2;\{q_\top\})$} ;
    \draw[-latex'] (nu) -- (nu2);
    \draw (-1.5cm,-2cm) node[draw] (nu3) {$\varepsilon,\EUpair(q_0\mapsto 1; \{q_1\})$} ;
    \draw (1.5,-2cm) node[draw] (nu4) {$\varepsilon,\EUpair(q_0 \mapsto 2;\{q_\top\})$} ;
    \draw[-latex'] (nu2) -- (nu3);
    \draw[-latex'] (nu2) -- (nu4);
    \draw (2.8cm,-3cm) node[draw,dotted,rounded corners=2mm] (nu10) {${\scriptscriptstyle \varepsilon,(0:q_\top,1:q_0,2:q_0)}$} ;
    \draw (1.6cm,-3.5cm) node[draw,dotted,rounded corners=2mm] (nu9) {${\scriptscriptstyle \varepsilon,(0:q_0,1:q_\top,2:q_0)}$} ;
    \draw (0,-4cm) node[draw,dotted,rounded corners=2mm] (nu8) {${\scriptscriptstyle \varepsilon,(0:q_0,1:q_0,2:q_\top)}$} ;
    \draw[-latex'] (nu4) -- (nu10);
    \draw[-latex'] (nu4) -- (nu9);
    \draw[-latex'] (nu4) -- (nu8);
   \draw (-1cm,-3cm) node[draw,dotted,rounded corners=2mm] (nu7) {${\scriptscriptstyle \varepsilon,(0:q_1,1:q_1,2:q_0)}$} ;
    \draw (-2.5cm,-3.5cm) node[draw,dotted,rounded corners=2mm] (nu6) {${\scriptscriptstyle \varepsilon,(0:q_1,1:q_0,2:q_1)}$} ;
    \draw (-4.8cm,-4.2cm) node[draw,dotted,rounded corners=2mm] (nu5) {${\scriptscriptstyle \varepsilon,(0:q_0,1:q_1,2:q_1)}$} ;
    \draw[-latex'] (nu3) -- (nu7);
    \draw[-latex'] (nu3) -- (nu6);
    \draw[-latex'] (nu3) -- (nu5);
    \draw (-5.7cm,-5cm) node[draw] (nu11) {${\scriptstyle 0,q_0}$} ;
    \draw (-4.8cm,-5cm) node[draw] (nu12) {${\scriptstyle1,q_1}$} ;
    \draw (-3.9cm,-5cm) node[draw] (nu13) {${\scriptstyle2,q_1}$} ;
    \draw[-latex'] (nu5) -- (nu11);
    \draw[-latex'] (nu5) -- (nu12);
    \draw[-latex'] (nu5) -- (nu13);
    \draw (-1.7cm,-4.5cm) node[draw] (nu14) {${\scriptstyle 2,q_1}$} ;
    \draw (-2.5cm,-4.5cm) node[draw] (nu15) {${\scriptstyle1,q_0}$} ;
    \draw (-3.3cm,-4.5cm) node[draw] (nu16) {${\scriptstyle0,q_1}$} ;
    \draw[-latex'] (nu6) -- (nu14);
    \draw[-latex'] (nu6) -- (nu15);
    \draw[-latex'] (nu6) -- (nu16);
    \draw (-.8cm,-4.7cm) node[draw] (nu18) {${\scriptstyle 0,q_0}$} ;
    \draw (0cm,-4.7cm) node[draw] (nu19) {${\scriptstyle1,q_0}$} ;
    \draw (.9cm,-4.7cm) node[draw] (nu20) {${\scriptstyle2,q_\top}$} ;
    \draw[-latex'] (nu8) -- (nu18);
    \draw[-latex'] (nu8) -- (nu19);
    \draw[-latex'] (nu8) -- (nu20);
    \foreach \a in {-80,-60,-40} {\draw[dashed,-latex'] (nu10) -- +(\a:6mm);}
    \foreach \a in {-100,-80,-60} {\draw[dashed,-latex'] (nu9) -- +(\a:6mm);}
    \foreach \a in {-110,-90,-70} {\draw[dashed,-latex'] (nu7) -- +(\a:6mm);}
    \foreach \n in {11,12,13,14,15,16,18,19,20} {\draw[dashed,-latex'] (nu\n) -- +(-90:6mm);}
    
  \end{scope}
\end{tikzpicture}
\caption{Example of game~$\GAT$}
\label{fig-ex-game}
\end{figure}

The resulting \kl{parity game} encodes the \kl{acceptance} of a \kl{tree} by an~\AAPTA:
\begin{proposition}\label{prop-gamesem}
The \SDtree{\Alp}{\Dir}~$\Tree$
is accepted
by the \AAPTA~$\Aut$ if, and only~if, \Pl0 has a \kl(S){winning} \kl{strategy}
from state~$(\emptyw_\tree, \initstate)$
in the associated \kl{parity game}~$\GAT$.
\end{proposition}

\begin{proof}
We prove a slightly stronger result:
\Pl0 has a winning strategy from a \kl{main state}~$(n,\state)$
in~$\GAT$ if, and only if, the input \kl{tree} rooted at~$n$ is
\kl{accepted} by the automaton~$\Aut$ with $\state$ considered as the
initial state (\ie, there exists an \kl(ET){accepting} \kl{execution
  subtree} rooted at a node labelled with~$(n,\state)$).

First assume that Player~$0$ has
a \kl(S){winning} \kl{strategy}~$\alpha_0$: by~pruning all subtrees
that are not selected by~$\alpha_0$, and removing non-main states,
we~get a tree which we can prove is
an \kl(ET){accepting} \kl{execution tree}: boolean operators of the
transition function are handled correctly, and from the states of the
form~$(n,\state,\EUpair(E;U))$, \kl{strategy}~$\alpha_0$ selects a
valid way of exploring the \kl{successor nodes}, so~as to fulfill
the \EUconstr, therefore the marking of $\succ(n)$ induced by this tree satisfies $(n,\Trans(\state,\lab(n)))$.
  By~definition of the priorities of the states
of~$\GAT$, all~infinite \kl{branches} are \kl(B){accepting} since
infinite paths are \kl(P){winning} for
Player~$0$. Finite \kl{branches} in the tree may only originates
from \kl{auxiliary states} of the form~$(n,\state,\bot)$
and~$(n,\state,\top)$, and of the form $(n,\state,\nu_n)$ for
which~$n$~has no \kl{successor nodes}. The~first two cases are
correctly handled by construction of the game (the~blocked player
loses). Since states of the form~$(n,\state,\nu_n)$ belong to
Player~$1$, they~are winning for Player~$0$ in case~$n$~has
no \kl{successor nodes}, \ie, in case $\nu_n$ is the \kl{marking} of the
empty set; but there may only be transitions
from~$(n,\state,\EUpair(E;U))$ to $(n,\state,\nu_n)$ with the empty
\kl{marking} when $\EUpair(E;U)$ is of the form~$\EUpair(\emptyset;U)$:
a~finite \kl{branch} of the \kl{execution tree} ending in a node
labelled~$(n,\EUpair(\emptyset;U)$ is indeed \kl(B){accepting}.

The converse is similar: given an \kl(ET){accepting} \kl{execution
tree}, we~can build a \kl(S){winning} \kl{strategy} for
Player~$0$. For~this, it~suffices to check which parts of disjunctive
formulas in transitions are satisfied, and how \EUconstrs are
fulfilled. The~winningness for infinite and finite paths then again
corresponds to the acceptance status of the corresponding \kl{branches} of the
\kl{execution tree}.
\end{proof}

Note that the subgames issued from $(n,\state,\vfi)$ and
$(n,\state',\vfi)$ (respectively, ~from~$(n,\state,\nu_n)$ and
$(n,\state',\nu_n)$) are isomorphic and admit the same wining
memoryless strategies. Therefore we do not distinguish them in the
following and consider only nodes of the form~$(n,\vfi)$
or~$(n,\nu_n)$.

When $\Tree$ is regular and corresponds to the \kl{execution tree} of some
\kl{Kripke structure}~$\calK=\tuple{V,E,\ell}$, we~can build
a \emph{finite} game
$\GAK=\tuple{\GState_0, \GState_1,R,\gprio}$ defined
exactly as above, but replacing nodes~$n$ of~$\tree$ with vertices~$v$
of~$V$.  The~sizes of~$\GState$ and of the transition relation~$R$ are
both in
$O(\size{V}\cdot(\size\State\cdot(1 + \sizeB{\Trans})+\size{\State}^{\textsf{arity}(\calK)}))$,
hence in $O(\size{V}\cdot( \size\State\cdot \sizeB{\Trans}+\size{\State}^{\size{V}}))$.

Moreover it is worth noticing that the complexity blow-up in the size
of the game is due to the treatment of \EUconstrs. For~automata
using only simple constraints of the form $\EUpair(\state\mapsto
1;\{\state_\top\})$ or $\EUpair(\emptyset;\{\state\})$ (which
correspond to \BDAutomata), the~number of \kl{unitary
markings} involved in reachable states of the form~$(n,\nu_n)$ in~$\GAK$ 
is~$O(\size V\cdot \size\State)$, so~that
the~sizes of~$Y$ and~$R$
are
in 
$O(\size{V}\cdot\size{\State}\cdot (\sizeB{\Trans}+\size V))$.

\section{Operations on \AATAs}
\label{sec-ops}

This section is the main technical part of our paper: we~develop
algorithms for performing various operations on \AAPTAs (namely union
and intersection, projection, complementation and alternation
removal), and carefully study the size of the \AAPTAs we obtain.

The~rest of this section gives detailed algorithms,
explanation of their correctness, and justifications for the sizes of
the resulting automata.

Section~\ref{ssec-summary} displays a table containing the sizes of
the automata obtained by our transformations.

\subsection{Union and intersection}
\label{ssec-union}\label{ssec-inters}

Union and intersection are straightforward for \AATAs, thanks to
\kl{alternation}. 
\begin{theorem}
  \label{thm-union}\label{thm-inters}
  Let $\Aut=\tuple{\State,\initstate,\Trans,\omega}$ and
  $\Aut'=\tuple{\State',\initstate',\Trans',\omega'}$ be two \AAPTAs.
  There exist \AAPTAs $\Aut_\cup$ and
  $\Aut_\cap$, respectively accepting the union and the intersection
  of the \kl{language} of~$\Aut$ and~$\Aut'$,
  and having size at most
  $\tuple{{{\size \State+\size{\State'}+1}},\penalty0\relax
  \sizeB\Trans+\sizeB{\Trans'}+1, \max(\sizeE\Trans,\sizeE{\Trans'}),
\max(\sizeU\Trans,\sizeU{\Trans'}),  \max(\priomax[\omega],\priomax[\omega'])+1}$.
\end{theorem}

\begin{proof}
  We~consider intersection (the~case of union is similar):
  automaton~$\Aut_{\cap}$ is defined as $\tuple{\State'',
    \initstate'',\Trans'', \omega''}$ with
  \begin{itemize}
  \item $\State''=\State\cup\State'\cup\{\initstate''\}$ (assuming
    w.l.o.g. that all three states are pairwise disjoint);
  \item $\Trans''(\initstate'',\alp)=\Trans(\initstate,\alp) \et
    \Trans'(\initstate',\alp)$, and $\Trans''$ coincides with $\Trans$
    on $\State\times\Alp$ and with~$\Trans'$ on $\State'\times\Alp$;
  \item $\omega''$ coincides with $\omega$ on~$\State$ and with~$\omega'$
    on~$\State'$; its value in~$\initstate''$ is irrelevant since $\initstate''$ will be visited only once. Notice that, using straightforward arguments, we~may assume that $\omega(\State)$ and $\omega'(\State')$ are
    subintervals
    of~$\llbracket 0;\priomax[\omega]\rrbracket$ and~$\llbracket 0;\priomax[\omega']\rrbracket$, so that
    $\omega''(\State'')$ is a
    subinterval
    of $\llbracket 0; \max(\priomax[\omega],\priomax[\omega'])\rrbracket$.
  \end{itemize}

  The correctness of this construction is not hard to prove, using the \kl{game semantics}.
  Consider~a~tree~$\Tree=\tuple{\tree,\lab}$  accepted
  by~$\Aut_{\cap}$. By~Prop.~\ref{prop-gamesem}, Player~$0$ has a
  \kl(S){winning} \kl{strategy} from~$(\emptyw_\tree,
  \initstate'')$ in the corresponding game~$\GATcap$.  Since~$\Trans''(\initstate'',\allowbreak
  \lab(\emptyw_\tree))=\Trans(\initstate, \lab(\emptyw_\tree)) \et
  \Trans'(\initstate', \lab(\emptyw_\tree))$, there is a unique
  transition from~$(\emptyw_\tree, \initstate'')$ to the Player-1
  state $(\emptyw_\tree, \Trans''(\initstate'',
  \lab(\emptyw_\tree)))$; from~there, Player~1 can decide to move
  either to $(\emptyw_\tree, \Trans(\initstate, \lab(\emptyw_\tree)))$
  or to $(\emptyw_\tree, \Trans'(\initstate',
  \lab(\emptyw_\tree)))$. Since Player~$0$ has a \kl(S){winning} \kl{strategy}
  from~$(\emptyw_\tree, \initstate'')$, she also has \kl(S){winning}
  \kl{strategies} from both~$(\emptyw_\tree, \Trans(\initstate,
  \allowbreak \lab(\emptyw_\tree)))$ and $(\emptyw_\tree, \Trans'(\initstate',
  \lab(\emptyw_\tree)))$ in~$\GATcap$. Since $\Trans''$ coincides
  with~$\Trans$ on~$\State\times\Alp$ and with~$\Trans'$
  on~$\State'\times\Alp$, Player~$0$ has \kl(S){winning} \kl{strategies}
  from~$(\emptyw_\tree, \initstate)$ in~$\GAT$ and
  from~$(\emptyw_\tree, \initstate')$ in~$\GATp$. Hence $\Tree$ is \kl{accepted} by both~$\Aut$ and~$\Aut'$.

  The converse implication follows the same arguments.
\end{proof}

\begin{remark}\label{rem-sizeomegaUI}
  Note
  that if the minimum priority of~$\prio$ and~$\prio'$ are
  equal, then
  the number of priorities in~$\Aut_\cup$ and $\Aut_\cap$ can be
  bounded by $\max(\size\prio,\size{\prio'})$ instead of
  $\max(\size\prio,\size{\prio'})+1$.
\end{remark}

\subsection{Projection}
\label{ssec-proj}

Given an \AAPTA~$\Aut$ over alphabet~$\Alp_1\times\Alp_2$, \intro{projection}
consists in building another \AAPTA~$\Aut_1$, over alphabet~$\Alp_1$,
accepting all
\Strees{\Alp_1}
whose labelling can be extended
on~$\Alp_1\times\Alp_2$ to make the \kl{tree} \kl{accepted} by~$\Aut$.  This is a
classical construction, and it can be performed easily
on \kl{non-alternating} automata~\cite{MS85}.

Formally, two
\Strees{\Alp_1\times\Alp_2}~$\calT=\tuple{t,l}$
and~$\calT'=\tuple{t',l'}$ are said
\intro*\labequiv{\Alp_1},
denoted
$\calT\reintro*\slabequiv{\Alp_1} \calT'$, whenever $t=t'$ and for any node~$n$ of
these trees,
it~holds $\proj_1(l(n))=\proj_1(l'(n))$.
\labequivce{\Alp_2}~is defined analogously.

\begin{theorem}%
  \label{thm-proj}
  Let~$\Aut=\tuple{\State,\initstate,\Trans,\prio}$
  be an \nAAPTA over $\Alp=\Alp_1\times\Alp_2$. For
  each~$i\in\{1,2\}$, we~can build an \nAAPTA $\Aut_i$ over $\Alp$ such
  that, for any \Stree{\Alp}~$\calT$, it~holds:
  $\calT \in \Lang(\Aut_i)$ if, and only~if, there is a
  $\Alp$-labelled tree~$\calT'$ in~$\Lang(\Aut)$ such that
  $\calT \slabequiv{\Alp_i} \calT'$.  The~size of~$\Aut_i$ is
     at most
  $\tuple{\size\State, \size{\Alp_{3-i}}\cdot \sizeB\Trans, \sizeE\Trans,\sizeU\Trans,\priomax}$.
\end{theorem}

\begin{proof}
We~define $\Aut_1$ over~$\Alp$
as $\tuple{\State,\initstate,\Trans_1,\prio}$ with:
\[
\Trans_1(q,(\alp_1,\alp_2)) = \OU_{\alp'_{2} \in \Alp_{2}}
  \Trans(q, (\alp_1, \alp'_2)).
\]

Take a tree~$\Tree=\tuple{\tree,\lab}$  accepted
by~$\Aut_1$, and pick an accepting execution
tree~$\ExTree=\tuple{\tree,\exlab}$ (automaton~$\calA_1$ is
non-alternating, so~we~can assume that the execution tree has the same
structure as the input~tree). Consider any node~$n$ of~$\tree$,
labelled with~$\lab(n)=(\sigma_1,\sigma_2)$ in~$\Tree$ and
with~$\exlab(n)=(n, \state)$ in~$\ExTree$. By~definition of~$\delta_1$,
there exists $\sigma^n_2\in\Sigma_2$ such that the successors of~$n$
in~$\ExTree$ satisfy $\delta(\state,(\sigma_1,\sigma^n_2))$. This~holds for
all nodes of~$\tree$, meaning that each node~$n$ can be relabelled
with~$(\sigma_1,\sigma_2^n)$ in such a way that this new tree~$\Tree'$ is
accepted by~$\Aut$.

Conversely, assume that tree~$\calT$ admits an
\labequiv{\Alp_1}
\kl{tree}~$\calT'$ that is accepted by
the~\nAAPTA~$\Aut$, and take an accepting execution
tree~\ExTree. By~construction of~$\delta_1$, this~execution tree is
also an accepting execution tree for~$\Aut_1$ on~$\calT$.
\end{proof}

\begin{remark}
This construction does not directly extend to alternating automata: 
indeed, let ${\Sigma_1 = \{a_1\}}$ and
${\Sigma_2=\{a_2, a_2'\}}$, and consider the \AAPTA
$\Aut = \tuple{\State, \initstate, \Trans, \prio}$ on $\Sigma_1\times\Sigma_2$ with:
\begin{itemize}
\item $\State=\{\initstate, \state_1, \state_2\}$, 
\item The transition  function is defined as follows (for any~$\sigma\in\Sigma_1\times\Sigma_2$):
\begin{xalignat*}3
\delta(\initstate,(a_1,a_2)) &= \delta(\initstate,(a_1,a'_2)) = \EUpair(\state_1 \mapsto 1;\emptyset) \et \EUpair(\state_2 \mapsto 1;\emptyset)   \\
\delta(\state_1,(a_1,a_2)) &= \delta(\state_2,(a_1,a'_2)) = \top  \\
 \delta(\state_1,(a_1,a'_2)) &= \delta(\state_2,(a_1,a_2)) = \bot 
\end{xalignat*}
\end{itemize}
We~let $\omega(\state)=0$ for
all  states.  Now,
for a \kl{tree} to be \kl{accepted}, its~\kl{root} must be labelled with~$(a_1,a_2)$ or $(a_1,a'_2)$ 
and have a single \kl{successor node}. That node will be explored both in
state~$\state_1$ and~$\state_2$, and for any~$\sigma$, either $\delta(\state_1,\sigma)$ or $\delta(\state_2,\sigma)$ is equal to~$\bot$. 
It~follows that $\calL(\Aut)=\emptyset$, hence also $\calL(\Aut_i)=\emptyset$ for $i\in\{1,2\}$. 

Now, applying our construction to this automaton
on the first component provides us with
an automaton~$\Aut_1$ with the following transition function: 

\begin{xalignat*}2
\delta_1(\initstate,\sigma) &=  \EUpair(\state_1 \mapsto 1;\emptyset) \et \EUpair(\state_2 \mapsto 1;\emptyset) & \forall\alp\in\Alp_1\times\Alp_2  \\
\delta_1(\state_1,\sigma) &= \delta(\state_1,(a_1,a_2)) \ou  \delta(\state_1,(a_1,a'_2)) =  \top \ou \bot = \top  \\
 \delta_1(\state_2,\sigma) &= \delta(\state_2,(a_1,a_2)) \ou \delta(\state_2,(a_1,a'_2))=  \bot \ou \top = \top 
\end{xalignat*}
This automaton accepts any \kl{tree} on~$\Sigma_1\times\Sigma_2$ whose
root has a single \kl{successor}.
This shows that Theorem~\ref{thm-proj} does not hold for alternating automata.

For similar reasons, this theorem does not extend
to \emph{universal} projection, where $\Aut_i$~would accept the trees
for which \emph{all} $\Alp_i$-equivalent trees would be accepted
by~$\Aut$.
\end{remark}

\subsection{Complementation}
\label{ssec-compl}

\intro{Complementation} is the operation of building an automaton
accepting the complement of the language accepted by some given automaton.
It~is usually easy for \kl{alternating} \kl{parity}
automata: it~suffices to dualise the transition function (swapping
disjunctions and conjunctions) and shifting the \kl{priorities}.
Such~a~construction is given in~\cite{Kir02}
for~\BDAutomata: in~that setting, $\Box$~and~$\Diamond$
are dual to each other, and the construction is
straightforward. The~same is true for \MSOAutomata~\cite{Wal02,Zan12}.
For~our~\AATA however, 
we~need to express the negation of any \EUpr~$\EUpair(E;U)$ as
an \EUconstr.

The question then is to characterise those nodes that \emph{fail to
satisfy} an \EUpr~${\EUpair(E;U)}$.
There can be two reasons for this, which we develop in the sequel:
\begin{enumerate}
\item either we cannot find $\MSsize E$ successors of the current node~$n$
  to associate with the $\MSsize E$ states of the existential part,
\item or for every way to satisfy~$E$ with nodes in~$\succ(n)$, 
  there~remain \kl{successors} that are \kl{accepted} by no states
  in~$U$.
  Equivalently, no~sets of nodes in~$\succ(n)$ that contain all the nodes
  that are not \kl{accepted} by any state in~$U$, can exactly match~$E$.
  This includes as a special case
  the situations where we have more than $\MSsize E$ successors that are
  \kl{accepted} by no states in~$U$.
\end{enumerate}

\subsubsection{Failing to satisfy the existential part of $\protect\EUpair(E;U)$.}
We first address the former situation, which is easier and
already contains most of the technicalities we need for solving the
general case.

Fix an \AAPTA $\Aut=\tuple{\State, \initstate,\Trans,\prio}$
over~$\Alp$.
We~assume w.l.o.g.\ that~$\Aut$ has a state~$\state_\top$ from which all
\kl{trees} are \kl{accepted}.
For a state~$\state\in\State$, we~write~$\Aut_\state$ for the
\AAPTA $\tuple{\State, \state,\Trans,\prio}$, obtained from~$\Aut$
by taking~$\state$ as the initial state.
For~a \kl{multiset}~$E$ over~$\State$,
we~define the \AAPTA $\Aut_{\suc E}=\tuple{\State\cup\{\state_{E}\},
  \state_{E}, \Trans_{E}, \prio_{E}}$ such that
\begin{itemize}
\item $\state_{E}$ is a new state, not in~$\State$; 
\item for any~$\alp\in\Alp$,
  $\Trans_{E}(\state_{E},\alp)=\EUpair(E;\{\state_\top\})$ and
  $\Trans_{E}(\state,\alp)=\Trans(\state,\alp)$ for
  all~$\state\not=\state_{E}$;
\item the priority function~$\prio_E$ coincides with~$\prio$ on~$\State$,
  and $\prio_E(\state_{E})=0$.
\end{itemize}

Automaton~$\Aut_{\suc E}$ visits some~of the \kl{successors} of the \kl{root}
in each of the states of~$E$.
Notice that if~$E=\emptyset$, then $\Aut_{\suc E}$ \kl{accepts} any \kl{tree}.
Notice also that the automaton $\Aut_{\state}$ defined above
is \emph{not} the same automaton as~$\Aut_{\suc E}$ with
$E=\{\state\}$: automaton~$\Aut_\state$ starts exploring the
\kl{root}~$\emptyw_\Tree$ of its input \kl{tree} in state~$\state$, whereas
$\Aut_{\suc \{q\}}$  explores the subtree rooted at some
\kl{successor node} of~$\emptyw_\Tree$ in state~$\state$.

In the sequel, given a tree~$\Tree$ and a node~$n$, we~write $\Tree_n$
for the subtree of~$\Tree$ rooted at~$n$.  We~call \intro{direct
subtree} of~$\Tree$ any subtree $\Tree_n$ where $n\in\succ(\troot_\Tree)$.

\medskip

\AP
Our~aim in this part is to compute an \AAPTA~$\CAut_{E}$
such that
$\Lang(\CAut_{E})$ is the complement of~$\Lang(\Aut_{\suc E})$.
This~construction is the main ingredient for complementing
an \kl{\EU-automaton}.
It~is based on  the notion of \kl{blocking pairs}:
\begin{proposition}\label{prop-bp}

Let $\Tree$ be an input \kl{tree}. If~$\Tree$ is not \kl{accepted} by~$\Aut_{\suc E}$,
then there exists a \kl{submultiset} $F\submultisetneq E$ and a state~$g\in E\msminus
F$ such that $\Tree$ is \kl{accepted} by~$\Aut_{\suc F}$ and is \kl{rejected}
by~$\Aut_{\suc{F\uplus \{g\}}}$.
\end{proposition}

\begin{definition}
\AP A pair $(F,g)$ satisfying the conditions of Prop.~\ref{prop-bp} is
called a \intro{blocking pair} for~$\Tree$ and $E$.
\end{definition}

\begin{proof}[Proof of Prop.~\ref{prop-bp}]
The proof is by induction on~$\MSsize E$: the~result holds vacuously
if~$E$ is empty, and it is trivial if $\MSsize E=1$. It~then suffices to
observe that for any $g\in E$, either $(E\msminus\{g\},g)$ is a
\kl{blocking pair} for~$\Tree$, or such a \kl{blocking pair} can be found
in~$E\msminus \{g\}$ (by~induction).
\end{proof}

\AP
We now focus on \intro(BP){minimal} \kl{blocking pairs} for~$\Tree$, \ie,
pairs~$(F,g)$ such that $(F\msminus\{g'\},g)$ is not a \kl{blocking pair}
of~$\Tree$, for any~$g'\in F$.
We~will prove that if $(F,g)$ is a \kl(BP){minimal} \kl{blocking pair} for~$\Tree$,
then in any \kl(ET){accepting} \kl{execution tree}~$\ExTree$ of~$\Aut_{\suc F}$ on~$\Tree$,
the~\kl{direct subtrees}
that are not used to fulfill~$F$
can be accepted by no states in~$\supp(F)\cup \{g\}$.

\begin{proposition}\label{prop-reduceBP}
Let $(F,g)$ be a \kl{blocking pair} for some \kl{tree}~$\Tree$.
Let~$\ExTree=\tuple{\extree,\exlab}$ be a \kl{minimal} \kl(ET){accepting}
\kl{execution tree}
of~$\Aut_{\suc F}$ on~$\Tree$, $\nu$ be the
corresponding \kl{unitary marking}
of $\succ(\emptyw_\tree)$ by~$\State$,
satisfying $\nu\models+ \EUpair(F;\{\state_\top\})$.
If~there exists a node $y$ in $\succ(\varepsilon_{\tree})$ such that
(1)~$\nu(y)=q_\top$ and (2)~$\Tree_y \in
\Lang(\Aut_{g'})$ for $g'\in F$, then $(F\msminus \{g'\},g)$ is a
\kl{blocking pair} for~$\Tree$.
\end{proposition}

\begin{proof}
  Assume that such a node~$y$ exists.  The~\kl{execution
  tree}~$\ExTree$ witnesses the fact that $\Tree$ is \kl{accepted}
  by~$\Aut_{\suc{F\msminus \{g'\}}}$. If~$(F\msminus\{g'\},g)$ were
  not a \kl{blocking pair} for~$\Tree$, then $\Tree$~would be \kl{accepted} 
  by~$\Aut_{\suc{F\msminus \{g'\}\uplus\{g\}}}$.  In~that~case,
  let~$\ExTree'=\tuple{\extree',\exlab'}$ be a \kl{minimal} \kl(ET){accepting}
  \kl{execution tree} of~$\Aut_{\suc{F\msminus \{g'\}\uplus \{g\}}}$
  over~$\Tree$,
  and $\nu'$ be a \kl{unitary marking}
  of~$\succ(\emptyw_\tree)$ by~$\State$ induced by~$\ExTree'$, which satisfies
  $\nu'\models+ \EUpair(F\msminus \{g'\}\uplus\{g\};\{\state_\top\})$.
  Let~$N$ and~$N'$ be 
  subsets of~$\succ(\emptyw_\tree)$ such that
  $\nu(N)=F$ (hence~${y\notin N}$) and
  $\nu'(N')=F\msminus \{g'\}\uplus\{g\}$.
  We~pick~$\ExTree'$, $N$ and~$N'$ so as to maximize the size of~$N\cap N'$.

  Then $\size N=\size{N'}$, but $N\not\subseteq N'$: indeed, if
  $N\subseteq N'$, then $N=N'$, and $y\notin N'$; then $\ExTree'$
  could be extended with an \kl{execution tree} of~$\Aut_{g'}$ on the
  subtree~$\Tree_y$ entailing that $\Tree$ would be \kl{accepted}
  by~$\Aut_{\suc{F\uplus\{g\}}}$.
  Pick $n\in N\setminus N'$, and a corresponding
  node~$x\in\succ(\emptyw_{\extree})$ be such that
  $\exlab(x)=(n,\state)$ for some~$\state\in F$; then $\state\not=g'$,
  as otherwise $\ExTree'$ could again be extended into an \kl(ET){accepting}
  \kl{execution tree} of~$\Aut_{\suc{F\uplus\{g\}}}$ over~$\Tree$. Since
  $\nu'(N')=F\msminus\{g'\}\uplus\{g\}$,
  there must exist a
  node~$n'$ in~$N'\setminus N$ and a corresponding
  node~$x'\in\succ(\emptyw_{\extree'})$
  with $\exlab(x')=(n',q)$ (if~this were
  not the case, then $\nu(N)$ would contain more copies of~$q$
  than~$\nu'(N')$~does).
  Replacing the subtree rooted at~$x'$ in~$\ExTree'$
  with the subtree rooted at~$x$ in~$\ExTree$, we~get a \kl{minimal}
  \kl(ET){accepting}
  \kl{execution tree}~$\ExTree''=\tuple{\extree'',\exlab''}$
  of~$\Aut_{\suc{F\msminus\{g'\}\uplus\{g\}}}$
  over~$\Tree$
  with \kl{induced} \kl{unitary marking}~$\nu''$,
  and a set
  $N''\subseteq \succ(\emptyw_\tree)$
  such that $\nu''(N'')=F\msminus\{g'\}\uplus\{g\}$, but
  for which $\size{N\cap N''}$ is larger than~$\size{N\cap N'}$,
  contradicting
  our choice of~$\ExTree'$, $N$ and~$N'$.
\end{proof}

Now when some tree $\Tree$ is not \kl{accepted} by~$\Aut_{\suc E}$, we can ensure
the existence of a \kl{blocking pair}~$(F,g)$ such that for any node~$n$
in~$\succ(\emptyw_\Tree)$
that is not involved in the satisfaction of~$F$, the~subtree rooted at~$n$
is not
\kl{accepted} by any automaton~$\Aut_{g'}$, for $g'\in F \uplus \{g\}$.
Formally, we have:

\begin{proposition}
\label{prop-compl-bpm} %
If $\Tree$ is not \kl{accepted} by~$\Aut_{\suc E}$, then there exists
a \kl{blocking pair}~$(F,g)$ such that for any \kl{minimal}
\kl(ET){accepting}
\kl{execution tree}~$\ExTree$ of~$\Aut_{\suc{F}}$ over~$\Tree$,
writing~$\nu$ for the \kl{unitary marking} of~$\succ(\emptyw_\tree)$
\kl{induced} by~$\ExTree$, it~holds: for~any~$y \in \succ(\emptyw_\tree)$
with $\nu(y)=\state_\top$, the~subtree~$\Tree_y$ is \kl{rejected}
by~$\Aut_{g'}$ for any $g' \in F\uplus\{g\}$.
\end{proposition}

\begin{proof}
Consider a \kl{blocking pair}~$(F,g)$ for~$\Tree$ where $F\submultiset E$ is
minimal (for~inclusion).
Let~$\ExTree$ be a \kl{minimal} \kl(ET){accepting} \kl{execution tree}
of~$\Aut_{\suc F}$ over~$\Tree$, $\nu$~be~the \kl{induced} \kl{unitary marking}
of~$\succ(\emptyw_\tree)$, and
$N\subseteq \succ(\emptyw_\tree)$ such that $\nu(N)=F$.

For any node~$y\in  \succ(\emptyw_\tree)\setminus N$,
if the subtree~$\Tree_y$ were \kl{accepted} by~$\Aut_g$, then $\Tree$ would be
\kl{accepted} by~$\Aut_{\suc{F\uplus\{g\}}}$.
Similarly, 
if~the subtree~$\Tree_y$
were \kl{accepted} by~$\Aut_{g'}$, for~any~${g'\in F}$, then by
Prop.~\ref{prop-reduceBP}, ${(F\msminus \{g'\},g)}$ would be
a \kl{blocking pair}, contradicting minimality of~$(F,g)$.
\end{proof}

\smallskip

Following this result, we will build the complement automaton~$\CAut_{E}$
of~$\Aut_{\suc E}$ by checking the existence of a \kl{blocking pair} satisfying
the conditions of Prop.~\ref{prop-compl-bpm}: the~transition from the
initial state will be a disjunction, over all pairs~$(F,g)$,
of~\EUprs $\EUpair(F; \{{(\overline{\supp(F)}\cup\{\overline g\} ,\et)} \})$,
where $(\overline{\supp(F)}\cup\{\overline g\},\et )$ denotes a new
state \kl{accepting} all \kl{trees} that do not belong to~$\Lang(\Aut_{g'})$,
for any ${g'\in\supp(F)\cup \allowbreak \{g\}}$.
This~is expressed by the following \EUconstr:
\[
\Phi_{E} = \OU_{F \submultisetneq E\vphantom{g\msminus}} \;
  \OU_{g\in E\msminus F} \EUpair(F; \{ {(\overline{\supp(F)}\cup
    \{\overline g\},\et)} \}).
\]
 Notice that
for the special case where $E$ is empty, we~end up with an empty
disjunction, which is equivalent to false. This~is coherent with the
fact that $\Aut_\emptyset$ accepts any~\kl{tree}.
Full definitions and correctness proofs will be given after
we have explained how we handle
the general case of failing to satisfy~$\EUpair(E;U)$.

\subsubsection{Failing to satisfy $\protect\EUpair(E;U)$.}

\AP
In order to handle the general case, we~extend the definition
of~$\Aut_{\suc E}$ so~as~to also involve the universal part: 
the automaton~$\Aut_{\suc{\EUpair(E;U)}}$ is defined in the same way as~$\Aut_{\suc E}$, but now the initial state is $\state_{\EUpair(E;U)}$, and for any~$\alp\in\Alp$, we~let
$\Trans_{\EUpair(E;U)}(\state_{\EUpair(E;U)},\sigma)=\EUpair(E;U)$. 

We~have the following characterization of \kl{trees} not \kl{accepted}
by~$\Aut_{\suc{\EUpair(E;U)}}$:
\begin{proposition}\label{prop-negEU}
  If $\Tree$ is not \kl{accepted} by~$\Aut_{\suc{\EUpair(E;U)}}$, then
  \begin{itemize}
  \item either it has at least $\MSsize E+1$ \kl{direct subtrees}
    \kl{accepted} by no automata~$\Aut_u$ for~$u\in U$,
  \item or there exists $0\leq k\leq \MSsize E$ such that $\Tree$~has at least $k$
    \kl{direct subtrees} \kl{accepted} by no automata~$\Aut_u$ for any~$u\in
    U$, and it~does not contain $\MSsize E$ \kl{direct
    subtrees} witnessing the fact that it~is \kl{accepted}
    by~$\Aut_{\suc E}$ and of which~$k$~subtrees are \kl{accepted} by no
    automata~$\Aut_u$ for~$u\in U$.
  \end{itemize}
\end{proposition}

\begin{proof}
  We~first rule out the case where $E$~is empty: in~that case, $\Tree$~is not
  \kl{accepted} by~$\Aut_{\suc{\EUpair(E;U)}}$ if, and only~if, at least
  one of its \kl{direct subtrees} is \kl{accepted} by no automata~$\Aut_u$, for
  any~$u\in U$. 

  Now assume that $E$ is not empty; write~$l$ for the size of~$E$, and
  $m$ for the number of \kl{direct subtrees} of~$\Tree$. That~$\Tree$~is
  \kl{accepted} by~$\Aut_{\suc{\EUpair(E;U)}}$ means that there are $l$
  \kl{direct subtrees} each \kl{accepted} by some automaton~$\Aut_e$ for $e$
  ranging over~$E$, and that the remaining $m-l$ \kl{direct subtrees} are
  \kl{accepted} by some automaton~$\Aut_u$, for some~$u\in U$. This means
  that for any~$k$, if there are at~least~$k$~\kl{direct subtrees} not
  \kl{accepted} by~$\Aut_u$ for any~$u\in U$, then there are $l$
  \kl{direct subtrees}, each \kl{accepted} by some automaton~$\Aut_e$ for~$e$
  ranging over~$E$, and of which at least $k$ are \kl{accepted} by no
  automaton~$\Aut_u$ for any~$u\in U$; this simply expresses the fact
  that those \kl{direct subtrees} \kl{accepted} by no automata~$\Aut_u$ for
  any~$u\in U$ must be used to fulfill the $E$-part of~$\EUpair(E;U)$.
 
  By duality, we~get that $\Tree$~is not \kl{accepted}
  by~$\Aut_{\suc{\EUpair(E;U)}}$ if, and only~if, there exists some
  integer~$k$ such that there are at least $k$ \kl{direct subtrees}
  \kl{accepted} by no automata~$\Aut_u$ for any~$u\in U$, and there do
  not exist~$l$~\kl{direct subtrees}, each \kl{accepted} by some
  automaton~$\Aut_e$ for~$e$ ranging over~$E$, and of which at least
  $k$ are \kl{accepted} by no automata~$\Aut_u$ for any~$u\in U$.
\end{proof}

\subsubsection{Composite states and parity conditions.}
The definition of $\Phi_{E}$ gives rise to composite states of the
form~$(P,\et)$ where some elements of~$P$ may be of the
form~$(P',\ou)$. As~their names indicate, these states behave as
conjunctions (resp.~disjunctions) of several states (or their dual) of
the original automaton.

While alternation can naturally be used to encode the conjunctive or
disjunction nature of these states in the transition function,
defining the parity conditions for such groups of states is more
problematic.

We~address this issue in the same way as in
Section~\ref{remUsingleton}: we~replace each state~$\state$ with a
pair of states of the form~$(\state',\state)$, where $\state'$ (in~$\State$ or~$\Statebar$) is
either the parent state (if~it was not composite), or one state in the
(composite) parent state, namely the one whose transition function
involves~$\state$.
The first component~$\state'$
is used for the acceptance condition, while the
second component~$\state$ is the current (possibly composite) state, and is used for the
transition function.

\subsubsection{Construction of the complement automaton.}

We~now define the complement automaton~$\Aut^c$
of~$\Aut=\tuple{\State,\initstate,\Trans,\omega}$; we~let $\Aut^c=
\tuple{\State^c,\initstate^c,\Trans^c,\omega^c}$ be such that:
\begin{itemize}
\item the set of states $\State^c$ is a subset of $\DState \times \Big(\DState \cup
  (2^{\DState}\times\{\et\}) \cup
  (2^{(2^\DState\times\{\ou\})}\times\{\et\})\Big)$, where $\DState =
  \State\cup\Statebar$, and $\Statebar=\{\overline\state \mid
  \state\in \State\}$ is a set of fresh states. This~defines an
  operator~$\state\mapsto \overline \state$ over~$\State$, which we~extend to
  states~$\overline x$ of $\Statebar$ by letting~${\overline{\overline
    x}= x}$, to~subsets~$X$ of~$\DState$ by letting $\overline X =
  \{\overline x\mid x\in X\}$, to~states~$(P,\et)$ 
  of~$2^{\DState}\times\{\et\}$ by letting $\overline{(P,\et)} = 
  (\overline P,\ou)$.
  We~finally extend it to~$\PBF{\State^c}$ by letting
  \begin{xalignat*}4
    \overline{\top} &= \bot &
    \overline{\bot} &= \top &
    \overline{\psi \et \phi} &= \overline{\psi} \ou \overline{\phi} &
    \overline{\psi \ou \phi} &= \overline{\psi} \et \overline{\phi}.
  \end{xalignat*}

  Intuitively, from a state~$(\state',\state)$ with $q \in \State$, automaton~$\Aut^c$ will
  accept the same \kl{language} as from the state~$\state$
    in~$\Aut$, while
  from a state~$(\state',\overline{\state})$, 
  it~will accept its
  complement \kl{language}.
  The~\kl{language accepted} by~$\Aut$ from a state~${(\state,(P,\et))\in
    \DState\times (2^{\DState}\times\{\et\})}$
  will be the intersection of the \kl{languages
  accepted}
  by~$\Aut$ from all states in~$P$, whereas the \kl{language accepted}
  from states of the form~$(\{(\state,(P_i,\ou))\mid 1\leq i\leq
  k\},\et)\in 2^{\DState\times (2^\DState\times\{\ou\})}\times\{\et\}$ 
  will be the intersection (over~$i$) of the
  unions of the \kl{languages accepted} from all states in~$P_i$. Notice that $P_i$ may be a singleton.
  Notice also  that the indications of~$\et$ and~$\ou$ are mainly used for the
  sake of clarity.
  
\item accordingly, $\initstate^c = (\overline{\initstate},\overline{\initstate})$, as we want~$\Aut^c$ to
  \kl{accept} the complement of the \kl{language accepted} by~$\Aut$ from~$\initstate$;

\item $\Trans^c$ is defined as follows:
first, for states in $\DState\times\DState$:
\begin{xalignat*}2
\Trans^c((u,\state),\alp)  &= \phi_\state(\Trans(\state,\alp))   &
\Trans^c((u,\overline{\state}),\alp) &=  \phi_{\overline{\state}}(\overline{\Trans(\state,\alp)}) 
\end{xalignat*}
where $\phi_\state$ is the mapping similar to the one used in Section~\ref{remUsingleton}, defined as
$\phi_\state(\state')=(\state,\state')$ and extended to boolean combinations and \EUprs in the natural way.

The operation $\overline{\Trans(\state,\alp)}$ gives rise to formulas
of the form~$\overline{\EUpair(E;U)}$. Before we explain how we
express these formulas with plain \EUprs, we~define the transitions
for composite states (which will appear in the development of
formulas~$\overline{\EUpair(E;U)}$); composite states represent
conjunctions or disjunctions of other (possibly composite) states, hence their
transition function is defined inductively:
\begin{xalignat*}2
\Trans^c( (u,(P,\et)),\alp)  &= \ET_{r\in P} \Trans^c((u,r),\alp)  &
\Trans^c( (u,(P,\ou)),\alp)  &= \OU_{r\in P} \Trans^c((u,r),\alp).
\end{xalignat*}

Now, for formulas of the form~$\overline{\EUpair(E;U)}$,
following Prop.~\ref{prop-negEU}, we~let:
\[
\overline{\EUpair(E;U)} =  \EUpair((\overline{U},\et) \mapsto \MSsize{E}+1;\{\state_\top\})  \ou {} %
\OU_{0\leq k \leq \MSsize E}
  \Bigl(  \EUpair((\overline{U},\et) \mapsto k;\{\state_\top\}) \: \et\:
  \ET_{\stackrel{m \submultiset E}{\MSsize m=k}}  \Phi_{(E\msminus m) \mscup m^{\overline U}} \Bigr)
\]
where $m^{\overline{U}}$ is the \kl{multiset} defined by $\{(\{x\}\cup
\overline{U},\et)\mapsto m(x)\}_{x\in\supp(m)}$, and $\Phi_{G_m}$ (defined
as~in the paragraph following Prop.~\ref{prop-compl-bpm}) characterises the
failure to satisfy~$\EUpair(G_m;\{q_\top\})$, thus the existence of a \kl(BP){minimal}
\kl{blocking pair}:
\[
\Phi_{G_m} \eqdef \OU_{F \submultisetneq G_m\vphantom{g\msminus}} \;
  \OU_{g\in G_m\msminus F} \EUpair(F; \{ {({\overline{\supp(F)}\cup \{\overline g\}},\et)}\}).
\]

Since $\Phi_{G_m}$ is used in~$\overline{\EUpair(E;U)}$ with $G_m=(E\msminus m)\mscup m^{\overline U}$,
the sets~$F$ and the state~$g$ in the disjunctions of $\Phi_{G_m}$ may involve states 
of the form~$(P, \et)$ with~${P\subseteq \DState}$.
Then~$\overline{\supp(F)}$ and $\overline g$ may give rise to states of the form~$\overline{(P,\et)}$, 
which we~rewrite as $(\overline P,\ou)$.

In~the~end, the (second components of~the) states involved
in~$\Trans^c$ are either states of~$\DState$, or conjunctions of states
of~$\DState$, or conjunctions of disjunctions of states in~$\DState$.
Hence our definition of~$\State^c$.
Finally, observe that, as claimed in Remark~\ref{remUsingleton},
this construction only introduces new \EUprs of the
form~$\EUpair(E;\{u\})$.

\item priorities are based on the first component of the pairs and are defined as follows
  \begin{xalignat*}2
    \prio^c((\state,\state')) &= \prio(\state) &
    \prio^c((\overline{\state},\state')) &=\prio(\state)+1
  \end{xalignat*}
for any two states $(\state,\state')$  and
$(\overline\state,\state')$ with $\state\in\State$.
\end{itemize}

\subsubsection{Examples.}
  As examples to  illustrate this construction, consider the negated \EUprs of the form
  $\overline{\EUpair(q \mapsto 1; \{\state_\top\})}$ and
  $\overline{\EUpair(\emptyset; \{\state\})}$, which corresponds to
  transitions~$\Diamond q$ and~$\Box q$
  in~\BDAutomata~\cite{Wil99a}.
  \smallskip
  
  The~negated \EUpr $\overline{\EUpair(q \mapsto 1; \{\state_\top\})}$ gives
  rise (after applying $\phi_{\state'}$ to obtain pairs of states)
  to the disjunction of formulas of the form 
  \begin{itemize}
  \item $\EUpair((\state',(\overline{\state_\top},\et))\mapsto 2;
    {\{(\state',\state_\top)\}})$; since $\Trans(\state_\top,\alp)=\top$
    (see~Section~\ref{ssec-ex-aut}), we~get
    $\Trans^c((\state',\overline{\state_\top}),\alp)=\bot$, hence this will
    never give rise to accepting execution trees;

  \item two formulas of the form
    \[
      {\EUpair((\state',(\overline{\state_\top},\et))\mapsto k; {\{(\state',\state_\top)\}})} \et
    \ET_{m\submultiset \{\state\}; \MSsize m=k}
    \phi_{\state'}(\Phi_{(\{\state\}\msminus m)\uplus m^{\overline U}}),
    \]
    for~$k\in\{0,1\}$. As~above, no states involving~$\overline{\state_\top}$ can appear in an
    \kl(ET){accepting} \kl{execution tree}, so that we can eliminate the case where $k=1$.
    For~$k=0$, we~end~up with
    $\phi_{\state'}(\Phi_{\{\state\}})$, which simplifies
    to~$\EUpair(\emptyset; {\{(\state',(\overline \state,\et))\}})$. This corresponds to
    formula~$\Box\overline\state$, which intuitively means that all
    \kl{direct subtrees} must be \kl{rejected} by~$\Aut_\state$, as expected.
  \end{itemize}

  Similarly, $\overline{\EUpair(\emptyset; \{\state\})}$ gives rise to the
  disjunction of 
  $\EUpair((\state',(\overline\state,\et)) \mapsto 1; {\{(\state',\state_\top)\}})$ and
  $\EUpair((\state',(\overline\state,\et)) \mapsto 0; {\{(\state',\state_\top)\}}) \et
  \phi_{\state'}(\Phi_{\emptyset})$; as~already seen above, $\Phi_\emptyset$~is equivalent to~false, and
  cannot give rise to an  \kl(ET){accepting} \kl{execution tree}.
  Hence $\overline{\EUpair(\emptyset; \{\state\})}$ gives rise to
  $\EUpair((\state',\overline\state) \mapsto 1; {\{(\state',\state_\top)\}})$, which
  corresponds to $\Diamond\overline\state$.

\subsubsection{Size of $\Aut^c$.}
We now evaluate the \kl(A){size} of the complement automaton:
\begin{itemize}
\item \textit{(over)approximating the number of states:}
  the first components of the states in~$\State^c$ are in~$\DState$; 
  their second components are in $\DState \cup (2^{\DState}\times\{\et\}) \cup
  (2^{(2^\DState\times\{\ou\})}\times\{\et\})$, but not all states of
  this set are reachable.  The~reachable states are either
  in~$\DState$, or they appear in~$\overline{\EUpair(E;U)}$ for some
  \EUpr~${\EUpair(E;U)}$ in~$\Trans$.
  
  Take an \EUpr $\EUpair(E;U)$ in~$\Trans$.
  Besides~$(\overline{U},\et)$ and~$\state_\top$, the~states invovled in
  $\overline{\EUpair(E;U)}$ all come from formulas of the form~$\Phi_{G_m}$, with
  $G_m=(E\msminus m)\uplus m^{\overline U}$,
  for each \kl{submultiset}~$m$ of~$E$.
  Formula~$\Phi_{G_m}$ involves
  \begin{itemize}
  \item states of~$G_m$, which are either in~$E$ (hence already
    in~$\State$) or of the form $(\{x\}\cup \overline U,\et)$
    for~$x\in E$;
  \item and states of the form $(\overline{\supp(F)}\cup\{\overline
    g\},\et)$ where $F$~is a \kl{submultiset} of~$G_m$ and $g\in
    G_m\msminus F$.
  \end{itemize}
  \medskip
  
  In~the~end, using $\size\State\cdot\size\Alp \cdot \sizeB\Trans$ as
  an upper bound on the number of \EUprs,

  \begin{center}
  \newcommand{\tikzmark}[2]{\tikz[overlay,remember picture]
    {\node[inner sep=0pt](tmp) {\hphantom{$#2$}};
    \node[minimum height=6mm,inner sep=0pt] at (tmp.east) (#1) {\hphantom{$#2$}};}#2}
  \(
    \size{\State^c}\leq \tikzmark{First}{2\size\State}\times\Bigl[\tikzmark{2Q}{2\size\State}+ \tikzmark{qtop}{1}
    + \tikzmark{deltaB}{\underbrace{\size\State\cdot \sizeB\Trans\cdot\size\Alp}} \cdot (\tikzmark{Ubar}{1} +
    \tikzmark{xUbar}{2\!\cdot\!\sizeE\Trans} + \tikzmark{2^E}{3^{\sizeE\Trans}})\Bigr]\quad
  \)
    \begin{tikzpicture}[overlay,remember picture]
      \path (2Q) -- node[pos=1,below] (2Q') {$\State\cup\overline\State$}  ++(-1.5,-1);
      \draw (2Q) edge[latex'-,bend left=10] (2Q');
      \draw (First) edge[latex'-,bend left=10] (2Q');
      \path (qtop) -- node[pos=1,below] (qtop') {$\state_\top\vphantom{\overline\State}$}  ++(-1.2,-1) ;
      \draw (qtop.south west) edge[latex'-,bend left=10] (qtop');
      \path (deltaB) -- node[pos=1,below,text width=2cm,align=center] (deltaB') {number of \EU-pairs}  ++(-1.3,-1) ;
      \draw (deltaB.-90) edge[latex'-] (deltaB');
      \path (Ubar) -- node[pos=1,below,text width=1cm,align=center] (Ubar') {state $(\overline U,\et)$}  ++(-1.25,-1) ;
      \draw (Ubar) edge[latex'-] (Ubar');
      \path (xUbar) -- node[pos=1,below,text width=2cm,align=center] (xUbar') {states $(\{x\}\cup\overline U,\et)$, $(\{\overline x\}\cup U,\ou)$}  ++(-0.5,-1) ;
      \draw (xUbar) edge[latex'-] (xUbar');
      \path (2^E) -- node[pos=1,below,text width=3cm,align=center] (2^E') {states $(\overline{\supp(F)}\cup\{\overline g\},\et)$}  ++(1.2,-1) ;
      \draw (2^E) edge[latex'-,bend right=15] (2^E');
    \end{tikzpicture}
  \vskip18mm\vbox{}
  \end{center}

\medskip

  Hence $\size{\State^c} \leq
    2\size\State\times \bigl[1+\size\State\cdot (2+ \size\Alp\cdot\sizeB\Trans\cdot
    (1+2\cdot\sizeE\Trans+3^{\sizeE\Trans}))\bigr]$.  %

\item \textit{bounding the size of transition function:}
in order to evaluate the size of the transition function,
  we~group the transitions according to the type~of (the~second component~of) their source
  states. We~write $\Trans^c_{\scriptscriptstyle\State}$, $\Trans^c_{\scriptscriptstyle\overline{\State}}$,
  $\Trans^c_{\scriptscriptstyle(\overline{U},\et)}$,
  $\Trans^c_{\scriptscriptstyle({\overline{\supp(F)}\cup  \{\overline g\}},\et)}$,
  and
    $\Trans^c_{\scriptscriptstyle F}$
  for the five
  categories, which we describe and bound below:
\begin{itemize}
\item $\Trans^c_{\scriptscriptstyle\State}$ is the set of transitions from states in~$\State$.
  We~have $\sizeB{\Trans^c_{\scriptscriptstyle\State}} = \sizeB{\Trans}$, and
  $\sizeE{\Trans^c_{\scriptscriptstyle\State}} = \sizeE{\Trans}$ and
  $\sizeU{\Trans^c_{\scriptscriptstyle\State}} = \sizeU{\Trans}$,
  since
  $\Trans^c(\state,\alp)=\Trans(\state,\alp)$;

\item $\Trans^c_{\scriptscriptstyle\overline{\State}}$ contains all transitions from
  states in~$\overline\State$;
  writing $\Trans(\state,\alp)=\psi(\EUpair(E_i;U_i)_i)$ for some positive boolean combination~$\psi$,
  we~have 
  $\Trans^c(\overline\state,\alp)=\overline\psi(\overline{\EUpair(E_i;U_i)}_i)$.
  Each $\overline{\EUpair(E_i;U_i)}$ gives rise to at most $1+\size{E_i}(1+2^{\size{E_i}}\cdot 2^{\size{E_i}}\cdot \size{E_i})$ \EUprs.   Hence the boolean size~$\sizeB{\Trans^c_{\scriptscriptstyle{\overline{\State}}}}$
  of~$\Trans^c_{\scriptscriptstyle\overline{\State}}$ is at most
  $\sizeB\Trans\cdot (1+\sizeE\Trans\cdot (1+\sizeE\Trans\cdot2^{2\sizeE\Trans}))$.
  Moreover, we~have $\sizeE{\Trans^c_{\scriptscriptstyle\overline{\State}}} \leq
  \sizeE\Trans+1$ and $\sizeU{\Trans^c_{\scriptscriptstyle\overline{\State}}} \leq
  \max(\sizeU\Trans, 1)$.

\item transitions in $\Trans^c_{\scriptscriptstyle(\overline{U},\et)}$ originate from
  states of the form~$(\overline U,\et)$. Since $U\subseteq \State$,
  those transitions~are conjunctions of at~most~$\size\State$
  transitions in~$\Trans^c_{\scriptscriptstyle\overline\State}$.  Hence
  $\sizeB{\Trans^c_{\scriptscriptstyle(\overline{U},\et)}}\leq
  \size\State\cdot \sizeB{\Trans^c_{\scriptscriptstyle\overline{\State}}}$,
  and
  $\sizeE{\Trans^c_{\scriptscriptstyle(\overline{U},\et)}}\leq
  \sizeE{\Trans^c_{\scriptscriptstyle\overline{\State}}}$ and
  $\sizeU{\Trans^c_{\scriptscriptstyle(\overline{U},\et)}}\leq
  \sizeU{\Trans^c_{\scriptscriptstyle\overline{\State}}}$.

  \item 
  transitions in $\Trans^c_{\scriptscriptstyle(\overline{\supp(F)}\cup \{\overline
    g\},\et)}$ originate from states of the form $(\overline{\supp(F)}\cup \{\overline g\},\et)$.
  From each such
    state, transitions are conjunctions of at most $\size\State+1$
    subformulas, which in the worst case can be disjunctions of
    at most $1+\sizeU\Trans$
    transitions from states
    in~$\overline\State$
    (corresponding to states belonging to $2^{(2^\DState\times\{\ou\})}\times\{\et\}$). Hence we have
    $\sizeB{\Trans^c_{\scriptscriptstyle (\overline{\supp(F)}\cup \{\overline g\},\et)}}\leq
       (\size\State+1)\cdot \sizeE\Trans\cdot
      \sizeB{\Trans^c_{\scriptscriptstyle \overline\State}}$,
    and
    $\sizeE{\Trans^c_{\scriptscriptstyle(\overline{\supp(F)}\cup \{\overline
    g\},\et)}}$
    and
    $\sizeU{\Trans^c_{\scriptscriptstyle(\overline{\supp(F)}\cup \{\overline
    g\},\et)}}$ are bounded by
    $\sizeE{\Trans^c_{\scriptscriptstyle \overline\State}}$
    and
    $\sizeU{\Trans^c_{\scriptscriptstyle \overline\State}}$, respectively.

  \item finally,
  $\Trans^c_{\scriptscriptstyle F}$ contains the
    transitions from states in~$F\submultisetneq G_m$ in~$\Phi_{G_m}$. Since
    $\Phi_{G_m}$ is used with~$G_m$ of the form $E\msminus m\uplus
    m^{\overline U}$, the~transition in~$\Trans^c_{\scriptscriptstyle
      F}$ are either transitions from states in~$E$, or transitions
    from states of the form~$(\{x\}\cup\overline U,\et)$. The~former
    case corresponds to transitions already in~$\Trans$; the~latter case gives rise to
    conjunctions of at most~$\sizeU\Trans$ transitions
    in~$\Trans^c_{\scriptscriptstyle \overline\State}$ and one transition in~$\Trans$.
    Thus $\sizeB{\Trans^c_{\scriptscriptstyle F}} \leq \sizeU\Trans\cdot
    \sizeB{\Trans^c_{\scriptscriptstyle\overline{\State}}}+\sizeB\Trans$, and
    $\sizeE{\Trans^c_{\scriptscriptstyle F}} \leq \max(\sizeE{\Trans},
    \sizeE{\Trans^c_{\scriptscriptstyle \overline\State}})$ and
    $\sizeU{\Trans^c_{\scriptscriptstyle F}} \leq \max(\sizeU{\Trans},
    \sizeU{\Trans^c_{\scriptscriptstyle \overline\State}})$.

\end{itemize}
We end up with
\begin{itemize}
\item $\sizeB{\Trans^c}\leq (1+\size\State)\cdot (1+\sizeU\Trans)\cdot \sizeB\Trans\cdot
  (1+\sizeE\Trans+\sizeE\Trans^2\cdot 2^{2\sizeE\Trans})$,
\item $\sizeE{\Trans^c}\leq \sizeE\Trans+1$,
\item $\sizeU{\Trans^c}\leq \max(\sizeU\Trans,1)$.
\end{itemize}

\end{itemize}

\medskip
Note that the complementation operation over \kl{alternating} \faTAutomata does not induce such a complexity blow-up:
the~construction is performed by using the dual of the transition
function and by incrementing the \kl{priority} of all states. This is an
important difference compared to our construction, which is due to the
expressiveness (and succinctness) of~\AATA: for~example, we~can easily
express that there are at~least $k$ \kl{successors} that are \kl{accepted} by
some state~$q$ with the constraint $\EUpair(q \mapsto k;\{q_\top\})$;
this cannot be expressed with \BDAutomata, and it
requires a much more complex formula in \faTAutomata
(for~each arity~$d$, we~have to consider all possible subsets of
\kl{successors} of size~$k$).

\subsubsection{Correctness proof.}
 We can now state and prove the  correctness of the construction:
\begin{proposition}
For any \SDtree{\Alp}{\Dir} $\Tree=(t,l)$, any node~$n$ in~$t$, any state $u \in \State\cup\overline{\State}$ and any
state $q$ in $\State$, we have:
\[
  \Tree_n \notin \calL(\calA,q)
\quad \Leftrightarrow \quad
  \Tree_n \in\calL(\calA^c,(u,\overline{q}))   
\]
\end{proposition}

\begin{proof}
  We use the \kl{game semantics} in order
  to prove both implications.  In~order to avoid confusions,
  we~name \Pl0 and~\Pl1 the players in~$\GAT$, and \Plc0 and \Plc1 the
  players in~$\GATc$.

  \medskip

  We take a \kl{tree}~$\Tree=\tuple{\tree,\lab}$, a node~$n\in \tree$,
  and a state~$\state$ of~$\Aut$ such that $\Tree_n$ is not \kl{accepted}
  by~$\Aut_\state$. We~will prove that~$\Tree_n$ is \kl{accepted} by
  $\Aut^c_{(u,\overline\state)}$ for any $u$.

  By~Prop.~\ref{prop-determined} and~\ref{prop-gamesem}, \Pl1 has
  a \kl{memoryless} \kl(S){winning} \kl{strategy}~$\str1$
  from~$(n,\state)$ in~$\GAT$.  Using this strategy,
  we~build a \kl{memoryless} \kl{strategy}~$\strc0$ for~$\Plc0$
  from~$(n,(u,\overline\state))$ in~$\GATc$, which we then prove is winning.
  By~Prop.~\ref{prop-gamesem}, this entails our result.

  From vertex~$(n,(u,\overline\state))$,
  there is only
  one edge, to~$(n,\Trans^c((u,\overline\state),
  \lab(n)))$. By~definition, $\Trans^c((u,\overline\state), \lab(n))$ is
  of the form~$\phi_{\overline\state}(\overline{\psi})$,
  with  $\psi=\Trans(\state,\lab(n))$.
    Moreover, \Pl1 wins from vertex~$(n,\psi)$.
  Using the \kl{memoryless} \kl{strategy}~$\sigma_1$ of \Pl1 from~$(n,\psi)$,
  we~define a \kl{strategy} of \Plc0 from $(n,\phi_{\overline{\state}}(\overline\psi))$, by considering the different possible forms of~$\psi$:
\begin{itemize}
\item if $\psi=\bot$, then $(n,\Trans^c((u,\overline\state),
  \lab(n)))$ is the state $(n,\top)$, which belongs to \Plc1 (and 
  has no outgoing transitions) and then $\Plc0$ wins. Notice that we cannot have $\psi=\top$ since $\Pl1$ does not have a \kl(S){winning} \kl{strategy} from~$(n,\top)$;
\item if $\psi = \OU_i \psi_i$,  
  then $\phi_{\overline{\state}}(\overline\psi)$
  is a conjunction, so that $(n,\phi_{\overline{\state}}(\overline\psi))$
  belongs to~$\Plc1$. For every $i$, $\Pl1$ has a winning strategy from $(n,\psi_i)$; this~ensures that for every choice made by~$\Plc1$, $\Plc0$~has a winning strategy from $(n,\phi_{\overline{q}}(\overline{\psi_i}))$;
\item if $\psi = \ET_i \psi_i$, then $\phi_{\overline{\state}}(\overline\psi)$  is a disjunction,
  and $(n,\phi_{\overline{\state}}(\overline\psi))$ belongs to~$\Plc0$.  Her~\kl{strategy} consists
  in following \Pl1's move from~$(n,\psi)$ in~$\GAT$: letting~$\psi_i$  be such that
  $\str1(n,\psi)=(n,\psi_i)$, we~define $\strc0(n,\phi_{\overline{\state}}(\overline\psi))=(n,\phi_{\overline{\state}}(\overline{\psi_i}))$;
\item if $\psi = \EUpair(E;U)$:  as $\Pl1$ has a \kl(S){winning} \kl{strategy} from~$(n,\EUpair(E;U))$,
  we know that for \emph{any} possible move (chosen by~$\Pl0$) of the form $(n,\nu)$ where $\nu$ is a \kl{marking} of~$\succ(n)$ by states in~$\supp(E)\cup U$ satisfying
  $\mkimg\nu\models+\EUpair(E;U)$, $\Pl1$~can choose a winning
  state $(n',\nu(n'))$. This means that there is no way to
  satisfy $\EUpair(E;U)$ with the \kl{successors} of~$n$, and by Prop.~\ref{prop-negEU},
  this is due to one (or~both) of the following cases:
\begin{itemize}
\item there are (at~least) $\MSsize E+1$ nodes in $\succ(n)$ \kl{rejected} by all states
  in~$U$. Let~$X$ be such a set of $\MSsize E+1$ nodes. Therefore from the node $(n,\phi_{\overline{\state}}(\overline{\EUpair(E;U)}))$, 
   $\Plc0$ may choose the first disjunct in the definition of
  $\overline{\EUpair(E;U)}$, move to $(n,\EUpair( (\overline{\state},(\overline{U},\et)) \mapsto k+1 ;\{\state_\top\}))$, and from there, select the move leading to
  $(n,\nu_c)$, with $\nu_c(n') = (\overline{\state},(\overline{U},\et))$ if $n' \in X$ and
  $\nu_c(n')=(\overline{\state},q_\top)$ otherwise. From~$(n,\nu_c)$,
  any~\kl{successor} (chosen by~$\Plc1$) will
  be of the form $(n',(\overline{\state},(\overline{U},\et)))$ with $n'\in X$ and  $q\in U$, or
  $(n',(\overline{\state},q_\top))$ with $n'\notin X$. If~$\Plc1$ chooses some $(n',(\overline{\state},q_\top))$, then $\Plc0$~wins; if $\Plc1$ chooses some $(n',(\overline{\state},(\overline{U},\et)))$, the~game continues with $(n',\Trans^c((\overline{q},(\overline{U},\et)),\ell(n')))$ which is defined as $(n',\allowbreak \ET_{r\in U} \phi_{\overline{r}}(\overline{\Trans(r,\ell(n'))}))$. For every choice of~$\Plc1$, $\Plc0$~will have a winning strategy, as by~assumption we~know that $n' \in X$, and then that it~is rejected by any~state~${r\in U}$.
  
  Concerning the parity condition, the sequence of visited nodes has the priority of~$u$ for $(n,(u,\overline{q}))$, the~priority of~$\overline{q}$ for $(n',(\overline{q},(\overline{U},\et)))$, the maximal priority  for states associated with formulas of the transition function, and we will also visit the priority of~$\overline{r}$ thanks to the prefixing with~$\phi_{\overline{r}}$.
Thus for this sequence, the~significant states (w.r.t.\ the satisfaction of the parity condition) are~$u$, $\overline{q}$, and~$\overline{r}$. 

\item there exists
  an integer~$k$, with $0 \leq k \leq |E|$, such that there are
  at~least~$k$ \kl{successors} of~$n$ that are \kl{rejected} by all
  states in~$U$, and there is no way to satisfy~$E$ with any set of
  \kl{successors} containing at~least~$k$ nodes \kl{rejected} by all states
  in~$U$.  Take such an integer~$k$, and a set~$X$ of~$k$ \kl{successors}
  of~$n$ \kl{rejected} by all states in~$U$. Then $\Plc0$ can choose the
  subformula corresponding to~$k$ in the second part of the formula defining
  $\phi_{\overline{q}}(\overline{\EUpair(E;U)})$. This~subformula is a conjunction, so that
  $\Plc1$ decides which \kl{successor} to move~to:
  \begin{itemize}
    \item if $\Plc1$ decides to move to state~$(n,\EUpair((\overline{q},(\overline{U},\et))
      \mapsto k;{\{(\overline{q},\state_\top)\}}))$, then $\Plc0$ can move to~$(n,\nu)$
      where $\nu$ is the \kl{marking} mapping the $k$ nodes in~$X$
      to~$(\overline{q},(\overline{U},\et))$, and the other nodes to~$(\overline{q},\state_\top)$. In~both~cases, $\Plc0$ wins (as~in the previous case). 

\item if $\Plc1$ chooses $(n,
      \phi_{\overline{q}}(\Phi_{(E\msminus m) \uplus
      m^{\overline U}}))$ for some set~$m$ of $k$~states, then $\Plc0$
      can  choose a \kl(BP){minimal} \kl{blocking pair}~$(F,g)$
      (for~$(E\msminus m) \uplus m^{\overline
      U}$ by~Prop.~\ref{prop-bp} and~Prop.~\ref{prop-compl-bpm}) in
      the subsequent disjunctive states, such that all \kl{successor
      nodes} not used to satisfy~$F$ are \kl{rejected} by all states
      in~$\supp(F)\cup\{g\}$. $\Plc0$ can then choose a node
      $(n,\nu)$ where $\nu$ is a \kl{unitary
      marking} witnessing the satisfaction of
      $\EUpair(\phi_{\overline{q}}(F);{\{(\overline{q},(\overline{\supp(F)}\cup\{\overline{g}\},\et))\}})$. Looking
      closely at $F$, we see that this multiset may contain states in
      $\State$ or of the form $(\{q'\}\cup \overline{U},\et)$. Every
      successor $(n',\nu(n'))$ is:
      \begin{itemize}
      \item either used to satisfy $\phi_{\overline{q}}(F)$ and then 
        $\nu(n')$ is either of the form  $(\overline{q},q')$ or $(\overline{q},(\{q'\}\cup \overline{U},\et))$. In both case, $\Plc0$ has a winning strategy. Note that from nodes of the form $(n', (\overline{q},(\{q'\}\cup \overline{U},\et)))$, $\Plc1$ will be able to choose which component to follow among $\{q'\}\cup\overline{U}$, in every case $\Plc0$ will keep her winning strategy (and note that the successors  will finally be prefixed by the corresponding state in $\State\cup\overline{\State}$  and then the parity function in $\GATc$ will be based on the correct priorities).
\item or associated with $(\overline{q},(\overline{\supp(F)}\cup\{\overline{g}\},\et))$. In that case, the alternating transition function will integrate the conjunction of the composite state and the  disjunctions of any state of the form $(\overline{q'}\cup U,\ou)$ inside. Again the prefixing will keep track of the actual states in $D$ used to satisfy the composite states, and then the parity function is also correct. 
      \end{itemize}
  \end{itemize}
\end{itemize} 
\end{itemize}
The~resulting \kl{memoryless} \kl{strategy} for
\Plc0 is \kl(S){winning}. Indeed consider an infinite game in $\GAT$: the parity acceptance (winning for $\Pl1$) is based
on the infinite sequence of configurations of the form $(n,q)$, and the corresponding game in $\GATc$ will have an infinite sequence of configurations $(n,(q,-))$ and it will determine the parity condition (winning for $\Plc0$) because other intermediate configurations have greater priorities.

\medskip

Conversely, we now show how to define a \kl(S){winning} \kl{strategy}  for $\Pl1$ in the game $\GAT$ from a \kl{memoryless} \kl(S){winning} \kl{strategy}  for~$\Plc0$ in~$\GATc$. 
The main case remains the correspondence between a node $(n,\EUpair(E;U))$ and a node $(n,\phi_{\overline{q}}(\overline{\EUpair(E;U)}))$. As~$\Plc0$~has a \kl(S){winning} \kl{strategy}, she~can choose some term in the disjunction $\phi_{\overline{q}}(\overline{\EUpair(E;U)})$;
there~are two~cases:
\begin{itemize}
\item if $\Plc0$ chooses the move leading to
  $(n,\EUpair((\overline{q},(\overline{U},\et)) \mapsto
  \MSsize{E}+1;{\{(\overline{q},\state_\top)\}}))$, and then a move $(n,\nu_c)$ where
  $\nu_c$ maps $|E|+1$ nodes in $\succ(n)$ to the state
  $(\overline{q},(\overline{U},\et))$: we~then let~$Y$ be this set of nodes. Note that
  $\Plc0$ has a \kl(S){winning} \kl{strategy} from any node $(y,(\overline{q},\overline{q'}))$ for
  any $q' \in U$ and $y\in Y$. Now consider the node $(n,\EUpair(E;U))$
  in $\GAT$. Every move of $\Pl0$ leads to a node of the form
  $(n,\nu)$ where $\nu$ is a mapping from $\succ(n)$ to $\supp(E)\cup
  U$ in order to satisfy the constraint $\EUpair(E;U)$. For~at~least
  one $y \in Y$, we~have $\nu(y) \in U$ (only~$k$~nodes are used to
  fulfil~$E$); then $\Pl1$ can select the move to $(y,\nu(y))$ in
  order to keep simulating the (\kl(S){winning}) \kl{strategy} of $\Plc0$ from
  $(y,(\overline{q},\overline{\nu(y)}) )$.
\item if $\Plc0$ chooses a term of the disjunction corresponding to some
  $k$ (with $0\leq k \leq |E|$): then $\Plc0$ has a \kl(S){winning} \kl{strategy}
  both from $(n,\EUpair((\overline{q},(\overline{U},\et)) \mapsto
  k;{\{(\overline{q},\state_\top)\}}))$ and from any node of the form $(n,$ $\phi_{\overline{q}}(\Phi_{(E\msminus m) \uplus
    m^{\overline{U}}}))$ with $m \submultiset E$ with
  $\MSsize{m}=k$. Let~$Y$ be the nodes in~$\succ(n)$ that $\Plc0$ can
  select in order to satisfy the existential part of
  ${\EUpair((\overline{q},(\overline{U},\et)) \mapsto k;{\{(\overline{q},\state_\top)\}})}$.
  Now consider a move of $\Pl0$ from~$(n,\EUpair(E;U))$ in~$\GAT$
  leading to some $(n,\nu)$. If a node $y \in Y$ is associated with a
  state in $U$ (\ie\ $\nu(y)\in U$), $\Pl1$ will have a \kl(S){winning} \kl{strategy}
  because $\Plc0$
  is winning from $(y,(\overline{q},\overline{\nu(y)}))$ by definition of~$Y$. 
  Otherwise all nodes in~$Y$
  are associated with states in~$\supp(E)$, but then we know that
  $\Plc0$ has a \kl(S){winning} \kl{strategy}
  from $(n,\phi_{\overline{q}}(\Phi_{(E\msminus m) \uplus
    m^{\overline U}}))$ with $m=\uplus_{y\in Y} \nu(y)$; this
  ensures that the constraint~$\EUpair(E;U)$ is not satisfied
  from~$(n,\nu)$ in~$\GAT$, and that $\Pl1$ has a winning strategy.
\end{itemize}

With arguments similar to the previous case, we can prove that the
resulting \kl{strategy} is \kl(S){winning} for \Pl1.
\end{proof}

To summarise our results:
\begin{theorem}\label{thm-compl}
Given an \AAPTA $\Aut=\tuple{\State,\initstate,\Trans,\omega}$, we~can build an \AAPTA~$\Aut^c$ 
recognising the complement of~$\calL(\Aut)$, with \kl(A){size} bounded by 
$(O({\size\State}^2\cdot \sizeB\Trans \cdot\size\Alp \cdot 3^{\sizeE\Trans},
O(\size\State\cdot (1+\sizeU\Trans)\cdot \sizeB\Trans\cdot (1+\sizeE\Trans^2\cdot 4^{\sizeE\Trans}),
 \sizeE\Trans+1,
  \max(\sizeU\Trans,1),
  \priomax+1)$.
\end{theorem}

\subsection{\kl{Alternation removal} (a.k.a. \kl{simulation})}
\label{ssec-simu}
\label{sec-simul}

Building a non-alternating automaton equivalent to a given alternating
automaton is an important construction, e.g.~in order to perform
\kl{projection}, or for algorithmic purposes. In~this section, we~present an
\intro{alternation-removal} (a.k.a.~\reintro{simulation}) algorithm, based on
ideas developed in~\cite{Wal02,Zan12} for \MSOAutomata.

For the rest of this section, we fix an \AAPTA $\Aut=\tuple{\State,
  \initstate,\Trans,\prio}$.
Intuitively, (conjunctive) \kl{alternation} consists in exploring each
subtree in several states of the automaton. In~order to remove
\kl{alternation}, we~follow the classical approach of building a kind of powerset automaton which, instead of visiting a single node in different
states~$\state_{i_1}$,~...,~$\state_{i_k}$ of~$\State$, 
explores that node in a single \emph{macro-state}, corresponding to
the union of all states~$\state_{i_1}$,~...,~$\state_{i_k}$.

We~illustrate this construction in Example~\ref{ex-simu}, where we
show why we need as a first step to keep track of the \emph{origin} of
each state (of~$\State$) appearing in a macro-state, in~order to be able to evaluate the acceptance condition.

\begin{example}\label{ex-simu}
Let $\Sigma=\{a,b\}$.  Consider an \AAPTA~$\Aut$
with an initial state~$\initstate$
with $\omega(\initstate)=1$, and a state~$\state_1$ with $\omega(\state_1)=0$,
and
\begin{xalignat*}1
\Trans(\initstate,a)&=\EUpair(\initstate\mapsto 1; \emptyset) \et
  \EUpair(\state_1\mapsto 1; \emptyset)
  \\
\Trans(\initstate,b) &=\EUpair(\state_1\mapsto 1; \emptyset)
 \\
\Trans(\state_1,a) = \Trans(\state_1,b) &=\EUpair(\initstate\mapsto 1; \emptyset).
\end{xalignat*}
Notice that this automaton only accepts \kl{trees} with a single \kl{branch}
(i.e.,~words), because all the constraints are of the form $\EUpair(\state \mapsto
1;\emptyset)$. It~is easily seen that the~word~${a^3\cdot b^\omega}$ is \kl{accepted},
while $a^\omega$ is~not. If~we~perform a simple powerset construction,
the~sequence of sets of states along the (unique) computation for both
\kl{words} is $\{\initstate\}\cdot \{\initstate,\state_1\}^\omega$.
This~does not keep enough information to decide if a run is accepting.

Now, if each state is paired with its ancestor (arbitrarily pairing
the initial state with itself), then the sequence of sets of pairs of
states visited along $a\cdot b^\omega$ is
$\{(\initstate,\initstate)\}\cdot \{(\initstate,\initstate),(\initstate,\state_1)\}\cdot
\{(\state_1,\initstate),(\initstate,\state_1)\}^\omega$,
while along $a^\omega$ it is
$\{(\initstate,\initstate)\}\cdot
\{(\initstate,\initstate),\allowbreak(\initstate,\state_1)\}\cdot
\{(\initstate,\initstate),(\initstate,\state_1),(\state_1,\state_1)\}^\omega$.
In~the latter sequence, we~can detect the presence of an infinite
\kl{branch} looping in~$\initstate$. Figure~\ref{fig-exalt} illustrates this
difference.
\end{example}

\begin{figure}[tb]
\centering
\begin{tikzpicture}
\begin{scope}[yshift=1cm]
\path(-2,0) node[text width=1.4cm,align=center] {input tree};
\path(0,0) node[text width=2cm,align=center] {execution tree};
\path(3,0) node[text width=2.7cm,align=center] {powerset of\\states};
\path(6.6,0) node[text width=2.5cm,align=center] {powerset of pairs of states};
\end{scope}
\begin{scope}
\begin{scope}[xshift=-2cm]
\draw (0,-.5) node {$a$};
\draw (0,-1.5) node {$a$};
\draw (0,-2.5) node {$a$};
\draw (0,-3.5) node {$b$};
\draw (0,-4.5) node {$b$};
\end{scope}
\begin{scope}[xscale=.75]
\draw (0,0) node (0q0) {$\initstate$};
\draw (-.75,-1) node (1q0) {$\initstate$};
\draw (1.5,-1) node (1q1) {$\state_1$};
\draw (-1.25,-2) node (2q0) {$\initstate$};
\draw (.25,-2) node (2q1) {$\state_1$};
\draw (1.5,-2) node (2q0') {$\initstate$};
\draw (-1.5,-3) node (3q0) {$\initstate$};
\draw (-.75,-3) node (3q1) {$\state_1$};
\draw (.25,-3) node (3q0') {$\initstate$};
\draw (1.25,-3) node (3q1') {$\initstate$};
\draw (2,-3) node (3q1'') {$\state_1$};
\draw (-1.5,-4) node (4q0) {$\state_1$};
\draw (-.75,-4) node (4q1) {$\initstate$};
\draw (.25,-4) node (4q0') {$\state_1$};
\draw (1.25,-4) node (4q1') {$\state_1$};
\draw (2,-4) node (4q1'') {$\initstate$};
\draw[-latex'] (0q0) -- (1q0);
\draw[-latex'] (0q0) -- (1q1);
\draw[-latex'] (1q0) -- (2q0);
\draw[-latex'] (1q0) -- (2q1);
\draw[-latex'] (1q1) -- (2q0');
\draw[-latex'] (2q0) -- (3q0);
\draw[-latex'] (2q0) -- (3q1);
\draw[-latex'] (2q1) -- (3q0');
\draw[-latex'] (2q0') -- (3q1');
\draw[-latex'] (2q0') -- (3q1'');
\draw[-latex'] (3q0) -- (4q0);
\draw[-latex'] (3q1) -- (4q1);
\draw[-latex'] (3q0') -- (4q0');
\draw[-latex'] (3q1') -- (4q1');
\draw[-latex'] (3q1'') -- (4q1'');
\draw[dashed] (4q0) -- +(-90:8mm);
\draw[dashed] (4q1) -- +(-90:8mm);
\draw[dashed] (4q0') -- +(-90:8mm);
\draw[dashed] (4q1') -- +(-90:8mm);
\draw[dashed] (4q1'') -- +(-90:8mm);
\end{scope}
\begin{scope}[xshift=3cm]
\draw (0,0) node (0q0) {$\initstate$};
\draw (0,-1) node (1q01) {$\{\initstate,\state_1\}$};
\draw (0,-2) node (2q01) {$\{\initstate,\state_1\}$};
\draw (0,-3) node (3q01) {$\{\initstate,\state_1\}$};
\draw (0,-4) node (4q01) {$\{\initstate,\state_1\}$};
\draw[-latex'] (0q0) -- (1q01);
\draw[-latex'] (1q01) -- (2q01);
\draw[-latex'] (2q01) -- (3q01);
\draw[-latex'] (3q01) -- (4q01);
\draw[dashed] (4q01) -- +(-90:8mm);
\end{scope}
\begin{scope}[xshift=6.6cm]
\def\mypair#1#2{(\ifnum#1=0\relax\initstate\else\state_{#1}\fi,
    \ifnum#2=0\relax\initstate\else\state_{#2}\fi)}
\draw (0,0) node (0q0) {$\{\mypair 00\}$};
\draw (0,-1) node (1q01) {$\{\mypair 00, \mypair 01\}$};
\draw (0,-2) node (2q01) {$\{\mypair 00, \mypair 01, \mypair 10\}$};
\draw (0,-3) node (3q01) {$\{\mypair 00, \mypair 01, \mypair 10\}$};
\draw (0,-4) node (4q01) {$\{\mypair 01, \mypair 10\}$};
\draw[-latex'] (0q0) -- (1q01);
\draw[-latex'] (1q01) -- (2q01);
\draw[-latex'] (2q01) -- (3q01);
\draw[-latex'] (3q01) -- (4q01);
\draw[dashed] (4q01) -- +(-90:8mm);
\end{scope}
\end{scope}
\begin{scope}[yshift=-5.7cm]
\begin{scope}[xshift=-2cm]
\draw (0,-.5) node {$a$};
\draw (0,-1.5) node {$a$};
\draw (0,-2.5) node {$a$};
\draw (0,-3.5) node {$a$};
\end{scope}
\begin{scope}[xscale=.75]
\draw (0,0) node (0q0) {$\initstate$};
\draw (-.75,-1) node (1q0) {$\initstate$};
\draw (1.5,-1) node (1q1) {$\state_1$};
\draw (-1.25,-2) node (2q0) {$\initstate$};
\draw (.25,-2) node (2q1) {$\state_1$};
\draw (1.5,-2) node (2q0') {$\initstate$};
\draw (-1.5,-3) node (3q0) {$\initstate$};
\draw (-.75,-3) node (3q1) {$\state_1$};
\draw (.25,-3) node (3q0') {$\initstate$};
\draw (1.25,-3) node (3q1') {$\initstate$};
\draw (2,-3) node (3q1'') {$\state_1$};
\draw[-latex'] (0q0) -- (1q0);
\draw[-latex'] (0q0) -- (1q1);
\draw[-latex'] (1q0) -- (2q0);
\draw[-latex'] (1q0) -- (2q1);
\draw[-latex'] (1q1) -- (2q0');
\draw[-latex'] (2q0) -- (3q0);
\draw[-latex'] (2q0) -- (3q1);
\draw[-latex'] (2q1) -- (3q0');
\draw[-latex'] (2q0') -- (3q1');
\draw[-latex'] (2q0') -- (3q1'');
\draw[dashed] (3q0) -- +(-100:8mm);
\draw[dashed] (3q0) -- +(-70:8mm);
\draw[dashed] (3q1) -- +(-90:8mm);
\draw[dashed] (3q0') -- +(-100:8mm);
\draw[dashed] (3q0') -- +(-70:8mm);
\draw[dashed] (3q1') -- +(-100:8mm);
\draw[dashed] (3q1') -- +(-70:8mm);
\draw[dashed] (3q1'') -- +(-90:8mm);
\end{scope}
\begin{scope}[xshift=3cm]
\draw (0,0) node (0q0) {$\initstate$};
\draw (0,-1) node (1q01) {$\{\initstate,\state_1\}$};
\draw (0,-2) node (2q01) {$\{\initstate,\state_1\}$};
\draw (0,-3) node (3q01) {$\{\initstate,\state_1\}$};
\draw[-latex'] (0q0) -- (1q01);
\draw[-latex'] (1q01) -- (2q01);
\draw[-latex'] (2q01) -- (3q01);
\draw[dashed] (3q01) -- +(-90:8mm);
\end{scope}
\begin{scope}[xshift=6.6cm]
\def\mypair#1#2{(\ifnum#1=0\relax\initstate\else\state_{#1}\fi,
    \ifnum#2=0\relax\initstate\else\state_{#2}\fi)}
\draw (0,0) node (0q0) {$\{\mypair 00\}$};
\draw (0,-1) node (1q01) {$\{\mypair 00, \mypair 01\}$};
\draw (0,-2) node (2q01) {$\{\mypair 00, \mypair 01, \mypair 10\}$};
\draw (0,-3) node (3q01) {$\{\mypair 00, \mypair 01, \mypair 10\}$};
\draw[-latex'] (0q0) -- (1q01);
\draw[-latex'] (1q01) -- (2q01);
\draw[-latex'] (2q01) -- (3q01);
\draw[dashed] (3q01) -- +(-90:8mm);
\end{scope}
\end{scope}
\end{tikzpicture}
\caption{Runs of the automaton~$\Aut$ of Example~\ref{ex-simu}
on $a^3\cdot b^\omega$ and on $a^\omega$}
\label{fig-exalt}
\end{figure}

Our construction follows this intuition. It~consists in four steps,
represented in Fig.~\ref{fig-schema-simu}: the~first step just
consists in pairing states with their predecessors, as we just
illustrated; the~second step builds an (alternating) powerset
automaton, involving a new satisfaction relation~$\models*$ and a new
acceptance condition; the~third step is our main step, where we
(inductively)
remove conjunctions from the transition function, until
it~is \kl{non-alternating}, which allows us to come back to our
original satisfaction relation~$\models+$; the fourth step turns the
acceptance condition back into a \kl{parity condition}, by~taking a
product with an auxiliary parity word automaton~$\MAut_{\prio}$
enforcing acceptance along each branch.

\begin{figure}[ht]
\centering
\begin{tikzpicture}
  \path (0,0) node {$\Aut$: original \AAPTA, using \EU-constraints
    over~$\State$, and relation~$\models+$};
  \path (0,-1) node {$\PAut$: \AAPTA, using \EU-constraints
    over~$\State\times\State$, and relation~$\models+$};
  \path (0,-2) node {$\QAut$: \AATA, using \EU-constraints
    over~$2^{\State\times\State}$, and relation~$\models*$};
  \path (0,-3) node {$\RAut$: \nAATA, using \EU-constraints
    over~$2^{\State\times\State}$, and relation~$\models*$};
  \path (0,-4) node {$\NAut$: \nAAPTA, using \EU-constraints
    over~$\State_{\prio}\times 2^{\State\times\State}$, and relation~$\models+$};
\end{tikzpicture}
\caption{Sequence of transformations for the simulation construction}
\label{fig-schema-simu}
\end{figure}

\subsubsection{Keeping track of ancestor states}

In this section, we~modify automaton~$\Aut$ so as to store, in each
state, its ancestor state. 
For~this, we~reuse the mapping~$\phi_{\state'}$ introduced in Section~\ref{remUsingleton}, defined on~$\State$ as $\phi_{\state'}(\state)=(\state',\state)$ and extended to~\EUconstrs in the natural way.

We~then define the \AAPTA $\PAut=\tuple{\State^2, (\initstate,\initstate),
  \PTrans, \Pprio}$ with $\PTrans((\state,\state'),\alp) =
\phi_{\state'}(\Trans(\state',\alp))$, and
$\Pprio(\state,\state')=\prio(\state')$.  Intuitively,
state~$(\state,\state')$ in~$\PAut$ corresponds to state~$\state'$
in~$\Aut$, with the extra information that this state originates from
state~$\state$. Notice that both $\PTrans$ and $\Pprio$ only depend on
the second state of the pair~$(\state,\state')$.

\begin{proposition}\label{prop-AtoP}
The \kl{languages} of~$\Aut$ and~$\PAut$ are equal.  Moreover, if a
\kl{tree}~$\Tree$ is \kl{accepted} by~$\PAut$, then there is an \kl(ET){accepting}
\kl{execution tree}~$\ExTree=\tuple{\extree,\exlab}$ of~$\PAut$ on~$\Tree$
in which the subtrees rooted at any two nodes~$n_\extree$
and~$n'_\extree$ for which $\exlab(n_\extree)=(m,(\state',\state))$
and $\exlab(n'_\extree)=(m,(\state'',\state))$ are equal.
The size of $\PAut$ is $\tuple{\size\State^2, \sizeB\Trans, \sizeE\Trans,\sizeU\Trans,\priomax}$.
\end{proposition}

\begin{proof}
Take a \kl{tree}~$\Tree$. Assuming that $\Tree\in\Lang(\Aut)$, take an
\kl(ET){accepting} \kl{execution tree}~$\ExTree=\tuple{\extree,\exlab}$ of~$\Aut$
on~$\Tree$. With~$\ExTree$, we~associate another
\kl{tree}~$\PExTree=\tuple{\extree,\pexlab}$ with the same structure, and
with labelling function~$\pexlab$ defined as follows:
for the \kl{root}, $\pexlab(\emptyw_\extree) = (\emptyw_\tree,
(\initstate,\initstate))$ and, for any non-root
node~$n_\extree=(\dir_i,\state_i)_{0\leq i<\size{n_{\extree}}}$ (hence
having $\exlab(n_\extree)=((\dir_i)_{0\leq i<\size{n_{\extree}}},\state_{\size{n_\extree}-1})$),
\[
\pexlab(n_\extree) = ((\dir_i)_{0\leq i<\size{n_{\extree}}},
(\state_{\size{n_\extree}-2},\state_{\size{n_\extree}-1}))),
\]
where $\state_{\size{n_\extree}-2}$ is~$\initstate$ when $\size{n_\extree}=1$.

It~should be clear that $\PExTree$ is an \kl(ET){accepting} \kl{execution tree}
of~$\PAut$ on~$\Tree$, since the only difference is the addition of
the previous state of the automaton in the labelling, which has no
impact on the transition function nor on the satisfaction of the
acceptance condition.

\smallskip
Conversely, given an \kl(ET){accepting} \kl{execution
tree}~$\PExTree=\tuple{\pextree,\pexlab}$ of~$\PAut$ on~$\Tree$
witnessing the fact that $\Tree$ is \kl{accepted} by~$\PAut$, we obtain an
\kl(ET){accepting} \kl{execution tree} of~$\Aut$ on~$\Tree$ by simply erasing the
first item of the second component of the labelling function. Again,
it~is easily seen that this defines an \kl(ET){accepting} \kl{execution
tree} of~$\Aut$ on~$\Tree$.

\smallskip

Finally, if a tree~$\Tree$ is \kl{accepted} by~$\PAut$, then it is \kl{accepted}
by~$\Aut$, and there exists an \kl(ET){accepting} \kl{execution
tree}~$\ExTree$
of~$\Aut$ on~$\Tree$ such that any two nodes of~$\ExTree$ carrying the
same labels are \kl{roots} of the same subtrees. The~result follows.
\end{proof}

\subsubsection{Building the powerset automaton}\label{ssec-powerset}

In this section, we~perform our powerset construction: we~build an
(\kl{alternating}) \AATA~$\QAut$ whose states are sets of states
of~$\PAut$. This requires two important changes in our setting:
we~will use a modified notion of satisfaction of \EUprs, based on
sets of states, and we will use a new acceptance condition. Notice
that we do \emph{not} remove \kl{alternation} here: this will be done in
the next section, and will allow us to come back to our original
notion of satisfaction of \EUprs.

From the \AAPTA $\PAut=\tuple{\State^2, (\initstate,\initstate), \PTrans,
  \Pprio}$, we~build the powerset \AATA $\QAut=\tuple{\QState,
  \{(\initstate,\initstate)\}, \QTrans, \Accept_{\prio}}$ by letting:
\begin{itemize}
\item $\QState=2^{\State^2}$ contains all the sets of states of~$\PAut$,
  hence all the sets of pairs of states of~$\Aut$;
\item $\QTrans\bigl(\{(\state_i, \state'_i) \mid 1\leq i\leq
  k\},\alp\bigr) = \ET_{1\leq i\leq k}
  \PTrans^s((\state_i,\state'_i),\alp)$,
  where $\PTrans^s((\state_i,\state'_i),\alp)$ is obtained from
  $\PTrans((\state_i,\state'_i),\alp)$ by replacing each pair of
  states~$(\state,\state')$ by the singleton $\{(\state,\state')\}$.
  For the time being, this powerset automaton still is alternating.

  Notice that if we keep our definition of \kl{execution trees}, then
  $\PAut$ and~$\QAut$ would have the same behaviours (and only
  \emph{singleton states} of~$\QAut$ would be used).  Below, we~introduce
  a new notion of \kl{execution trees}, which uses the same tree
  structure as the input tree, and gathers all states of~$\PAut$
  visiting a given node of the input tree into a single state
  of~$\QAut$ visiting that node.

\item
  the acceptance condition $\Accept_{\prio}$ for~$\QAut$ will be based on $\prio$ (and $\Pprio$), but it is \emph{not}
  a parity acceptance condition. We~define it formally below.
\end{itemize}

We~call \intro*\powAATA any automaton of the
form~$\tuple{\QState, \{(\initstate,\initstate)\}, \QTrans', \Accept_\prio}$,
which only differs from~$\QAut$ in its transition
function~$\QTrans'\colon \QState\times\Alp \to \PBF{\EUset(\QState))}$.

\smallskip
We~now define our new notion of
\kl(pow){execution trees} for \powAATA, based on
a new notion of satisfaction for \EUprs.  This~is based on
identifying \kl{markings} of~$\Set'$ by~$\Set$ as \kl{unitary markings}
of~$\Set'$ by~$2^{\Set}$.

\begin{definition}\label{def-models*}
Let $\Set$ and~$\Set'$ be two sets,
$\EUpair(E;U)$ be an \EUpr over~$2^\Set$,
and~$\nu$~be a \kl{marking} of~$\Set'$ by~$\Set$, seen as a unitary marking of~$\Set'$ by~$2^{\Set}$.
Then $\nu$~\intro{set-satisfies}~$\EUpair(E;U)$, denoted $\nu\mathrel{\reintro{\models*}} \EUpair(E;U)$,
if there exists a \kl{submarking}~$\nu'\submarking \nu$,
seen as a \kl{unitary marking} of~$\Set'$ by~$2^{\Set}$, such that
$\mkimg{\nu'}\models \EUpair(E;U)$.
\end{definition}

This relation is extended to \kl{positive boolean combinations}
of \EUprs in the same way as for~$\models+$. 
With this new satisfaction relation, we~define a new notion of
\kl(pow){execution trees} for \powAuta,
whose structure is the same as that of the input \kl{tree}:

\begin{definition}\label{def-exectree2}
Let
$\OAut=\tuple{\QState,\{(\initstate,\initstate)\},\QTrans',\Accept_{\prio}}$
be an  \powAATA over~$\Alp$, and
$\Tree=\tuple{\tree,\lab}$ be a \SDtree{\Alp}{\Dir}, for some
finite set~$\Dir$.  An~\intro(pow){execution tree} of~$\OAut$ over~$\Tree$ is a
\SDtree{\QState}{\Dir}~$\ExTree=\tuple{\extree,\exlab}$ such that $\extree=\tree$
and
\begin{itemize}
\item the \kl{root}~$\emptyw_\extree$ in~$\ExTree$ is labelled
  with~$\{(\initstate,\initstate)\}$;
\item for any node $n_{\extree}=(\dir_i)_{0\leq i<\size{n_{\extree}}}$ of~$\extree$
  (which we can identify with the corresponding node~$m_\tree=(\dir_i)_{0\leq i<\size{m_{\tree}}}$ of the input \kl{tree}),
  letting~$\nu_{n_\extree}$ be the \kl{marking} of~$\succ(m_\tree)$ by~$2^{\State\times\State}$
  such that $\nu_{n_\extree}(m_\tree\cdot d)= \exlab(n_\extree\cdot d)$,
   we~have
   $\nu_{n_\extree}\models* \QTrans'(\exlab(n_{\extree}), \lab(m_{\tree}))$.
\end{itemize}

Whether such an \kl(pow){execution tree} is \kl(powET){accepting} is defined as follows:
  consider an infinite \kl{branch}~$b=(n_{t_i})_{0\leq i<\infty}$
    of the \kl(pow){execution tree} of~$\OAut$ on~$\Tree$,
  with~$\exlab(n_{t_i})=\{(\state_{i,j}, \state'_{i,j}) \mid 0\leq j\leq
  z_i\}$ for each~$i\in\bbN$. A~sequence~$(r_i)_{0\leq i<k}$ of states
  of~$\Aut$ is said to \intro{appear} in \kl{branch}~$b$ if for each~$i\in\bbN$, there
  exists an index~$0\leq j\leq z_i$ such that
  $(r_{i-1},r_i)=(\state_{i,j},\state'_{i,j})$. \kl{Branch}~$b$ is
  \intro(powB){accepting} if all the sequences~$(r_i)_{0\leq i<k}$
  that \kl{appear} in that \kl{branch} satisfy
  the \kl{parity condition} $\prio$ of~$\Aut$; the~\kl(pow){execution
  tree} is \intro(powET){accepting} if all its \kl{branches}~are.
\end{definition}

\makeatletter
\newcounter{save@cptr}
\setcounter{save@cptr}{\value{example}}
\expandafter\ifx\csname r@ex-simu\endcsname\relax
\def\@tmp{{1}{}{}{}{}}
\else
\edef\@tmp{\csname r@ex-simu\endcsname}
\fi
\setcounter{example}{\expandafter\@firstoffive\@tmp}
\addtocounter{example}{-1}
\makeatother
\begin{example}[contd]
  Consider again the automaton~$\Aut$ of Example~\ref{ex-simu}, and
  write~$\QAut$ for the \powAuton obtained from~$\Aut$ by applying the
  transformation above. The~(one-branch) trees to the right of
  Fig.~\ref{fig-exalt} are \kl(pow){execution trees} of~$\QAut$ on
  (one-branch) input trees $a^3\cdot b^\omega$ and $a^\omega$.

  For instance, consider the second node of the \kl(pow){execution tree}
  of~$\QAut$ on~$a^3\cdot b^\omega$, labelled with the state
  $s=\{(\initstate,\initstate), (\initstate,\state_1)\}$ of~$\QAut$.
  By~construction of~$\QAut$, we~have
  \[
  \QTrans(s,a)=
  \bigl(\EUpair({\{(\initstate,\initstate)\}\mapsto 1};\emptyset)
  \et
  \EUpair({\{(\initstate,\state_1)\}\mapsto 1};\emptyset)\bigr)
  \; \et \;  
  \EUpair({\{(\state_1,\initstate)\}}\mapsto 1;\emptyset),
  \]
  where the first term corresponds to~$\Trans(\initstate,a)$ and the
  second term corresponds to~$\Trans(\state_1,a)$.
  And  the third node of the \kl(pow){execution tree} indeed
  \kl{set-satisfies}~$\QTrans(s,a)$.

  On~input~$a^3\cdot b^\omega$, the only~\kl{branch} of the
  \kl(pow){execution tree} is
  \begin{multline*}
    \{(\initstate,\initstate)\}
    \{(\initstate,\initstate),(\initstate,\state_1)\}
    (\{(\initstate,\initstate),(\initstate,\state_1),(\state_1,\initstate)\}) \\
    (\{(\initstate,\initstate),(\initstate,\state_1),(\state_1,\initstate)\})
        (\{(\initstate,\state_1),(\state_1,\initstate)\})^\omega
  \end{multline*}
  There are five sequences of $\State^\omega$ \kl{appearing} in this \kl{branch}:
  $\initstate^3\cdot(\initstate\cdot \state_1)^\omega$,
  $\initstate^2\cdot(\initstate\cdot \state_1)^\omega$,
  $\initstate\cdot(\initstate\cdot \state_1)^\omega$,
  $\initstate\cdot\state_1\cdot\initstate\cdot(\initstate\cdot \state_1)^\omega$, and
  $(\initstate\cdot \state_1)^\omega$.
  All~five of them are accepting w.r.t.~the \kl{parity condition}
  of~$\Aut$ ($\omega(\initstate)=1$ and $\omega(\state_1)=0$), and
  thus this \kl(pow){execution tree} on~$a^3\cdot b^\omega$ is
  \kl(powET){accepting}.

  On the other hand, on input~$a^\omega$, the~unique \kl{branch} of
    the \kl(pow){execution tree}~is:
  \[
  \{(\initstate,\initstate)\} \{(\initstate,\initstate),(\initstate,\state_1)\}
    (\{(\initstate,\initstate),(\initstate,\state_1),(\state_1,\initstate)\})^\omega.
   \]
   The~sequences of $\State^\omega$ \kl{appearing} in this branch are of the form
    $(\initstate^+\cdot\state_1)^+\cdot\initstate^\omega$ and
    $(\initstate^+\cdot\state_1)^\omega$; sequences of the former form
   contain only finitely many occurrences of~$\state_1$, so that
   this branch is not \kl(powET){accepting}.
\end{example}
\makeatletter
\setcounter{example}{\value{save@cptr}}
\makeatother

\begin{proposition}\label{prop-PtoQ}
A \SDtree{\Alp}{\Dir}~$\Tree$ is accepted by~$\PAut$ if, and
only~if, it~is accepted by the \powAATA~$\QAut$.
The size of $\QAut$ is $\tuple{2^{\size\State^2}, \size\State^2\cdot
  \sizeB\Trans, \sizeE\Trans,\sizeU\Trans,\mathord -}$ (remember that $\QAut$ is \emph{not}
a parity automaton).
\end{proposition}

\begin{proof}
  Assuming that $\Tree=\tuple{\tree,\lab}$ is \kl{accepted} by~$\PAut$,
  take an \kl(ET){accepting} \kl{execution tree}~$\ExTree=\tuple{\extree,\exlab}$
  of~$\PAut$ on~$\Tree$. By~Prop.~\ref{prop-AtoP}, we~may assume that
  any two subtrees rooted at any two nodes~$n_\extree$ and~$n'_\extree$
  of~$\ExTree$ such that $\exlab(n_\extree)=(m,(\state',\state))$ and
  $\exlab(n'_\extree)=(m,(\state'',\state))$ are the same.

  Consider the \kl{tree}~$\ExTree'=\tuple{\tree,\exlab'}$ having the
  same \kl{tree structure} as~$\Tree$ and with $\exlab'(n_{\tree})=\{
  (q,q')\in\State\times\State \mid \exists
  n_\extree\in\extree.\ \exlab(n_\extree)=(n_\tree,(q,q'))\}$.
  Notice that the transition function is satisfied at each node:
  consider a node~$n_\tree$ whose labelling by~$\exlab'$ is a set of
  pairs~$(q,q')$. Then by definition, the~successors of~$n_\tree$
  collect all the pairs~$(r,r')$ used to satisfy every \EUpr
  required by~$\PTrans$ for the labels of the form~$(n_\tree,(r,r'))$,
  which allows to \kl{set-satisfy}~$\QTrans$ from~$n_\tree$.
  It~follows that $\ExTree'$ is an \kl(pow){execution tree} of~$\QAut$
  on~$\Tree$.
  
  We~now prove that $\ExTree'$ is \kl(powET){accepting}: take a
  \kl{branch}~$b=(n_{t_i})_{i\in\bbN}$ of~$\ExTree'$, with
  $\exlab'(n_{t_i})=\{(\state_{i,j}, \state'_{i,j}) \mid 0\leq j\leq z_i\}$
  for each~$i\in\bbN$.  Take a~sequence~$(r_i)_{i\in\bbN}$ of states
  that
  \kl{appears} in~$b$. We~claim that $(n_i,(r_{i-1},r_i))_{i\in\bbN}$,
  with $r_{-1}=\initstate$, is a \kl{branch} of~$\ExTree$. If~this~were not
  the case, consider the first index~$i_0$ such that
  $(n_i,(r_{i-1},r_i))_{i\leq i_0}$ is a prefix of a \kl{branch}
  of~$\ExTree$, and $(n_i,(r_{i-1},r_i))_{i\leq i_0+1}$ is~not. 
  
  By the definition of $\ExTree'$, we know that $\exlab'(n_{t_{i_0}})
  \ni (r_{i_0-1},r_{i_0})$, and then there exists a node $n_u$ in
  $\ExTree$
  s.t.\ $\exlab(n_u)=(n_{t_{i_0}},(r_{i_0-1},r_{i_0}))$. Moreover we
  have that there exists $d\in \Dir$ such that
  $\exlab'(n_{t_{i_0}\cdot d}) \ni (r_{i_0},r_{i_0+1})$ and then there
  exists a node $n'_u$ in $\ExTree$
  s.t.\ $\exlab(n'_u)=(n_{t_{i_0}\cdot d},(r_{i_0},r_{i_0+1}))$.
  
  The predecessor of $n'_u$ in $\ExTree$ is then labelled by some
  $(n_{t_{i_0}},(s,r_{i_{0}}))$, and by Prop.~\ref{prop-AtoP}, we~can
  assume that the subtrees rooted from this node and from~$n_u$ are
  the same: this entails that $(n_i,(r_{i-1},r_i))_{i\leq i_0+1}$ is a
  prefix of a \kl{branch} of~$\ExTree$. Therefore $\ExTree'$ is
  \kl(powET){accepting}, and $\Tree=\tuple{\tree,\lab}$ is \kl(powET){accepted}
  by~$\QAut$.

  \smallskip
  
Conversely, assume  that $\Tree=\tuple{\tree,\lab}$ is \kl(powET){accepted}
by~$\QAut$ and consider an \kl(powET){accepting} \kl(pow){execution
tree}~$\ExTree'=\tuple{\tree,\exlab'}$ of~$\QAut$ on~$\Tree$. From
$\ExTree'$, we~build a \SDtree{(\tree\times \State^2)}{(\Dir
\times \State^2)}~$\ExTree = (\extree,\exlab)$ level-by-level, in
such a way that it~is an \kl(ET){accepting} \kl{execution tree} of~$\PAut$
on~$\Tree$.  During the inductive construction, we will maintain the
invariant that for any node~$n_\tree$ at depth~$i$ in~$\Tree$,
$\exlab'(n_\tree)$~is exactly the set of pairs~$(q,q')$ occurring in a
labelling of $\ExTree$-nodes at depth~$i$ of the
form~$(n_\tree,(q,q'))$.

First we define the labelling of the~\kl{root}:
$\exlab(\emptyw_\extree)=(\emptyw_\tree,(\initstate,\initstate))$.
The~invariant property clearly holds true at level~$0$.

Now consider a previously-defined node $n_\extree$ of~$\ExTree$ labelled with
$(n_\tree,(q,q'))$. Then by the invariant, we~have $\exlab'(n_\tree)
\ni (q,q')$, and by~definition of~$\QAut$,
all its successors $\{n_\tree\cdot d\in\Tree \mid d\in\Dir\}$
are labelled by~$\exlab'$ with a set of pairs of the form $(q',r)$ that satisfy the
$\PTrans$-function. We~precisely add successors to~$n_\extree$ in
order to get exactly the same labels $(n_{\tree}\cdot d,(q',r))$ for
all $n_\tree\cdot d$ in~$\tree$.
This maintains the invariant and the transition
function~$\PTrans$ is locally satisfied by the definition of~$\ExTree$.
 
Now we can easily see that this \kl{execution tree} is
\kl(ET){accepting}: consider a \kl{branch}~$b$ of~$\ExTree$; its~labelling
describes a sequence $(r_i)_{i\in\bbN}$ that also \kl{appears} in the
corresponding \kl{branch} in~$\ExTree'$.
\end{proof}

\subsubsection{Removing conjunctions}

We now remove conjunctions from the transition function~$\QTrans$
of~$\QAut$. As~a first step, we~turn each formula~$\QTrans(P,\alp)$ in
disjunctive normal form.
We~can bound the number of different \EUprs appearing in any
given~$\QTrans(P,\alp)$ by $\size\State\cdot \sizeB\Trans$: indeed,
while it~is built as a conjunction of up to~$\State^2$ transition
formulas, any~two pairs~$(\state',\state)$ and~$(\state'',\state)$
give rise to the same \EUprs.
It~follows that $\QTrans(P,\alp)$
can be written as the disjunction of at most
$2^{\size\State\cdot \sizeB\Trans}$ conjunctions of at most
$\size\State\cdot\sizeB\Trans$~\EUprs.

We~now turn those conjunctions into
disjunctions. We~proceed inductively, by replacing any conjunction
$\EUpair(E_1;U_1)\et \EUpair(E_2;U_2)$ of two \EUprs
over~$2^{\State\times\State}$ with an ``equivalent'' disjunction of
\EUprs over~$2^{\State\times\State}$ (in~the sense that the
transformation preserves the language of the automaton).

\begin{figure}[t]
  \centering
  \begin{tikzpicture}[scale=1.5]
    \draw (0,-.1) node[draw,rond, minimum size=3mm,inner sep=1.2pt] (m) {};
    \foreach \x in {-1.8,-1.45,...,1.81}
             {\draw (\x,-.7) node[draw,rond, minimum size=3mm,inner sep=1pt] (n) {};
              \draw (m) -- (n);}
    \draw (-.8,-1) node[draw,rounded corners=1mm,inner sep=1pt,
      minimum width=3.6cm] (E1) {$E_1$}; 
    \draw (1.275,-1) node[draw,rounded corners=1mm,inner sep=1pt,
      minimum width=2.175cm] (U1) {$U_1$};
    \draw (-.2,-1.4) node[draw,rounded corners=1mm,inner sep=1pt,
      minimum width=2.925cm] (E2) {$E_2$}; 
    \draw (1.45,-1.4) node[draw,rounded corners=1mm,inner sep=1pt,
      minimum width=1.65cm] (U2) {$U_2$};
    \draw (-1.65,-1.4) node[draw,rounded corners=1mm,inner sep=1pt,
      minimum width=1.05cm] (U3) {$U_2$};
    \draw[decoration={calligraphic brace,mirror},decorate,line width=.6pt]
    (-1.24,-1.7) -- (.47,-1.7) node[below,midway,text width=2cm, align=center]
         {existential overlap};
    \draw[decoration={calligraphic brace,mirror},decorate,line width=.6pt]
    (.835,-1.7) -- (2,-1.7) node[below,midway,text width=2cm, align=center]
         {universal overlap};
    \draw[dashed] (-1.24,-1.6) -- +(0,.8);
    \draw[dashed] (.47,-1.6) -- +(0,.8);
    \draw[dashed] (.835,-1.6) -- +(0,.8);
  \end{tikzpicture}
  \caption{Representation of the overlaps in an execution tree
    when satisfying a conjunction of two \EU-constraints
  \protect\(\protect\EUpair(E_1;U_1)\protect\) and
  \protect\(\protect\EUpair(E_2;U_2)\protect\).}
  \label{fig-overlaps}
\end{figure}

Write $m_1$, $n_1$, $m_2$ and $n_2$ for the sizes of~$E_1$, $U_1$,
$E_2$ and~$U_2$, respectively.  The~disjunction we build ranges over
the possible ways the ``existential'' and ``universal'' parts of
the \EUprs overlap (see~Fig.~\ref{fig-overlaps}).
For~each combination, we~write an \EUpr whose existential part
contains the ``existential'' overlaps and the two ``mixed'' overlaps,
and whose universal part handles the ``universal'' overlap.

The disjunction of \EUprs can then be written as follows:
\begin{multline*}
  C(\EUpair(E_1;U_1),\EUpair(E_2;U_2)) =
  \OU_{\substack{J_1\submultiset E_1, J_2\submultiset E_2 \\ \MSsize{J_1}=\MSsize{J_2}}}
  \
  \OU_{\substack{\tau \text{ permutation }\\\text{of~$[1;\size{J_1}]$}}} 
  \OU_{%
      \substack
      {g_1\colon E_1\msminus J_1 \to U_2 \\
       g_2\colon E_2\msminus J_2 \to U_1}
  } \\
  \left(
  E'= \biguplus
  \begin{array}{l}
    \mset{ j^1_k \cup j^2_{\tau(k)} \mid 1\leq k\leq \size{J_1} } \\
    \mset{ e^1_k \cup g_1(e^1_k) \mid 1\leq k\leq n_1-\size{J_1}} \\
    \mset{ g_2(e^2_k) \cup e^2_k \mid 1\leq k\leq n_2-\size{J_2}}
  \end{array};
  U'=U_1 \otimes U_2
  \right)
\end{multline*}
where we use the notations
\begin{xalignat*}2
J_1&=\mset{ j^1_k \mid 1\leq k\leq o} &
J_2&=\mset{ j^2_k \mid 1\leq k\leq o} \\
E_1\msminus J_1 &= \mset{ e^1_k \mid 1\leq k\leq n_1-o} &
E_2\msminus J_2 &= \mset{ e^2_k \mid 1\leq k\leq n_2-o} \\
U_1 \otimes U_2&= \{ u_1 \cup u_2 \mid u_1 \in U_1, u_2 \in U_2\}. & 
\end{xalignat*}

Note that the sizes of existential and universal parts of any \EUpr
in $C({\EUpair(E_1;U_1)},\allowbreak{\EUpair(E_2;U_2)})$
are bounded by $n_1+n_2$ and $m_1\cdot m_2$, respectively.

\begin{remark}\label{rk-four}
In case~$U_2$ is empty (the case of~$U_1$ being empty would be
symmetric), then the only possible overlaps are between~$E_2$
and~$\EUpair(E_1;U_1)$. This is reflected in our formula by
considering that, when~$U_2$ is empty, there exist no functions
$g_1\colon E_1\setminus J_1\to U_2$ when $E_1\setminus J_1$~is~not
empty, while there is a single one when $E_1\setminus J_1$ is
empty. In~other terms, if~$U_2$ is empty, we~must have $J_1=E_1$.
Notice also that if $U_2$ is empty, then $U_1\otimes U_2$ also~is.

In case both $U_1$ and $U_2$ are empty, then we must have $J_1=E_1$
and $J_2=E_2$, and both must have the same size.  Then
$C(\EUpair(E_1;U_1),\EUpair(E_2;U_2))$ is just a disjunction, over all
permutations of~$E_1$, of constraints of the form~$\EUpair(\mset{ e^1_{\tau(k)} \cup e^2_{k}
\mid 1\leq k\leq \size{E_1} };\emptyset)$.
\end{remark}

\makeatletter
\def\mypair{\@ifnextchar*{\@pairstar}{\@pair}}
\def\@pairstar*#1{\genfrac(){0pt}{1}{\star}{\ifnum#1=0\relax\initstate\else\state_{#1}\fi}}
\def\@pair#1#2{\genfrac(){0pt}{1}{\ifnum#1=0\relax\initstate\else\state_{#1}\fi}{\ifnum#2=0\relax\initstate\else\state_{#2}\fi}}
\makeatother
\newcounter{runningex}
\setcounter{runningex}{\value{example}}

\begin{example}\label{ex-transform}
Consider an \AATA~$\Aut$ with
$\State=\{\state_i \mid 0\leq i\leq 4\}$, and assume that the transition function for $\state_1$ and $\state_2$ is as follows:
\begin{xalignat*}1
\delta(\state_1,\alp) &= \EUpair(\state_1\mapsto 2; \{\state_3\})
  \ou \EUpair(\state_2\mapsto 2; \{\state_2,\state_3\}) \\
\delta(\state_2,\alp) &= \EUpair(\state_3\mapsto 1; \{\state_1,\state_4\}) 
\end{xalignat*}
Now assume that after building the
corresponding automata~$\PAut$ and~$\QAut$, we~have to deal with the state $\{\mypair31,\mypair42\}$. We 
get the following formula
(where, for the sake of readability, brackets are omitted for singleton sets):
\begin{multline*}
\QTrans\left(\{\mypair31,\mypair42\},\alp\right) =
 \left(\EUpair(\mypair11\mapsto 2; \{\mypair13\})
  \ou \EUpair(\mypair12\mapsto 2; \{\mypair12,\mypair13\})\right) \et {} \\
 \EUpair(\mypair23\mapsto 1; \{\mypair21,\mypair24\})
\end{multline*}
Turning this into disjunctive normal form gives
\begin{multline*}
\QTrans\left(\{\mypair31,\mypair42\},\alp\right) =
\left( (\EUpair(\mypair11\mapsto 2; \{\mypair13\}) \et  \EUpair(\mypair23\mapsto 1; \{\mypair21,\mypair24\})\right)
\ou {}\\
\left(\EUpair(\mypair12\mapsto 2; \{\mypair12,\mypair13\})) \et
 \EUpair(\mypair23\mapsto 1; \{\mypair21,\mypair24\})\right)
\end{multline*}

Consider the first disjunct of
this formula:
\[
\EUpair(\mypair11\mapsto 2; \{\mypair13\}) \et
\EUpair(\mypair23\mapsto 1; \{\mypair21,\mypair24\}),
\]
and write $E_1$ for the multiset~$\mypair11\mapsto 2$,
$E_2$~for~$\mypair23\mapsto 1$, $U_1=\{\mypair13\}$ and
$U_2=\{\mypair21,\mypair24\}$.
\begin{itemize}
  \item In case $E_1$ and $E_2$
    do not overlap (i.e., for~$J_1=J_2=\emptyset$), the~state~$\mypair23$ of~$E_2$
    will be paired with the only state~$\mypair13$ of~$U_1$, and the
    two occurrences of~$\mypair11$ required to fulfill~$E_1$
    may be paired with one of the states~$\mypair21$
    and~$\mypair24$ of~$U_2$. For~this case we~obtain a disjunction of
    three \EUprs, each having
    $U'=\{\{\mypair13,\mypair21\},\{\mypair13,\mypair24\}\}$ as their
    second component, and having the following \kl{multisets} as
    their first component:
    \begin{xalignat*}1
    E'_1&= \{\mypair11,\mypair21\}\mapsto 2,
    \{\mypair23,\mypair13\}\mapsto 1 \\ E'_2&=
    \{\mypair11,\mypair21\}\mapsto 1,\{\mypair11,\mypair24\}\mapsto 1,
    \{\mypair23,\mypair13\}\mapsto 1\\ E'_3&=
    \{\mypair11,\mypair24\}\mapsto 2, \{\mypair23,\mypair13\}\mapsto 1.
    \end{xalignat*}
  \item
    otherwise, one of the two occurrences of~$\mypair11$ required by~$E_1$ will
    overlap with the state~$\mypair23$ of~$E_2$,
    the other
    occurrence of $\mypair11$ being paired with an element of~$U_2$.
    We~get a disjunction of
    two \EUprs, again having~$U'$ as their second component,
    and having the following \kl{multisets} as their first
    component: \begin{xalignat*}1
    E'_4&= \{\mypair11,\mypair23\}\mapsto
    1, \{\mypair11,\mypair21\}\mapsto 1 \\
    E'_5&= \{\mypair11,\mypair23\}\mapsto
    1, \{\mypair11,\mypair24\}\mapsto
    1.  \end{xalignat*} \end{itemize} For this example, the~resulting
    formula then is a disjunction of 5 \EUconstrs.
  \end{example}

In~order to prove correctness of this construction, we~establish
a correspondence between a
conjunction~$\EUpair(E_1;U_1) \et \EUpair(E_2; U_2)$ and its resulting
formula $C(\EUpair(E_1;U_1),\EUpair(E_2;U_2))$:
\begin{lemma}\label{lemma-QtoR}
Let $\Set$ and~$\Set'$ be two finite sets, and
$\EUpair(E_1;U_1) \et \EUpair(E_2; U_2)$ be a conjunction of two
\EUprs on~$2^\Set$. For any \kl{unitary marking}~$\nu$ of~$\Set'$ by~$2^\Set$, it~holds
$\nu\models* \EUpair(E_1;U_1) \et \EUpair(E_2; U_2)$ if, and only~if,
$\nu\models* C(\EUpair(E_1;U_1),\EUpair(E_2;U_2))$.
\end{lemma}

\begin{proof}
Assume that $\nu\models* \EUpair(E_1;U_1) \et \EUpair(E_2; U_2)$. Then
there exists two \kl{unitary submarkings}~$\nu_1$ and~$\nu_2$ such that
$\nu_1\models*\EUpair(E_1;U_1)$ and $\nu_2\models*\EUpair(E_2;U_2)$.
We~let $\nu_1\Cup\nu_2$ be the \kl{marking} such that $\nu_1\Cup\nu_2(s')
= \nu_1(s') \cup \nu_2(s')$. Then $\nu_1\Cup\nu_2$ is a \kl{unitary
submarking} of~$\nu$ by~$2^{\Set}$; we~now prove that $\nu_1\Cup\nu_2\models*
C(\EUpair(E_1;U_1),\EUpair(E_2;U_2))$.

For~$i\in\{1,2\}$, $\nu_i\models*\EUpair(E_i;U_i)$ means that
$\mkimg{\nu_i}\models \EUpair(E_i;U_i)$, which in turn means that there
exists a subset~$S'_i$ of~$\Set'$ such that $\nu_i(S'_i)=E_i$ and
${\supp(\nu_i(\Set'\setminus S'_i))\subseteq U_i}$. We~let $O=S'_1\cap
S'_2$ be the overlap between~$S'_1$ and~$S'_2$, $o=\size O$ be the
size of this overlap, and $\mset{ j^i_k \mid 1\leq k\leq o} = J_i
=\nu_i(O)$ be (multiset) images of~$O$ by~$\nu_i$. Then the multiset
$\nu_1\Cup\nu_2(O)$ corresponds to~$\mset{ j^1_k\cup j^2_{\tau(k)}\mid
1\leq k\leq o}$ for some permutation~$\tau$. Similarly,
letting~$H_i=S'_i\cap (\Set'\setminus S'_{3-i})$, the~elements
of~$\nu_1\Cup \nu_2(H_i)$ are unions of one set of~$E_i$ and one set
of~$U_{3-i}$, of the form $\mset{ e^i_k\cup g_i(e^i_k) \mid 1\leq
k\leq \size{E_i}-o}$ for some functions $g_i\colon E_i\msminus
J_i \to U_{3-i}$. Finally, any~$s'\in (\Set'\setminus S_1)\cap (\Set'\setminus S_2)$,
$\nu_1\Cup \nu_2(s')$ is the union of two sets in~$U_1$ and~$U_2$, respectively.
This shows that $\nu_1\Cup\nu_2\models* C(\EUpair(E_1;U_1),\EUpair(E_2;U_2))$. 

\medskip
The converse direction is similar: assuming that $\nu\models*
C(\EUpair(E_1;U_1),\EUpair(E_2;U_2))$, we~pick a
\kl{unitary submarking}~$\nu'$ such that
\[
\mkimg{\nu'} \models 
  \left(
  E'= \biguplus
  \begin{array}{l}
    \mset{ j^1_k \cup j^2_{\tau(k)} \mid 1\leq k\leq \size{J_1}} \\
    \mset{ e^1_k \cup g_1(e^1_k) \mid 1\leq k\leq n_1-\size{J_1}} \\
    \mset{ g_2(e^2_k) \cup e^2_k \mid 1\leq k\leq n_2-\size{J_1}}
  \end{array};
  U'=U_1 \otimes U_2
  \right)
\]
for some~$J_1=\mset{ j^i_k\mid 1\leq k\leq o} \submultiset E_1$ and
$J_2=\mset{ j^i_k\mid 1\leq k\leq o} \submultiset E_2$ of the same
size, some permutation~$\tau$ of~$[1;o]$, and some functions
$g_1\colon E_1\msminus J_1\to U_2$ and $g_2\colon E_2\msminus
J_2\to U_1$.
We~fix three disjoint subsets~$O$, $H_1$ and~$H_2$ of~$\Set'$ such that
\begin{xalignat*}1
\nu'(O)   &=\mset{ j^1_k \cup j^2_{\tau(k)} \mid 1\leq k\leq \size{J_1}} \\
\nu'(H_1) &=\mset{ e^1_k \cup g_1(e^1_k) \mid 1\leq k\leq \size{E_1}-\size{J_1}} \\
\nu'(H_2) &=\mset{ g_2(e^2_k) \cup e^2_k \mid 1\leq k\leq \size{E_2}-\size{J_1}} \\
\supp(\nu'(\Set'\setminus (O\cup H_1\cup H_2))) &\subseteq U_1\otimes U_2.
\end{xalignat*}
It~should be clear that from~$\nu'$ (hence also from~$\nu$), we~can extract two submarkings~$\nu_1$ and~$\nu_2$ such that
$\nu_1(O\cup H_1)=E_1$ and $\supp(\nu_1(\Set'\setminus (O\cup H_1)))\subseteq U_1$, and
$\nu_2(O\cup H_2)=E_2$ and $\supp(\nu_2(\Set'\setminus (O\cup H_2)))\subseteq U_2$.
This~proves that $\nu\models* \EUpair(E_1;U_1) \et \EUpair(E_2; U_2)$.
\end{proof}

As~a consequence, replacing~$\EUpair(E_1;U_1) \et \EUpair(E_2; U_2)$
with $C(\EUpair(E_1;U_1),\EUpair(E_2;U_2))$ in the transition function
of an \powAATA does not change the \kl{execution
trees}. Let~$\RAut$ be an \powAATA obtained from~$\QAut$ by replacing
all conjunctions $\EUpair(E_1;U_1) \et \EUpair(E_2; U_2)$ (in~no
specific order).
Then $\RAut$ is \kl{non-alternating}, and by Lemma~\ref{lemma-QtoR}:
\begin{proposition}\label{prop-QtoR}
The \kl{languages} accepted by the two \powAATAs~$\QAut$ and~$\RAut$ are equal.
\end{proposition}

Moreover, since~$\RAut$ is \kl{non-alternating}, it~only has to visit each
node of the input tree in one of the states given by the transition
function, so that both notions of \kl{execution trees} (with~$\models+$
and~$\models*$) coincide.

\paragraph{Size of~$\RAut$.} 
We~now evaluate the size
of~$\RAut=\tuple{2^{\State\times\State},\{(\initstate,\initstate)\},\RTrans
,\Accept_\prio}$.  In~order to evaluate the size of the transition
function~$\RTrans$ of~$\RAut$, we~first focus on the size of
$C(\EUpair(E_1;U_1),\EUpair(E_2;U_2))$.  For~this, we~define the size
of an
\EUpr~${\EUpair(E;U)}$ as the pair~${(\size E, \size U)}$. 
In~the following, a~set of
\EUprs is said to have a size at most~$(n,m)$ if all its \EUprs have
existential parts of size at most~$n$ and universal parts of size at
most~$m$.

Consider a conjunction $\EUpair(E_1;U_1)\et \EUpair(E_2;U_2)$ of two
\EUprs of size  $(n_1,m_1)$ and~$(n_2,m_2)$, respectively,
and assume w.l.o.g.~that ${n_1 \leq n_2}$. We~also assume that $m_1$ and~$m_2$ are positive.
Then the~formula ${C(\EUpair(E_1;U_1),\EUpair(E_2;U_2))}$~is a
disjunction of $N$~\EUprs of size at most~$(n_1+n_2,
m_1\cdot \allowbreak m_2)$, with:
\[
N \leq \sum_{l=0}^{n_1}
  {n_1 \choose l}\cdot{n_2 \choose l}\cdot l! \cdot m_2^{n_1-l} \cdot m_1^{n2-l}.
\]
This formula follows from the definition
of~$C(\EUpair(E_1;U_1),\EUpair(E_2;U_2))$: there~is one \EUpr
for~every possible size~$l$ of
the overlap of the existential parts (which cannot exceed~$n_1$),
every submultisets~$J_1 \submultiset E_1$ and $J_2\submultiset E_2$ of size~$l$, every bijection from~$J_1$ to~$J_2$ 
(so~as~to consider any possible combination in the overlap), and every
combination outside the overlap, between the remaining states of the
existential parts and the states of the universal parts.

This number~$N$ can then be overapproximated as follows:
\begin{xalignat*}1
N & \leq \biggl( \sum_{l=0}^{n_1} \; {n_1 \choose l}\cdot \frac{n_2!}{(n_2-l)! \cdot  l!}\cdot l! \biggl) \cdot m_2^{n_1} \cdot m_1^{n2} \\
  & \leq \biggl( \sum_{l=0}^{n_1} \; {n_1 \choose l}\cdot  n_2^{l} \biggl) \cdot m_2^{n_1} \cdot m_1^{n2} = (n_2+1)^{n_1} \cdot m_2^{n_1} \cdot m_1^{n2}
\end{xalignat*}

\begin{remark}\label{rk-five}
In~case $U_1$ and\slash or~$U_2$ are empty, the first
overapproximation of~$N$ simplifies. Assume for example that $U_2$ is
empty (which is consistent with our hypothesis that $n_1\leq
n_2$). As~noticed in Remark~\ref{rk-four}, we~must have $J_1=E_1$ in
that case, hence we only have to consider the case where $l=n_1$
(notice that for other values of~$l$, $m_2^{n_1-l}$~is~zero).  Then
$N\leq {n_2 \choose n_1} \cdot n_1!\cdot m_1^{n_2-n_1}\leq
n_2^{n_1}\cdot m_1^{n_2-n_2}$.

In~the sequel, we keep the previous bound $N\leq (n_2+1)^{n_1} \cdot
m_2^{n_1} \cdot m_1^{n2}$, and assume that $m_1\geq 1$ and $m_2\geq
1$, e.g. by letting~$m_1=\max(1,\size{U_1})$ (resp.~$m_2=1$) in case $U_1=\emptyset$
(resp.~$U_2=\emptyset$).
\end{remark}

We~now prove that any conjunction of $k$~\EUprs of sizes at
most~$(n,m)$ (where we assume $m\geq 1$) can be turned into disjunctions of at most
$\bigl((k-1)!\cdot (n+1)^{k-1}\cdot m^{k^2}\bigr)^n$ \EUprs, each of
size at~most~$(k\cdot n, m^k)$.  According
to our computation above, this result holds true for~$k=2$.

Consider a conjunction~$\mathcal C$ of~$k+1$ such \EUprs, assuming
that the result holds for up to~$k$ \EUprs. Then the conjunction of the
first $k$ \EUprs can be turned into a disjunction of at most
$((k-1)!\cdot (n+1)^{k-1}\cdot m^{k^2})^n$  
\EUprs of size at
most $(k\cdot n, m^k)$.  By~distributing the $(k+1)$-th conjunction over
this disjunction, we~obtain an expression of~$\calC$ as the
disjunction of at most ${((k-1)!\cdot (n+1)^{k-1}\cdot m^{k^2})^n}$
conjunctions of two \EUprs, of~sizes at~most
${(k\cdot n,m^k)}$ and~$(n,m)$ respectively.

We~apply our construction to each conjunction of two \EUprs. Each
such conjunction is then replaced with the disjunction of (at~most)
$(kn+1)^n\cdot  m^{kn}\cdot (m^k)^n$
\EUprs of size at~most~$((k+1)\cdot n, m^{k+1})$.
In~the~end, we~obtain an expression of~$\calC$ as a disjunction of at most 
$M$ \EUprs with
\begin{xalignat*}1
M &\leq \Bigl((k-1)!\cdot (n+1)^{k-1}\cdot m^{k^2}\Bigr)^n
\cdot (kn+1)^n\cdot  m^{kn}\cdot (m^k)^n\\
&\leq \bigl({k!}\cdot
(n+1)^{k}\cdot m^{(k+1)^2}\bigr)^n.
\end{xalignat*}

\medskip

We~now evaluate the size of the transition function~$\RTrans$: as
explained at the beginning of the present section, $\RTrans$~is
obtained from the transition function~$\QTrans$ of~$\QAut$ by first
putting each formula~$\QTrans(P,\alp)$ into disjunctive normal form,
as the disjunction of at most $2^{\size\State\cdot\sizeB\Trans}$
conjunctions of at~most $\size\State\cdot\sizeB\Trans$ \EUprs,
with \EUprs of size at most~$(\sizeE\Trans,\sizeU\Trans)$.

Applying our formula above (and assuming $\sizeU\Trans\geq 1$ as
explained in Remark~\ref{rk-five} above), we~get a disjunctive
expression for~$\QTrans(P,\alp)$ involving at most
\[
2^{\size\State\cdot\sizeB\Trans} \cdot
\Bigl(
  (\size\State\cdot\sizeB\Trans-1)! \cdot
  (\sizeE\Trans+1)^{\size\State\cdot\sizeB\Trans-1} \cdot
  (\sizeU\Trans)^{(\size\State\cdot\sizeB\Trans)^2}
\Bigr)^{\sizeE\Trans}
\]
\EUprs of size at most $(\size\State\cdot \sizeB\Trans\cdot \sizeE\Trans,
(\sizeU\Trans)^{\size\State\cdot\sizeB\Trans})$.
In~the~end:
\begin{proposition}\label{prop-simu-EUTA}
The languages of the original \AAPTA~$\Aut$ and of the resulting
\nAATA~$\RAut$ are the same. The~size of~$\RAut$ is at
most\footnote{We~omit the size of the acceptance condition of~$\RAut$
here as it is not a \kl{parity condition}.}
$(2^{\size\State^2},
2^{\QBBE} \cdot
\bigl[\QBBE! \cdot
  (\sizeE\Trans+1)^{\QBBE} \cdot
  (\max(\sizeU\Trans,1))^{\QBBE^2}\bigr]^{\sizeE\Trans},\allowbreak
\QBBE\cdot\sizeE\Trans,(\sizeU\Trans)^\QBBE,-)$
where $\QBBE=\size\State\cdot \sizeB\Trans$.
\end{proposition}

\makeatletter
\setcounter{save@cptr}{\value{example}}
\expandafter\ifx\csname r@ex-simu\endcsname\relax
\def\@tmp{{1}{}{}{}{}}
\else
\edef\@tmp{\csname r@ex-simu\endcsname}
\fi
\setcounter{example}{\expandafter\@firstoffive\@tmp}
\addtocounter{example}{-1}

\def\mypairline{\@ifnextchar*{\@pairlinestar}{\@pairline}}
\def\@pairlinestar*#1{(\star,\ifnum#1=0\relax\initstate\else\state_{#1}\fi)}
\def\@pairline#1#2{(\ifnum#1=0\relax\initstate\else\state_{#1}\fi,\ifnum#2=0\relax\initstate\else\state_{#2}\fi)}
\makeatother

\begin{example}[contd]%
  Consider again  Example~\ref{ex-simu}.
  We~describe the corresponding automaton~$\RAut$. The~initial state still is $\{ \mypairline00\}$.
  In~the following, we~use~$\star$ to represent any of the two states~$\initstate$ and~$q_1$. 
  The~previous construction provides the following transition function:
\begin{xalignat*}1
  \RTrans(\{ \mypairline*0 \},a) &=
     \EUpair(\{ \mypairline00,\mypairline01\}\mapsto 1; \emptyset)  \\
  \RTrans(\{ \mypairline*1 \},a) &=
     \EUpair(\{ \mypairline11\}\mapsto 1; \emptyset)  \\
  \RTrans(\{ \mypairline*0 \},b) &=
     \EUpair(\{ \mypairline01\}\mapsto 1; \emptyset)  \\
  \RTrans(\{ \mypairline*1 \},b) &=
     \EUpair(\{ \mypairline10\}\mapsto 1; \emptyset) \\
  \RTrans(\{ \mypairline*0,\mypairline*1 \},a) &=
     \EUpair(\{ \mypairline00,\mypairline01,\mypairline11\}\mapsto 1; \emptyset)  \\
  \RTrans(\{ \mypairline*0,\mypairline*1 \},b) &=
     \EUpair(\{ \mypairline01,\mypairline10\}\mapsto 1; \emptyset)
\end{xalignat*}

Note that the transition function from $\{\mypairline00,\mypairline01,\mypairline11\}$ is the same as the one from~$\{\mypairline*0,\mypairline*1\}$. 
Note also that this transition function does not involve any disjunction because the
$U$-parts of the transition function of~$\Aut$ all are empty.

We then obtain the \kl(pow){execution trees} of~$\RAut$ over
the 1-branch trees~$a^\omega$
 and $a^3\cdot b^\omega$,
 as depicted on Fig.~\ref{fig-exalt2}. They~(fortunately)
 correspond to the ones depicted to the right of
 Fig.~\ref{fig-exalt}.
\begin{figure}[ht]
\centering
\begin{tikzpicture}
\begin{scope}
\def\mypair#1#2{(\ifnum#1=0\relax\initstate\else\state_{#1}\fi,
    \ifnum#2=0\relax\initstate\else\state_{#2}\fi)}
\draw (0,0) node (0q0) {$\{\mypair 00\}$};
\draw (0,-1) node (1q01) {$\{\mypair 00, \mypair 01\}$};
\draw (0,-2) node (2q01) {$\{\mypair 00, \mypair 01, \mypair 10\}$};
\draw (0,-3) node (3q01) {$\{\mypair 00, \mypair 01, \mypair 10\}$};
\draw (0,-4) node (4q01) {$\{\mypair 01, \mypair 10\}$};
\draw (0,-5) node (5q01) {$\{\mypair 01, \mypair 10\}$};
\draw[-latex'] (0q0) -- (1q01);
\draw[-latex'] (1q01) -- (2q01);
\draw[-latex'] (2q01) -- (3q01);
\draw[-latex'] (3q01) -- (4q01);
\draw[-latex'] (4q01) -- (5q01);
\draw[dashed] (5q01) -- +(-90:8mm);
\begin{scope}[xshift=-2.8cm]
 \draw (0,-.5) node {$a$};
 \draw (0,-1.5) node {$a$};
 \draw (0,-2.5) node {$a$};
 \draw (0,-3.5) node {$b$};
 \draw (0,-4.5) node {$b$};
\end{scope}
\end{scope}
\begin{scope}[xshift=-6.2cm]
\def\mypair#1#2{(\ifnum#1=0\relax\initstate\else\state_{#1}\fi,
    \ifnum#2=0\relax\initstate\else\state_{#2}\fi)}
\draw (0,0) node (0q0) {$\{\mypair 00\}$};
\draw (0,-1) node (1q01) {$\{\mypair 00, \mypair 01\}$};
\draw (0,-2) node (2q01) {$\{\mypair 00, \mypair 01, \mypair 10\}$};
\draw (0,-3) node (3q01) {$\{\mypair 00, \mypair 01, \mypair 10\}$};
\draw (0,-4) node (4q01) {$\{\mypair 00, \mypair 01, \mypair 10\}$};
\draw (0,-5) node (5q01) {$\{\mypair 00, \mypair 01, \mypair 10\}$};
\draw[-latex'] (0q0) -- (1q01);
\draw[-latex'] (1q01) -- (2q01);
\draw[-latex'] (2q01) -- (3q01);
\draw[-latex'] (3q01) -- (4q01);
\draw[-latex'] (4q01) -- (5q01);
\draw[dashed] (5q01) -- +(-90:8mm);
\begin{scope}[xshift=-2.8cm]
\draw (0,-.5) node {$a$};
\draw (0,-1.5) node {$a$};
\draw (0,-2.5) node {$a$};
\draw (0,-3.5) node {$a$};
\draw (0,-4.5) node {$a$};
\end{scope}
\end{scope}
\end{tikzpicture}
\caption{Execution tree of~$\RAut$ on~$a^\omega$ and $a^3\cdot b^\omega$}
\label{fig-exalt2}
\end{figure}
 The \kl(pow){execution tree} on the left is not \kl(powET){accepting}:
 the~\kl{branch}~$\initstate^\omega$ does not satisfies the \kl{parity
 condition}. But~in~the execution tree on the right,
 states~$\initstate$ and~$q_1$ alternate along any sequence appearing
 in the unique branch, which ensures that the \kl(pow){execution tree} is
 \kl(powET){accepting}.
\end{example}
\makeatletter
\setcounter{example}{\value{save@cptr}}
\makeatother

\subsubsection{Adapting the acceptance condition}\label{ssec-backtoparity}

It finally remains to turn the acceptance condition of~$\RAut$ into a
\kl{parity condition}. The~transformation is the same as
in~\cite{Zan12}: we~first build a non-deterministic \kl{parity word
  automaton}~$\WAut$ accepting all \kl{words} on the
alphabet~$2^{\State\times\State}$ that contain an infinite sequence of
states~$(r_i)_{i\in\bbN}$ of~$\Aut$ (in~the sense of
Def.~\ref{def-exectree2}) \emph{not satisfying} the \kl{parity
  condition} of~$\Aut$; we~then turn~it into a deterministic \kl{parity word
automaton}, take its complement, and run~it in parallel with~$\RAut$.

\smallskip
Let $\WAut=\tuple{\State\cup\{\initstate'\}, \initstate', \WTrans,
  \Wprio}$ where $\initstate'$ is a new state not in~$\State$,
$\WTrans(\initstate',L)=\OU_{(\state',\state)\in L} \state$ and
$\WTrans(\state',L)=\OU_{(\state',\state)\in L} \state$, and
$\Wprio(\state)=\prio(\state)+1$ (the~value of $\Wprio(\initstate')$
can be set arbitrarily since that state is visited only~once).
Intuitively, this~automaton guesses a sequence of
states~$(r_i)_{i\in\bbN}$ contained in the input
\kl{word}~$(L_i)_{i\in\bbN}$ on alphabet~$2^{\State\times\State}$ and
the \kl{parity condition} of~$\WAut$ ensures that this sequence does
not satisfy the \kl{parity condition} of~$\Aut$. Note that the number
of \kl{priorities} remains unchanged and equal to~$\priomax$.

From $\WAut$, we can build an equivalent \emph{deterministic} \kl{parity word
automaton}~$\WAut_d$.  For~this, we~first turn~$\WAut$ into a
non-deterministic B\"uchi automaton~$\WAut'$ with at~most
$\size\State\cdot\priomax+1$ states. This~can be achieved by
considering several copies of~$\WAut$: an~initial one, with no
accepting states, and for each even integer~$p$ less than or
equal~to~$\size\Wprio$, one~copy of~$\WAut$ involving only those states
with priority larger than or equal to~$p$, with exactly those states
of priority~$p$ being accepting (for the B\"uchi condition).

We~then apply the construction of~\cite{Pit07} to get a
deterministic \kl{parity word automaton} with $2\cdot
(\size\State\cdot\priomax+1)^{\size\State\cdot\priomax+1}\cdot
(\size\State\cdot\priomax+1)!$ states and at most $2\cdot
(\size\State\cdot\priomax+1)$~\kl{priorities}.
The~number of states can then be bounded by
$2^{1+2(\size\State\cdot\priomax+1)\cdot\log(\size\State\cdot\priomax+1)}$.

It~remains to complement~$\WAut_d$, in order to get an
automaton~$\MAut$ that recognises precisely all input \kl{words}
containing only sequences of states satisfying the \kl{parity condition}
of~$\Aut$: these are
precisely the \kl{branches} that~$\RAut$ has to
accept. Complementing~$\WAut_d$ is easy, as it consists in
incrementing its \kl{priorities} by~$1$
(leaving the number of \kl{priorities} unchanged),
and the resulting automaton (namely~$\MAut$) is still
deterministic. Let~$\Mprio$ be the \kl{priority function} of~$\MAut$.
The~number of states of~$\MAut$ is then bounded by
$2^{1+2(\size\State\cdot\priomax+1)\cdot\log(\size\State\cdot\priomax+1)}$
and the number of \kl{priorities} is at~most ${2\cdot
(\size\State\cdot\priomax+1)}$.

We~can
then run~$\RAut$ and~$\MAut$ in parallel, thereby obtaining a
\kl{non-alternating} \nAAPTA~$\NAut$ \kl{accepting} the same \kl{language} as~$\Aut$. 
The~construction is as follows: the states of~$\NAut$ are
pairs~$(\overline{\state}, \state)$ where $\overline\state$ is a state
of~$\RAut$ and~$\state$~is a state of~$\MAut$. The~transition
function~$\NTrans$ of~$\NAut$ is defined as follows: take a
state~$(\overline\state,\state)$ and a letter~$\alp$,
and write $\RTrans(\overline\state,\alp) =
\OU_{i\in I} \EUpair(E_i; U_i)$ for some set~$I$.
Then $\NTrans((\overline\state,\state),\alp)=\OU_{i\in
I}\EUpair(E'_i;U'_i)$, where each state~$(\overline\state')$ in~$E_i$ and
in~$U_i$ is replaced with
$(\overline\state',\delta_{\MAut}(\state,\overline\state'))$, so that
when $\RAut$ explores some \kl{successor node} in state~$\overline\state'$,
the~state of~$\MAut$ is updated accordingly.  By~letting
the~\kl{priority} of $(\overline{\state},\state)$ be that of~$\state$
in~$\MAut$, we~make $\MAut$~keep track of whether all the sequences
of state of~$\Aut$ that \kl{appear} along each \kl{branch} of the
\kl(pow){execution tree} of~$\RAut$ are
\kl{accepting} for the \kl{parity condition} of~$\Aut$.

The size of the state space of the product automaton~$\RAut \times \MAut$ is at most
$2^{\size\State^2} \cdot 2^{1+2(\size\State\cdot\priomax+1)\cdot\log(\size\State\cdot\priomax+1)}$,
which is in $\size\State^{O(\size\State^2)}$ (assuming $\size\State\geq 2$).
The~sizes of the transition function and of the \EUconstrs are the
same as those of~$\RAut$, and the number of priorities is the same as
for~$\MAut$.

Summarising our results:
\begin{theorem}\label{thm-simu}
Given an \AAPTA $\Aut=\tuple{\State,\initstate,\Trans,\omega}$, we can build an \nAAPTA~$\NAut$ 
recognising the same language. 
The~\kl(A){size} of~$\NAut$~is  bounded by 
$(
\size\State^{O(\size\State^2)},
2^{\QBBE} \cdot
\bigl[\QBBE! \cdot
  (\sizeE\Trans+1)^{\QBBE} \cdot
  (\max(\sizeU\Trans,1))^{\QBBE^2}\bigr]^{\sizeE\Trans},\allowbreak
\QBBE\cdot\sizeE\Trans, 
\sizeU\Trans^{\QBBE},
2(\size\State\cdot\priomax[\omega]+1))$,
where $\QBBE=\size\State\cdot \sizeB\Trans$.
\end{theorem}

\begin{proof}
The correctness comes from previous propositions. 
The number of states of $\NAut$ can be estimated at $2^{1+\size\State^2+2(\size\State\cdot\priomax+1)\cdot\log(\size\State\cdot\priomax+1)}$.
Assuming $\priomax$ in $O(\size\State)$ and $\size\State\geq 2$, we get the upper bound.
\end{proof}

\makeatletter
\setcounter{save@cptr}{\value{example}}
\expandafter\ifx\csname r@ex-simu\endcsname\relax
\def\@tmp{{1}{}{}{}{}}
\else
\edef\@tmp{\csname r@ex-simu\endcsname}
\fi
\setcounter{example}{\expandafter\@firstoffive\@tmp}
\addtocounter{example}{-1}

\def\mypairline{\@ifnextchar*{\@pairlinestar}{\@pairline}}
\def\@pairlinestar*#1{(\star,\ifnum#1=0\relax\initstate\else\state_{#1}\fi)}
\def\@pairline#1#2{(\ifnum#1=0\relax\initstate\else\state_{#1}\fi,\ifnum#2=0\relax\initstate\else\state_{#2}\fi)}
\makeatother

\begin{example}[contd]
We apply the approach above to the automaton of Example~\ref{ex-simu}.
We build a deterministic automaton
$\calM=\tuple{\State',\initstate',\pi',\omega'}$ which recognises the
infinite words over the alphabet $2^{\State\times\State}$ that satisfy
the parity condition $\omega$ (in the sense of
Def.~\ref{def-exectree2}).
Note that we build~$\calM$ directly here, without using~$\WAut$,
because the parity condition we consider turns out to be equivalent to
a B\"uchi condition, which makes the construction simpler.
The~states of~$\MAut$  are pairs~$(s,s')$ with $s,s'\subseteq \State$:
$s\cup s'$ is the set of all possible last states of sequences of states that
\kl{appear} in the word~$w\in (2^{\State\times\State})^*$ that has been read by~$\MAut$;
the~states in~$s'$ have \emph{recently} visited an accepting state
(with priority~$0$), while those in~$s$ have~not.
More formally, when reading a letter $\sigma \in
2^{\State\times\State}$, the~transition function updates~$(s,s')$
into~$(t,t')$ by transforming each state~$q$ of~$s$ or~$s'$ into a
state~$q'$ in~$t$ or~$t'$, for each $(q,q')\in\sigma$. All accepting
states~$q'$ are placed in~$t'$; non-accepting states~$q'$ are placed
in~$t$ if they originate from a state~$q$ in~$s$, or if~$s$ was empty,
otherwise they are placed in~$t'$.  States of the form
$(\emptyset,s')$ are accepting. This corresponds to the classical
procedure to transform an alternating B\"uchi automaton into a
non-deterministic one.
 
In our case of Example~\ref{ex-simu}, we~even get a deterministic automaton:
\begin{itemize}
\item $\State'= 2^{\{\initstate\}}\times 2^{\{\initstate, \state_1 \}}$ and $\initstate'= (\{\initstate\},\emptyset)$;
\item $\omega'((\emptyset,s'))=0$ and $\omega'((\{\initstate\},s'))=1$ for any $s'\subseteq \State$;
\item Given $s\in 2^\State$ and $\sigma \in 2^{\State\times\State}$, we use $\sigma(s)$ to denote
$\{q' \mid {\exists (q,q') \in \sigma} \et {q\in s}\}$; we~then define
the transitions as follows:
\begin{xalignat*}1
\pi((\emptyset,s'),\sigma) & = (\sigma(s')\cap \{\initstate\},\sigma(s')\setminus \{\initstate\}) \\
\pi((\{\initstate\},s'),\sigma) & = (\sigma(\{\initstate\})\cap \{\initstate\},
\sigma(\{\initstate\})\setminus \{\initstate\} \cup \sigma(s'))
\end{xalignat*}

\end{itemize}

Now we can define $\NAut$ as the product $\RAut\times\calM$.  The
states are pairs $(\bar{q},q')$ wit $\bar{q}\in
2^{\State\times\State}$ and $q' \in 2^{\{\initstate\}}\times
2^\State$. The initial state is
$(\{\mypair00\},(\{\initstate\},\emptyset))$. We~simplify the
transition function by removing rejecting states and we present only
the reachable part (from the initial state) of the relation.  We~get:

\begin{xalignat*}1
\RTrans((\{ \mypair*0 \},(\{\initstate\},\emptyset)),a)  &=  \EUpair( (\{ \mypair00,\mypair01\},(\{\initstate\},\{q_1\}))\mapsto 1; \emptyset)  \\[1.2ex]
 \RTrans((\{ \mypair*0 \},(\{\initstate\},\emptyset)),b)  &=  \EUpair( (\{ \mypair01\},(\emptyset,\{q_1\}))\mapsto 1; \emptyset)  \\[1.2ex]\noalign{\allowbreak}
 \RTrans((\{ \mypair*1 \},(\emptyset,\{q_1\})),a)  &=  \EUpair( (\{ \mypair11\},(\emptyset,\{q_1\}))\mapsto 1; \emptyset)  \\[1.2ex]\noalign{\allowbreak}
 \RTrans((\{ \mypair*1 \},(\emptyset,\{q_1\})),b)  &=  \EUpair( (\{ \mypair10\},(\{\initstate\},\emptyset))\mapsto 1; \emptyset)  \\[1.2ex]\noalign{\allowbreak}
\RTrans((\{ \mypair*0,\mypair*1 \},(\{\initstate\},\{q_1\})),a)  &=  \EUpair((\{ \mypair00,\mypair01,\mypair11\},(\{\initstate\},\{q_1\}))\mapsto 1; \emptyset)   \\[1.2ex]\noalign{\allowbreak}
  \RTrans((\{ \mypair*0,\mypair*1 \},(\{\initstate\},\{ q_1\})),b)  &=  \EUpair((\{ \mypair01,\mypair10\},(\emptyset,\{\initstate,q_1\})) \mapsto 1; \emptyset)   \\[1.2ex]\noalign{\allowbreak}
  \RTrans((\{ \mypair*0,\mypair*1 \},(\emptyset,\{\initstate,q_1\})),a)  &=  \EUpair((\{ \mypair00, \mypair01,\mypair11\},(\{\initstate\},\{q_1\})) \mapsto 1; \emptyset)   \\[1.2ex]
  \RTrans((\{ \mypair*0,\mypair*1 \},(\emptyset,\{\initstate,q_1\})),a)  &=  \EUpair((\{ \mypair01,\mypair10\},(\{\initstate\},\{q_1\})) \mapsto 1; \emptyset) 
\end{xalignat*}

It remains to compare the \kl{execution trees} of~$\NAut$ over the
 word $a^\omega$ (which does not belong to the \kl{language} of~$\Aut$) and
 the word $a^3\cdot b^\omega$ (which is \kl{accepted}
 by~$\Aut$). Figure~\ref{fig-exalt3} displays both execution trees;
 we~observe that the \kl{execution tree} on the left is
 not \kl(ET){accepting} (assuming that it continues
 reading~$a^\omega$), while the \kl{execution tree} on the right
 is \kl(ET){accepting} (if it continues reading~$b^\omega$).
\end{example}
\begin{figure}[ht]
\centering
\def\RMstate#1#2{\genfrac{}{}{0pt}{0}{#1}{#2}}
\begin{tikzpicture}[yscale=1.4]
  \begin{scope}
\def\mypair#1#2{(\ifnum#1=0\relax\initstate\else\state_{#1}\fi,
    \ifnum#2=0\relax\initstate\else\state_{#2}\fi)}
\draw (0,0) node (0q0) {$\RMstate{\{\mypair 00\}}{(\{\initstate\},\emptyset)}$};
\draw (0,-1) node (1q01) {$\RMstate{\{\mypair 00, \mypair 01\}}{(\{\initstate\},\{\state_1\})}$};
\draw (0,-2) node (2q01) {$\RMstate{\{\mypair 00, \mypair 01, \mypair 10\}}{(\{\initstate\},\{\state_1\})}$};
\draw (0,-3) node (3q01) {$\RMstate{\{\mypair 00, \mypair 01, \mypair 10\}}{(\{\initstate\},\{\state_1\})}$};
\draw (0,-4) node (4q01) {$\RMstate{\{\mypair 01, \mypair 10\}}{(\emptyset,\{\initstate,\state_1\})}$};
\draw (0,-5) node (5q01) {$\RMstate{\{\mypair 01, \mypair 10\}}{(\{\initstate\},\{\state_1\})}$};
\draw (0,-6) node (6q01) {$\RMstate{\{\mypair 01, \mypair 10\}}{(\emptyset,\{\initstate,\state_1\})}$};
\draw[-latex'] (0q0) -- (1q01);
\draw[-latex'] (1q01) -- (2q01);
\draw[-latex'] (2q01) -- (3q01);
\draw[-latex'] (3q01) -- (4q01);
\draw[-latex'] (4q01) -- (5q01);
\draw[-latex'] (5q01) -- (6q01);
\draw[dashed] (6q01) -- +(-90:8mm);
\begin{scope}[xshift=-2.7cm]
 \draw (0,-.5) node {$a$};
 \draw (0,-1.5) node {$a$};
 \draw (0,-2.5) node {$a$};
 \draw (0,-3.5) node {$b$};
 \draw (0,-4.5) node {$b$};
 \draw (0,-5.5) node {$b$};
\end{scope}
\end{scope}
\begin{scope}[xshift=-6.2cm]
\def\mypair#1#2{(\ifnum#1=0\relax\initstate\else\state_{#1}\fi,
    \ifnum#2=0\relax\initstate\else\state_{#2}\fi)}
\draw (0,0) node (0q0) {$\RMstate{\{\mypair 00\}}{(\{\initstate\},\emptyset)}$};
\draw (0,-1) node (1q01) {$\RMstate{\{\mypair 00, \mypair 01\}}{(\{\initstate\},\{\state_1\})}$};
\draw (0,-2) node (2q01) {$\RMstate{\{\mypair 00, \mypair 01, \mypair 10\}}{(\{\initstate\},\{\state_1\})}$};
\draw (0,-3) node (3q01) {$\RMstate{\{\mypair 00, \mypair 01, \mypair 10\}}{(\{\initstate\},\{\state_1\})}$};
\draw (0,-4) node (4q01) {$\RMstate{\{\mypair 00, \mypair 01, \mypair 10\}}{(\{\initstate\},\{\state_1\})}$};
\draw (0,-5) node (5q01) {$\RMstate{\{\mypair 00, \mypair 01, \mypair 10\}}{(\{\initstate\},\{\state_1\})}$};
\draw (0,-6) node (6q01) {$\RMstate{\{\mypair 00, \mypair 01, \mypair 10\}}{(\{\initstate\},\{\state_1\})}$};
\draw[-latex'] (0q0) -- (1q01);
\draw[-latex'] (1q01) -- (2q01);
\draw[-latex'] (2q01) -- (3q01);
\draw[-latex'] (3q01) -- (4q01);
\draw[-latex'] (4q01) -- (5q01);
\draw[-latex'] (5q01) -- (6q01);
\draw[dashed] (6q01) -- +(-90:8mm);
\begin{scope}[xshift=-2.7cm]
\draw (0,-.5) node {$a$};
\draw (0,-1.5) node {$a$};
\draw (0,-2.5) node {$a$};
\draw (0,-3.5) node {$a$};
\draw (0,-4.5) node {$a$};
\draw (0,-5.5) node {$a$};
\end{scope}
\end{scope}
\end{tikzpicture}
\caption{Execution tree of~$\NAut=\RAut\times\MAut$
  on~$a^\omega$ and $a^3\cdot b^\omega$: each state of~$\NAut$ is a
pair made of one state of~$\RAut$ (which is a set of pairs of states
of~$\Aut$) and one state of~$\MAut$ (which is a pair of sets of states
of~$\Aut$)}
\label{fig-exalt3}
\end{figure}

\subsection{Summary of the complexity of operations on \AAPTAs}
\label{ssec-ops-concl}
\label{ssec-summary}
Table~\ref{tab-summary} gathers our results, giving the \kl(A){size}
of the resulting \AAPTAs depending on the \kl(A){size} of the \AAPTAs
given in input. 

\begin{table}[ht]
  \centering
  \def\arraystretch{1.1}
  \begin{tabular}{|l|r@{}l|}
    \hline
   intersection   & $\size{\State_\cap}, \size{\State_\cup}$ & $ {}\leq\size\State+\size{\State'}+1$  \\
    union  & $\sizeB{\Trans_\cap}, \sizeB{\Trans_\cup}$&${}\leq \sizeB\Trans+\sizeB{\Trans'}+1$ \\
   (Thm~\ref{thm-union})  & $\sizeE{\Trans_\cap}, \sizeE{\Trans_\cup}$&${}\leq \max(\sizeE\Trans,\sizeE{\Trans'})$ \\
   & $\sizeU{\Trans_\cap},\sizeU{\Trans_\cup}$&${}\leq \max(\sizeU\Trans,\sizeU{\Trans'})$ \\
          & $\priomax[\omega_\cap],\priomax[\omega_\cup] $&${}\leq \max(\priomax[\omega],\priomax[\omega'])+1$ \\
    \hline
    projection & $\size{\State_{\textsf{proj}}} $&${}\leq \size\State$ \\
    (Thm~\ref{thm-proj})  & $\sizeB{\Trans_{\textsf{proj}}} $&${}\leq \size{\Alp'}\cdot \sizeB\Trans$ \\
               & $\sizeE{\Trans_{\textsf{proj}}} $&${}\leq \sizeE\Trans$ \\
               & $\sizeU{\Trans_{\textsf{proj}}} $&${}\leq \sizeU\Trans$ \\
               & $\size{\omega_{\textsf{proj}}} $&$ {}\leq\priomax[\omega]$ \\
    \hline
    complement & $\size{\State^{c}} $&${}\in 
    O( \size\State^2 \cdot\sizeB\Trans\cdot \size\Alp\cdot  3^{\sizeE\Trans})$ \\
   (Thm~\ref{thm-compl})   & $\sizeB{\Trans^{c}} $&${}\in 
   O(\size\State\cdot (1+\sizeU\Trans) \cdot \sizeB\Trans \cdot (1+\sizeE\Trans^2\cdot 4^{\sizeE\Trans}))$ \\
    & $\sizeE{\Trans^{c}} $&${}\leq \sizeE\Trans+1$ \\
    & $\sizeU{\Trans^{c}} $&${}\leq \max(\sizeU\Trans,1)$ \\
    & $\priomax[\omega^{c}] $&$ {}\leq \priomax[\omega]+1$ \\
\hline
    simulation & $\size{\State^{s}} $& ${}\in  2^{O(\size\State^2\log(\size\State))}$ \\
   (Thm~\ref{thm-simu})   & $\sizeB{\Trans^{s}} $&${}\leq
2^{\QBBE} \cdot
\bigl[\QBBE! \cdot
  (\sizeE\Trans+1)^{\QBBE} \cdot
  (\max(\sizeU\Trans,1))^{\QBBE^2}\bigr]^{\sizeE\Trans}$ \\
 && \hfill with $\QBBE = \size\State\cdot\sizeB\Trans$ \\
               & $\sizeE{\Trans^{s}} $&${}\leq \size\State\cdot\sizeB\Trans\cdot\sizeE\Trans$\\  
               & $\sizeU{\Trans^{s}} $&${}\leq  \sizeU\Trans^{\size\State\cdot\sizeB\Trans}$ \\
               & $\size{\omega^{s}} $&$ {}\leq 2(\size\State\cdot\priomax[\omega]+1)$ \\
\hline
  \end{tabular}
  \caption{Bounds on the size of the automata obtained by our algorithms}
  \label{tab-summary}\label{tab-ops}
\end{table}

\section{Algorithms for \AATAs}

Given some \AATA $\Aut$, we are interested in two decision procedures:
deciding whether a \kl{regular} \kl{tree} belongs to $\Lang(\Aut)$ and
deciding whether $\Lang(\Aut)=\emptyset$. Both consist in building a
\kl{parity game} and deciding whether Player~$0$ has a \kl(S){winning}
\kl{strategy}. For this we use the following results:

\begin{proposition}[\cite{Lod21}]
\label{prop-compl-pg}
Solving a finite \kl{parity game} can be done in time
$O(n^d)$ or in time $n^{O(\log(d))}$, 
where $n$~is the number of states of the~game,
and $d$~is the number of \kl{priorities}.
\end{proposition}

Considering that~$d$ is usually small, we~will mainly use the former
result (namely~$O(n^d)$) in the sequel. Notice that better complexity
results have been obtained recently for parity games~\cite{CJKLS17},
but they would make complexity results even harder to read, without
significantly improving them.

\subsection{Membership checking}

Let $\calK=\tuple{V,E,\ell}$ be a \kl{Kripke structure}. Deciding
whether $\Tree_\calK \in \Lang(\Aut)$ is equivalent to deciding
whether Player~$0$ has a \kl(S){winning} \kl{strategy} in the
\kl{parity game}~$\GAK$ defined in Section~\ref{sec:sem-game}.
Remember that the number of states
of~$\GAK$ is in 
$O(\size{V}\cdot( \size\State\cdot\sizeB{\Trans}+\size{\State}^{\size{V}}))$
and the number of priorities is the same as~$\Aut$. 
Using Prop.~\ref{prop-compl-pg}:

\begin{theorem}
\label{theo-mc-aapta}
Deciding whether a \kl{regular} \kl{tree}~$\Tree_K$, defined by a
finite \kl{Kripke structure}~$\calK=\tuple{V,E,\ell}$, is accepted by
an \AAPTA $\Aut=\tuple{\State, \state_0,\Trans,\prio}$ can be
performed in time $O((\size
V\cdot(\size\State\cdot\sizeB\Trans+\size\State^{\size V}))^{\priomax})$.
\end{theorem}

\subsection{Emptiness checking}

Before addressing emptiness checking, we~prove a side result proving
that, if the language of an \AAPTA is non-empty, then it~contains a
tree whose arity can be bounded. The technique used in the proof of
this \emph{sufficient-degree} theorem is a first step in the
correctness proof of our emptiness-checking algorithm.

\begin{theorem}\label{thm-sufficientdeg}
Let $\Aut=\tuple{\State, \state_0,\Trans,\prio}$ be an \nAAPTA. Then
$\Lang(\Aut)$~is non-empty if, and only~if, it~contains a tree whose
arity is at~most~$\sizeE{\Trans}$.
\end{theorem}

\begin{proof}
Take a tree~$\Tree=\tuple{\tree,\lab}$ accepted by~$\Aut$, and an
\kl(ET){accepting} \kl{execution tree} $\ExTree=\tuple{\tree',\exlab}$
in which $\tree'$ is (isomorphic~to) a subset of~$\tree$
(see~Remark~\ref{rk-nonalt}).

In~each node of~$\ExTree$, some~\EUpr $\EUpair(E;U)$ of the
(disjunctive) transition function is satisfied, and the execution tree
can be pruned in such a way that the \EUpr $\EUpair(E;\emptyset)$ is
satisfied in that node; notice that the \EUpr~$\EUpair(E;U)$ is still
satisfied, but we have just enough nodes for satisfying the $\textsf E$-part,
and no nodes satisfy the $\textsf U$-part.  We~end~up with a new
tree~$\ExTree'=\tuple{\tree'',\exlab'}$, where $\tree''$ has been
obtained from~$\tree'$ by pruning subtrees (hence it is also
isomorphic to a subset of~$\tree$), and whose arity is at
most~$\sizeE{\Trans}$.  This~tree~$\ExTree'$ is easily seen to be an
\kl(ET){accepting} \kl{execution tree} of~$\Aut$ on the input
tree~$\Tree'=\tuple{\tree'',\lab_{|\tree''}}$.
\end{proof}

Using the simulation theorem, we can get a similar result for
alternating automata:
\begin{corollary}
  Let $\Aut=\tuple{\State, \state_0,\Trans,\prio}$ be an \AAPTA. Then
  $\Lang(\Aut)$~is non-empty if, and only~if, it~contains a tree whose
  arity is at~most~$\size\State\cdot \sizeB{\Trans}\cdot
  \sizeE{\Trans}$.
\end{corollary}

We now address emptiness checking:

 \begin{theorem}
Let $\Aut=\tuple{\State, \state_0,\Trans,\prio}$ be an \nAAPTA.
Checking emptiness of $\Lang(\Aut)$ can be performed in time
$O((\size \State\cdot (1+\sizeB\Trans\cdot\size\Alp))^{\priomax})$.
\end{theorem}
 
\begin{proof}
Continuing the proof of Theorem~\ref{thm-sufficientdeg}, we~consider
the automaton~$\widehat\Aut$ obtained from~$\Aut$ by replacing
each \EUpr~$\EUpair(E;U)$ with the \EUpr~$\EUpair(E;\emptyset)$:
the~language of this automaton is non-empty if, and only~if, the
language of~$\Aut$ also is non-empty.

From $\widehat{\Aut}=\tuple{\State, \state_0,\Trans',\prio}$, we build
the parity game $\GA=\tuple{\GState_0\cup \GState_1,R,\gprio}$
where:
\begin{itemize}
\item $\GState_0$ is $Q$ and $\GState_1$ is the set of all \EU-pair
  $\EUpair(E;\emptyset)$ appearing in some $\Trans'(q,\sigma)$.
\item $R$ contains two kinds of edges:
  first, it~contains an edge $(q,\EUpair(E;\emptyset))$ if, and only~if,
  $\EUpair(E;\emptyset)$ occurs in the disjunction $\Trans'(q,\sigma)$
  for some $\sigma \in \Sigma$; second, it~contains an edge
  $(\EUpair(E;\emptyset),q)$ if, and only~if, $q \in E$.
\item $\prio'(q)=\prio(q)$ for all~$q\in Q$,
  and $\prio'(\EUpair(E;\emptyset))$ is set
  to the maximum value of~$\prio$ on~$\State$
  (in~order to have no effect on the
  acceptance of the~play).
\end{itemize}
In this game, in state~$q$, Player~$0$ has to select
an~\EUpr~$\EUpair(E;\emptyset)$ in $\Trans'(q,\alp)$ for
some~$\alp\in\Alp$, and from a node $\EUpair(E;\emptyset)$, Player~$1$
may select any state in~$E$ to continue the play. This~way, Player~$0$
builds a \Stree{\Alp} step-by-step.

It remains to show that there exists some \Stree{\Alp}~$\Tree$
\kl{accepted} by $\widehat{\Aut}$ if, and only~if, Player~$0$ has a
\kl(S){winning} \kl{strategy} in~$\GA$ from~$\state_0$. The~proof is
based on a direct correspondence between the \kl(ET){accepting}
\kl{execution tree} for~$\Tree$ and a \kl(S){winning} \kl{strategy}
for Player~$0$.  Indeed consider an \kl(ET){accepting} \kl{execution
  tree} for~$\Tree$: this~\kl{tree} has the same structure as~$\Tree$,
and it~associates with every node~$n$ of~$\Tree$ a state~$\state\in \State$
such that the successor nodes of~$n$ satisfy some
$\EUpair(E;\emptyset) \in \Trans'(\state,\alp)$. This~pair
$\EUpair(E;\emptyset)$ is precisely the move Player~$0$ should select
to win from~$\state$ (and whatever the choice of Player~$1$ in~$E$, Player~$0$
will be able to continue to select winning moves).

Conversely given a \kl(S){winning} \kl{strategy}
for Player~$0$ from a state~$\state$,
one can build an \kl{accepted} \kl{tree} level-by-level: any~move to some
$\EUpair(E;\emptyset)$ corresponds to some (possibly several) 
letter~$\alp\in\Alp$  s.t.~$\EUpair(E;\emptyset) \in \delta(q,\sigma)$);
this~fixes the arity of the current node to
$\size E$, and associates its $\size E$ successors with the states in~$E$.

\smallskip
The number of
states of~$\GA$ is bounded by $\size\State + \size\State \cdot
\sizeB\Trans \cdot \size\Sigma$;
the~number of \kl{priorities} is $\priomax[\prio]$. By~using
Prop.~\ref{prop-compl-pg}, we get a decision procedure in
$O((\size \State\cdot (1+\sizeB\Trans\cdot\size\Alp))^{\priomax})$.
\end{proof}

Emptiness checking for \AAPTA can then be decided by first using the
simulation theorem to get a non-alternating automaton (with an
exponential \mbox{blow-up}). Then:
\begin{corollary}\label{coro-algo-aeupta}
  Let $\Aut=\tuple{\State, \state_0,\Trans,\prio}$ be an \AAPTA.
  Checking emptiness of~$\Lang(\Aut)$ can be performed in
  time~$\bigl(\size\Alp\cdot
(\size\State\cdot(\sizeE\Trans+1))^{\size\State^2\sizeB\Trans^2\cdot\sizeE\Trans}\bigr)^{O(\size\State\cdot\size\prio)}$.
  
\end{corollary}

\begin{proof}
From the results of Table~\ref{tab-summary}, after removing
alternation, we~can bound
\begin{itemize}
\item $\size{\State^s}^{\size{\prio^s}}$ with $\size\State^{O(\size\State^3\cdot\size\prio)}$;
\item $\sizeB{\Trans^s}^{\size{\prio^s}}$ with $((\sizeE\Trans+1)\cdot\max(\sizeU\Trans,1))^{O(\size\State^3\cdot\size\prio\cdot \sizeB\Trans^2\cdot\sizeE\Trans)}$;
\item $\size\Alp^{\size{\prio^s}}$ with $\size\Alp^{O(\size\State\cdot\size\prio)}$.
\end{itemize}
Since $\sizeU\Trans \leq \size\State$, 
we~get the expected bound.
\end{proof}

\subsection{Summary of algorithm complexities for \AAPTAs}
\label{ssec-algo-concl}

Table~\ref{tab-summary-algos} gathers the complexity results proven above. 
Note that the complexity of the membership problem for~\nAAPTAs does
not depend and $\sizeB\Trans$, unlike for~\AAPTAs. Indeed,
the~corresponding game~$\GATk$ can be simplified when the transition
function is non-alternating: from~a state~$(n,\state)$, Player~$0$ may
directly choose a configuration~$(n,\nu_n)$ s.t.~$\nu_n$~fulfills
some \EUpr~$\EUpair(E_i;U_i)$ in~$\Trans(q,\lab(n))$, and then Player~$1$ plays
from that state~$(n,\nu_n)$. This~makes the game slightly smaller than
in the general case.

\begin{table}[ht]
  \centering
  \def\arraystretch{1.1}
  \begin{tabular}{|c|c|c|}
    \hline
     & \nAAPTA & \AAPTA \\
     & (non-alternating) & (alternating) \\
     \hline
   Membership checking   & \multirow{2}*{$O \bigl((\size
V\cdot \size\State^{\size V})^{\priomax}\bigr)$}  &  \multirow{2}*{$O\bigl((\size
V\cdot(\size\State\cdot\sizeB\Trans+\size\State^{\size V}))^{\priomax}\bigr)$}  \\
$\Tree_\calK \in \Lang(\Aut) ?$   &  &    \\
   \hline 
   Emptiness checking & \multirow{2}*{$O\bigl((\size \State\cdot (1+\sizeB\Trans\cdot\size\Alp))^{\priomax}\bigr)$} &  \multirow{2}*
   {$\bigl(\size\Alp\cdot
(\size\State\cdot(\sizeE\Trans+1))^{\size\State^2\sizeB\Trans^2\cdot\sizeE\Trans}\bigr)^{O(\size\State\cdot\size\prio)}$}\\  
   $ \Lang(\Aut)=\emptyset ?$ & & \\
   \hline
    \end{tabular}
  \caption{Complexity of the algorithms for \nAAPTAs and \AAPTAs}
  \label{tab-summary-algos}
  \label{tab-ops-algos}
\end{table}

\section{Application to \QCTL}
\label{sec-applitoqctl}

\QCTL extends the temporal logic \CTL with quantifications over atomic
propositions.
In this section, we~establish a tight link between \QCTL and
\EU-automata, from
which we obtain expressiveness and algorithmic results for~\QCTL.

\subsection{Syntax and (tree) semantics}
\label{sec-semqctl}

\begin{definition}
\label{def-QCTL}
The syntax of \intro*{\QCTL} over a finite set~$\AtP$ of atomic propositions
is defined by the following grammar:
\begin{xalignat*}1
\reintro*\QCTL\ni \phi,\psi  &\coloncolonequals  q \mid \neg\phi \mid \phi\ou\psi  
    \mid \EX \phi \mid \AX \phi \mid \Ex \phi \Until \psi \mid \All \phi \Until \psi \mid \exists p.\ \phi
\end{xalignat*}
where $q$ and~$p$ range over~$\AtP$.
The~\intro{size of a formula}~$\phi\in\QCTL$, denoted $\reintro*\Fsize\phi$, is the number of
steps needed to build~$\phi$.
\intro*\CTL~is the restriction of~\QCTL in which the rule $\exists p.\ \phi$ is
not allowed. 
\end{definition}

\QCTL formulas are evaluated over \kl{trees} (usually computation trees of finite \kl{Kripke
structures}). It~is worth noticing that there exist several semantics
for \QCTL in the literature~\cite{Kup95a,Fre06,LM14} and this leads to
important differences in term of complexity or expressiveness. Here we
consider the so-called \intro{tree semantics}: given a \QCTL
formula~$\phi$, a~\SDtree{2^{\AtP}}{\Dir}~$\Tree=\tuple{\tree,\lab}$,
and a node~$n$, we~write~${\Tree, n \intro*\models \phi}$ to denote
that~$\phi$~holds at node~$n$ in~$\Tree$, which~is defined inductively
as follows:
\[
\begin{array}{rcl}
\Tree, n \models p & \text{ iff } & p\in l(n) \\
\Tree,  n \models \neg\phi & \text{ iff } & \Tree,n \not\models \phi \\
\Tree, n \models \phi \ou \psi & \text{ iff } & \Tree, n
\models \phi \text{ or }  \Tree,  n \models \psi \\
\noalign{\pagebreak[2]}
\Tree, n \models \EX \phi & \text{ iff }& \exists d\in \Dir \text{ s.t. } 
   n\cdot d \in \tree  \text{ and }
   \Tree, n\cdot d \models \phi \\
\Tree, n \models \AX \phi & \text{ iff }& \forall d\in \Dir \text{ s.t. } 
   n\cdot d \in \tree\text{, we have } 
   \Tree, n\cdot d \models \phi \\
\noalign{\pagebreak[2]}
\Tree, n \models \Ex \phi \Until \psi & \text{ iff } &\exists w \in \Dir*
\text{ s.t. }
n\cdot w \in \tree \text{ and }
 \Tree, n\cdot w \models \psi 
\text{ and } \\
&&
\forall 0\leq j < \size w.\  %
  \Tree, n\cdot \prefix wj \models \phi\\
\Tree, n \models \All \phi \Until \psi & \text{ iff } &
    \forall w \in \Dir*\cup\Dir~.\ \text{if }
 n\cdot w \text{ is a \kl{branch} in } \tree \text{, then } \exists i \geq 0 . \\
 & &\Tree, n\cdot \prefix wi \models \psi 
\text{ and } \forall 0 \leq j < i\text{, we have }  \Tree, n\cdot \prefix wj \models \phi\\
\Tree, n \models \exists p.\ \phi & \text{ iff }& \exists \Tree'
\equiv_{2^{\AtP\backslash\{p\}}} \Tree \text{ s.t. } \Tree', n \models \phi,
\end{array}
\]
where, following the definition given at
Section~\ref{ssec-proj}, $\Tree' \equiv_{2^{\AtP\backslash\{p\}}} \Tree$
means that $\Tree$ and~$\Tree'$ are identical except for the labelling
with atomic proposition~$p$: formula~$\exists p.\ \phi$ intuitively
means that it is possible to modify the labelling of~$\Tree$ for
proposition~$p$ in such a way that $\phi$~holds.

\AP
Finally, for a \kl{Kripke
structure}~$\calK$ and one of its states~$v$, we~write $\calK,v\models\phi_s$
whenever $\Tree_{\calK,v},\emptyw_{\Tree_{\calK,v}}\models \phi_s$.
We~say that two formulas $\phi_1$ and~$\phi_2$ are \intro(F){equivalent} (denoted $\phi_1 \reintro*\equivF \phi_2$)
when their truth value are equal for every \Stree{2^{\AtP}}.

In~the sequel, we use standard abbreviations such as $\top \eqdef p\ou\neg p$,
$\bot \eqdef \non\top$, $\EF\phi \eqdef \Ex \top \Until \phi$, $\AF\phi \eqdef \All \top \Until \phi$, 
$\AG\phi \eqdef \neg \EF\neg \phi$, 
$\EG\phi \eqdef \neg \AF\neg \phi$ and $\forall p.\ \phi_s \eqdef \neg\exists p.\ \neg\phi_s$.

Quantification over atomic propositions increases the expressiveness
of~\CTL. For example, it~allows us to
count the number of \kl{successors}, as illustrated by  formula 
\[
\Ex_1\Next\phi = \EX\phi \et \non\exists p.\ \big(\EX(p\et \phi) \et
\EX(\non p \et\phi)\big)
\]
where we assume that $p$ does not appear in~$\phi$.  This formula
states that there is exactly one \kl{successor} satisfying~$\phi$: the
first part of the formula enforces the presence of at least one
\kl{successor} satisfying~$\phi$, and if there were two of them, then
labelling only one of them with~$p$ would falsify the second part of
the formula.  It~is well-known~\cite{HM85} that \CTL cannot express
such properties.

Generalising the idea above, \QCTL can also (succinctly) express that a
node has at most $2^k$ successors: this~is achieved by requiring the
existence of a labelling with $k$ atomic propositions~$(p_i)_{1\leq
  i\leq k}$ in such a way that no two successors have the same
labelling:
\begin{multline*}
  \chi_{k} = \exists (p_i)_{1\leq i\leq k}.\ \non\Bigl(\exists q,q'.\ \Ex_1\Next q
  \et \Ex_1\Next q' \et \non\EX(q\et q')\et{} \\
  \ET_{1\leq i\leq k}\!\!\EX(q\et p_i) \Leftrightarrow \EX(q'\et p_i)\Bigr).
\end{multline*}
The negation of this formula expresses the existence of \emph{at
least} $2^k+1$ successors. Using an extra atomic proposition for
isolating a single node, we~can get a formula of size linear in~$k$
expressing the existence of exactly $2^k$ successors. Similar ideas
can be used to succinctly express the existence of $2^k$ successors
satisfying some formula~$\phi$.

In order to give more precise results about \QCTL, 
we~introduce restricted fragments, depending on the nesting of
quantifiers.
Given two \QCTL  formulas~$\phi$ and~$(\psi_i)_i$, and atomic
propositions~$(p_i)_i$ that appear free in~$\phi$ (\ie, not as
quantified propositions), we~write $\phi[(p_i\to \psi_i)_i]$
(or~$\phi[(\psi_i)_i]$ when $(p_i)_i$ are understood from the context) for the
formula obtained from~$\phi$ by replacing each occurrence of~$p_i$
with~$\psi_i$. Given two sublogics~$L_1$ and~$L_2$ of~\QCTL, we~write
$L_1[L_2]=\{\phi[(\psi_i)_i] \mid \phi\in L_1,\ (\psi_i)_i\in L_2\}$.
We~then inductively define the following fragments:
\begin{itemize}
\item \AP $\reintro*{\EQCTL[0]}$ and $\reintro*{\AQCTL[0]}$ correspond to \CTL, and for~$k>0$,
$\intro*{\EQCTL[k]}$~is the set of formulas of the form
$\exists p_1.\ \exists p_2\ldots \exists p_n.\ \phi$
for~$\phi\in\AQCTL[k-1]$, and
$\intro*{\AQCTL[k]}$ is the set of formulas of the form
$\forall p_1.\ \forall p_2\ldots \forall p_n.\ \phi$
for~$\phi\in\EQCTL[k-1]$,
\item\AP $\reintro*{\QCTL[0]}$ is~\CTL,
  $\reintro*{\QCTL[1]}=\CTL{}[\EQCTL[1]]$, and for~$k>1$,
  $\intro*{\QCTL[k]}$ is the logic $\QCTL[1][\QCTL[k-1]]$.
\end{itemize}
Hence formulas in \EQCTL[k] and \AQCTL[k] are in prenex form, and
involve $k-1$ quantifier alternations (respectively starting with
existential and universal quantifiers); on the other hand, \QCTL[k]
counts the maximal number of nested blocks of quantifiers, allowing
boolean and \CTL operators between blocks.
An~easy induction shows that \EQCTL[k] and~\AQCTL[k] are syntactic fragments of~\QCTL[k].
As examples, it can be seen that formula $\Ex_1\Next\phi$ (for
$\phi\in\CTL$) is in \AQCTL[1], and that formula~$\chi_k$ is
in \QCTL[3],  but can easily be rewritten as a formula in~\EQCTL[3].

\subsection{From \QCTL to \AATA}

\subsubsection{From \CTL to tree automata}

  Any \CTL formula~$\phi$ can be turned into an \AAPTA~$\calA_\phi$
  accepting exactly the \kl{trees} where~$\phi$~holds.  One of the first
  such constructions is given in~\cite{BVW94,jacm47(2)-KVW};
  it~is based on \faTAutomata, but the
  construction has then been extended to arbitrary-arity trees
  using \BDAutomata~\cite{Wil99a}
  (see~Section~\ref{sec-relW}).

  Here we adapt the construction of~\cite{Wil99a} to our \EUconstrs
  in the transitions of the automaton.
  We~assume~w.l.o.g.~that negations in~$\phi$ may only appear at the
  level of atomic propositions; transforming a formula in such a negation-normal
  form may at~most double the size of the formula. The~automaton~$\calA_\phi =
  \tuple{\State_\phi,\initstate,\Trans_\phi,\prio_\phi}$ can then be
  defined as follows:
  \begin{itemize}
  \item $\State_\phi$ is the set of state subformulas
    (including~$\top$) of~$\phi$. In~order to~avoid confusion, for~each
    subformula~$\psi$, we~write~$\st\psi$ for the associated state
    in~$\State_\phi$. %
  \item the initial state~$\initstate$ is~$\st\phi$,
  \item given  $\st\psi \in \State_\phi$ and
    $\alp \in 2^{\AtP}$, we define
    $\Trans_\phi(\st\psi,\alp)$ inductively as follows:

    \begin{xalignat*}2
      \delta_\phi(\st\top,\sigma) &= \top &
      \delta_\phi(\st\bot,\sigma) &= \bot \\
      \delta_\phi(\st P,\sigma) &= \begin{cases} \top & \mbox{if } P\in \sigma \\ 
        \bot & \mbox{otherwise} \end{cases} &
      \quad\delta_\phi(\st{\non P},\sigma) &= \begin{cases}\bot&\mbox{if }P\not\in \sigma \\ 
        \top & \mbox{otherwise} \end{cases} \\\noalign{\allowbreak}
  \delta_\phi(\st{\psi_1\et\psi_2},\sigma) &\multicolumn{3}{l}{${}= \delta_\phi(\st{\psi_1},\sigma) \et \delta_\phi(\st{\psi_2},\sigma)$} 
\\[1.2ex]\noalign{\allowbreak}
\delta_\phi(\st{\psi_1\ou\psi_2},\sigma) &\multicolumn{3}{l}{${}= \delta_\phi(\st{\psi_1},\sigma) \ou \delta_\phi(\st{\psi_2},\sigma)$}
      \\[1.2ex]\noalign{\allowbreak}
      \delta_\phi(\st{\EX \psi},\sigma) &=  \EUpair(\st\psi\mapsto 1;\{\st\top\})  
      \\[1.2ex]\noalign{\allowbreak}
\delta_\phi(\st{\AX \psi},\sigma) &=  \EUpair(\emptyset;\{\st\psi\})
      \\[1.2ex]\noalign{\allowbreak}
      \delta_\phi(\st{\Ex \psi_1 \Until \psi_2},\sigma) &\multicolumn{3}{l}{${}=\delta_\phi(\st{\psi_2},\sigma) \ou
        \Bigl(\delta_\phi(\st{\psi_1},\sigma) \et  \EUpair(\st{\Ex\psi_1\Until\psi_2}\mapsto 1;\{\st{\top}\})  \Bigr)$} \\[1.2ex]\noalign{\allowbreak}
      \delta_\phi(\st{\Ex \psi_1 \WUntil \psi_2},\sigma)&\multicolumn{3}{l}{${}=\delta_\phi(\st{\psi_2},\sigma) \ou
        \Bigl(\delta_\phi(\st{\psi_1},\sigma) \et \EUpair(\st{\Ex\psi_1\WUntil\psi_2}\mapsto 1;\{\st\top\})  \Bigr)$}\\[1.2ex]\noalign{\allowbreak}
      \delta_\phi(\st{\All \psi_1 \Until \psi_2},\sigma)&\multicolumn{3}{l}{${}=\delta_\phi(\st{\psi_2},\sigma) \ou
        \Bigl(\delta_\phi(\st{\psi_1},\sigma) \et  \EUpair(\emptyset;\{\st{\All\psi_1\Until\psi_2}\}) \Bigr)$}\\[1.2ex]
      \delta_\phi(\st{\All \psi_1 \WUntil \psi_2},\sigma) &\multicolumn{3}{l}{${}=\delta_\phi(\st{\psi_2},\sigma) \ou
        \Bigl(\delta_\phi(\st{\psi_1},\sigma) \et \EUpair(\emptyset;\{\st{\All\psi_1\WUntil\psi_2}\}) \Bigr)$}
    \end{xalignat*}

  \item the acceptance condition is a \kl{parity condition} defined through
    the following \kl{priority} function:
    \begin{xalignat*}1
      \prio_\phi(\st{\Ex \psi_1 \Until \psi_2}) &= \prio_\phi(\st{\All  \psi_1
      \Until \psi_2}) = 1 \\
      \prio_\phi(\st{\Ex \psi_1 \WUntil \psi_2}) &= \prio_\phi(\st{\All  \psi_1
      \WUntil \psi_2}) = 0
    \end{xalignat*}
    The \kl{priority} for all other states~$\st{\psi}$ is irrelevant as
    they can only appear finitely many times along a~\kl{branch}.

  \end{itemize}

 Then we have the following theorem, proved by Wilke for his
 construction with \BDAutomata:
\begin{theorem}[\cite{Wil99a}]
\label{thm-ctl}
  For any \CTL formula~$\phi$, there exists an
    \AATA~$\calA_\phi$
  that \kl{accepts} exactly the \kl{trees} in which~$\phi$~holds.
  The~\kl(A){size} of this automaton is bounded by 
  $(O(\size\phi),O(\size\phi),\allowbreak 1, 1, 2)$.
\end{theorem}

\begin{remark}
Consider a \CTL formula $\phi$ and a \kl{Kripke structure}~$\calK$.
Using~automaton~$\calA_\phi$ (and~the~fact that it contains only
\EUconstrs with $\max(\sizeE\Trans,\allowbreak\sizeU\Trans)\leq 1$),
we~can build a \kl{parity game}~$\GAK$ whose size is in
${O(\size\calK\cdot\size \phi)}$ and that can be solved
(using~Prop.~\ref{prop-compl-pg}) in time $O\bigl((\size K\cdot\size
\phi)^2\bigr)$.  This~is not the optimal complexity for \CTL
model-checking; however, this~can be improved by using the fact that
the automaton~$\Aut_\vfi$ is \emph{weak}~\cite{MSS86,BVW94,VW08}.
This provides us with a \emph{weak} game~$\GAK$, which allows us to
get an algorithm running in $O(\size{\calK}\cdot\size{\phi})$, thereby
recovering the classical complexity for \CTL model checking.
\end{remark}

\begin{remark}
In Theorem~\ref{thm-ctl}, we can observe that the boolean size of the
transition function of $\Aut_\phi$ is linear in $\size\phi$. This is a
direct consequence of its definition.  While it~is constant for the
transition formulas from states of the
form~$\st{\Ex \psi_1 \Until \psi_2}$, it~might be linear 
e.g. for subformulas of the form $\psi = \psi_1 \et \ldots \psi_p$: the
transition function $\delta_\phi([\psi],\sigma)$ is then defined as
the conjunction of every $\delta_\phi([\psi_i],\sigma)$,
hence it has linear size.
\end{remark}

\subsubsection{A tree-automata construction for \QCTL formulas}

\label{sec-qctl2aut}

Combining the construction for \CTL and the operations over \AAPTA
allows us to extend the automata construction to \QCTL formulas.
The crucial point is the handling of quantifications.

Consider a \QCTL formula $\Phi$ where
the negations can only be followed by atomic propositions or $\exists
p.\psi$ subformulas, and a subformula~$\phi$ of~$\Phi$ of the form $\exists
p.\psi$, assuming that we have built an \AAPTA $\calA_\psi$ for~$\psi$.
If $\calA_\psi$ is \kl{non-alternating}, we~can use the \kl{projection}
operation on~$\calA_\psi$ (see~Section~\ref{ssec-proj})
and get an~\AAPTA for~$\phi$. Otherwise,
we~have to first turn~$\calA_\psi$ into a \kl{non-alternating}
automaton (with an exponential blow-up in the size of the automaton)
before using \kl{projection}.
Note that we can handle in one
step a block of existential quantifiers of the form $\exists p_1\cdots
\exists p_n.\psi$ (hence with a single exponential blow-up).
On~top~of~this, there may be negations in front of existential quantifiers,
which may require complementing the automaton.

Therefore the size of the resulting automaton will drastically depend
on the number of such nested blocks of existential quantifiers in the
\QCTL formula. Thus complexity results are stated for formulas
in~\QCTL[k], \EQCTL[k] and~\AQCTL[k].

\AP
In the following, we write $\intro*\EXP{k}{n}$ to denote the family of sets
of functions of one variable~$n$ defined inductively as follows:
$\reintro*\EXP{0}{n}$ is the set of functions bounded by a polynomial in~$n$,
and $\reintro*\EXP{(k+1)}{n}$ contains all functions~$f$ such that $f \in
O(2^{g})$ with $g \in \EXP{k}{n}$.

We can now formally state the result as follows:

\begin{theorem}\label{thm-aut-qnctl}
  Given a \QCTL[k] formula~$\phi$ over~$\AtP$ with $k>0$, we~can
  construct a \AAPTA $\Aut_\phi$ over $2^\AtP$ accepting exactly the
  trees satisfying~$\phi$. The~automaton~$\Aut_\phi$ has size
  $\tuple{\EXP{k}{\size\phi},\EXP{k}{\size\phi},\EXP{(k-1)}{\size\phi},
    1, \EXP{(k-1)}{\size\phi}}$.  If~additionally $\phi$ is
  in~\EQCTL[k], then $\Aut_\phi$ is non-alternating.
\end{theorem}

\begin{proof}
We proceed by induction over $k$. We~prove 
the result for formulas in~\QCTL[k], showing along the way the property for formulas in~\EQCTL[k].
\begin{itemize}
\item if $\phi \in \QCTL[1]$, then $\phi$ is of the form
  $\Phi[(\psi_i)_{1\leq i\leq m}]$ where $\Phi$ is a \CTL formula and
  $(\psi_i)_{1\leq i\leq m}$ are \EQCTL[1] formulas. We handle each
  $\psi_i$ separately. Assume that $\psi_i = \exists
  p^i_1\ldots\exists p^i_{l_i}.\ \psi'_i$ with ${\psi'_i \in \CTL}$.
  From~Theorem~\ref{thm-ctl}, one~can build an \AAPTA $\Aut_{\psi'_i}$
  recognising the \kl{trees} satisfying~$\psi'_i$; moreover,
  $\smsize{\Aut_{\psi'_i}}$ is bounded by
  $\tuple{O(\size{\psi'_i}),O(\size{\psi'_i}),1,1,2}$.
  We~then apply the constructions
  of Section~\ref{sec-ops}, and the results summarised in
  Table~\ref{tab-ops}:
  \begin{itemize}
  \item we~can remove alternation and get
    an equivalent \nAAPTA~$\NAut_{\psi'_i}$ whose size is bounded by
    $\tuple{2^{O(\size{\psi'_i}^2\cdot \log(\size{\psi'_i}))},
    2^{O(\size{\psi'_i}^2\cdot \log(\size{\psi'_i}))}, O(\size{\psi'_i}^2) , 1, O(\size{\psi'_i})}$;
            
  \item applying projection (to~$\calN_{\psi'_i}$ and for
    atomic propositions~$p^i_1$ to~$p^i_{l_i}$, \ie\ with $\size{\Alp'}= 2^{l_i} \leq 2^{\size{\psi_i}}$), we~get an~\nAAPTA
    $\calB_{\psi_i}$
    recognising the models of~$\psi_i$, whose size
    is bounded by:
$\tuple{2^{O(\size{\psi_i}^2\cdot \log(\size{\psi_i}))},2^{O(\size{\psi_i}^2\cdot \log(\size{\psi_i}))}, O(\size{\psi_i}^2) , 1, O(\size{\psi_i})}$.
     
    This shows that for a formula~$\psi_i$ in~\EQCTL[1], we~can build
    a \kl{non-alternating} \nAAPTA of size
    $\tuple{\EXP{1}{\size\phi},\EXP{1}{\size\phi},\EXP{0}{\size\phi},
      1, \EXP{0}{\size\phi}}$.
    
  \item in case formula~$\psi_i$ is in the (direct) scope of a
    negation operator, we~complement~$\BAut_{\psi_i}$ into an
    \AAPTA~$\CAut_{\overline\psi_i}$ whose size is bounded by:
    \[\tuple{\size\Alp\cdot 2^{O(\size{\psi_i}^2\cdot \log(\size{\psi_i}))},
    2^{O(\size{\psi_i}^2\cdot \log(\size{\psi_i}))}, O(\size{\psi_i}^2) , 1,
      O(\size{\psi_i})}\]
      where $\size\Alp$ can be bounded by $2^{\size\phi}$;
  \item we build the final
    \AAPTA~$\Aut_\phi=\tuple{Q,\initstate,\Trans,\prio}$ using the
    construction of Theorem~\ref{thm-ctl}, in which we add the
    following rule to deal with subformulas~$\psi_i$ in~\EQCTL[1] (or~their
    negations):
    \begin{xalignat*}2
    \Trans(\psi_i,\sigma) & =   \Trans_{\psi_i}(\initstate^{\psi_i},\sigma)   &
    \Trans(\non\psi_i,\sigma) & =   \Trans_{\overline\psi_i}(\initstate^{\overline\psi_i},\sigma),
    \end{xalignat*}
    where $\initstate^{\psi_i}$ and $\initstate^{\overline\psi_i}$ are
    the initial states of~$\BAut_{\psi_i}$ and~$\CAut_{\psi_i}$,
    respectively.
    Our~automaton then has size at most
    $\tuple{\EXP{1}{\size\phi},\EXP{1}{\size\phi},\EXP{0}{\size\phi},
      1, \EXP{0}{\size\phi}}$.
which proves the base case of our result.
   \end{itemize}

\item if $\phi \in \QCTL[k]$ with $k>1$, we show that for every
  $\phi$-subformula $\psi \in \QCTL[k']$ with $1 \leq k' \leq k$,
  the~\AAPTA~$\Aut_\psi$ has size
  $\tuple{\EXP{k'}{\size\phi},\EXP{k'}{\size\phi},\EXP{(k'-1)}{\size\phi},
    1, \EXP{(k'-1)}{\size\phi}}$. Note that the size of the automaton
  depends on $\size\phi$ (and not on $\size\psi$) because the
  complement operation provides an automaton whose size depends
  on~$\size\Alp$, which can only be bounded by $2^{\size\phi}$.  We
  prove this result by induction over~$k'$. The~base case is similar
  to the previous case with $k=1$.

Now consider $k'>1$. Assume  $\psi \in \QCTL[{k'}]$~is of the
  form $\Psi[(\psi_i)_{1\leq i\leq m}]$, where $\Psi$~is a~\CTL
  formula and each $\psi_i$ is of the form $\exists p^i_1\ldots \exists
  p^i_{l_i}.\ \psi'_i$ with $\psi'_i\in \QCTL[{k'-1}]$.

  From the~induction hypothesis, we~can build \AAPTAs
  $\calA_{\psi'_i}=\tuple{\State_i,{\initstate}_i,\Trans_i,\prio_i}$ recognising the
  trees satisfying $\psi'_i$ for all~$i$, and whose size~$\tuple{s_i,b_i,e_i,1,p_i}$
  is bounded by
  $\tuple{{\EXP{(k'-1)}{\size\phi}},{\EXP{(k'-1)}{\size\phi}},
    {\EXP{(k'-2)}{\size\phi}},1,{\EXP{(k'-2)}{\size\phi}}}$.
  
    Applying the simulation theorem provides us with
    \nAAPTAs~$\NAut_{\psi'_i}$, each accepting the same language
    as~$\Aut_{\psi'_i}$, whose sizes are at most:
    \[\tuple{2^{O(s_i^2\cdot\log(s_i))}, 
    2^{s_i\cdot b_i}\cdot[s_i\cdot b_i \cdot(e_i+1)^{s_i \cdot b_i}]^{e_i\cdot s_i \cdot b_i},
      s_i\cdot b_i\cdot e_i, 1, 2\cdot(
      s_i\cdot p_i+1)}\]
      After projection, we~get
    \nAAPTAs~$\BAut_{\psi_i}$ for formulas~$\psi_i$ whose sizes are
    bounded by:
    \[ \tuple{2^{O(s_i^2\cdot\log(s_i))}, 2^{l_i}\cdot 2^{s_i\cdot b_i} \cdot [s_i\cdot b_i \cdot (e_i+1)^{s_i \cdot b_i}]^{s_i \cdot b_i \cdot e_i},
      s_i\cdot  b_i\cdot e_i, 1, 2\cdot (s_i\cdot p_i+1)}\]
    In~case~$\psi_i$ is
    used negatively in~$\phi$,
    we~compute the
    complement~$\CAut_{\overline\psi_i}$ of~$\BAut_{\psi_i}$, which is an \AAPTA
    of~size at~most:
 \[
\bigl(
  O(\size\Alp\cdot 2^{l_i} \cdot 2^{O(s_i^2\cdot b_i^2 \cdot e_i \cdot\log(s_i\cdot b_i (e_i+1)))}),
  2^{l_i} \cdot  2^{O(s_i^2\cdot b_i^2 \cdot e_i \cdot\log(s_i\cdot b_i (e_i+1)))}, 
  s_i\cdot b_i\cdot e_i,
  1,
  2\cdot s_i\cdot p_i+3
\bigr)
\]  
      and
      again $\size\Alp$ can be bounded by $2^{\size\phi}$.
      
    Finally, as~for the base case, we~apply the
    construction of Theorem~\ref{thm-ctl}, and  end up with an \AAPTA for formula~$\psi$,
    whose size is bounded by 
    $\tuple{{\EXP{k'}{\size\phi}},{\EXP{k'}{\size\phi}},
    {\EXP{(k'-1)}{\size\phi}},1,{\EXP{(k'-1)}{\size\phi}}}$.

    \smallskip

    This result applies to~$\phi$ itself: we~get that the size
    of~$\Aut_\phi$ is bounded by
    $\tuple{{\EXP{k}{\size\phi}}, \allowbreak {\EXP{k}{\size\phi}},
      {\EXP{(k-1)}{\size\phi}},1,{\EXP{(k-1)}{\size\phi}}}$, and if
    $\phi$ belongs to \EQCTL[k], we~get a non-alternating automaton as
    the last step is the projection operation (on a non-alternating
    automaton).  This concludes our proof.  \qed
    \end{itemize}
\def\qed{}
\end{proof}

Model-checking and satisfiability for \QCTL can be solved using
\EU-automata.  Using Corollary~\ref{coro-algo-aeupta}, Theorem~\ref{theo-mc-aapta}
and the hardness results of~\cite{LM14}, we~get:
\begin{theorem}
The satisfiability problem for \QCTL[k], \AQCTL[k] and \EQCTL[k+1] is
\EXPTIME[(k+1)]-complete. The~model-checking problem for \QCTL[k],
\AQCTL[k] and \EQCTL[k] is \EXPTIME[k]-complete.
\end{theorem}

\subsubsection{Extension to \QCTLs}

The automata construction for~\QCTL can be extended to \intro*\QCTLs,
the~extension of \CTLs with quantifications over atomic propositions. 
In~\intro*\CTLs, we~distinguish between state formulas~$\phi_s$, interpreted
over states, and path formulas~$\phi_p$, interpreted over infinite
paths. Informally, a~state formula corresponds to some boolean
combination of atomic propositions and formulas of the form $\Ex
\phi_p$ and $\All \phi_p$ (\ie\ path formulas prefixed by some path
quantifier), and path formulas are defined as~\LTL formulas with state
formulas appearing in place of atomic propositions. The~logic~\QCTLs
extends~\CTLs by allowing formulas of the form $\exists p.\ \phi_s$ as
state formulas\footnote{This precision is important: allowing such
quantifications inside a path formula changes the expressiveness of
the logics~\cite{LM14}.}.

As for \QCTL, we can define several fragments of~\QCTLs: \intro*{\QCTLs[k]}
contains formulas in which the maximum number of nested blocks of
quantifiers is at most~$k$. The~construction of~$\Aut_\phi$ for~${\phi
  \in \QCTL}$ follows the same steps as for \QCTL; the~main difference
is that we have to consider formulas of the form~$\Ex \phi_p$ where
$\phi_p$ is an \LTL formula: in that case, we have to first build a
word automaton to capture~$\phi_p$, and then use
Proposition~\ref{prop-autEx} to derive a tree automaton for $\Ex
\phi_p$. The~complexity is then higher.  We~start with this following
result for~\CTLs:

\begin{proposition}\label{autctls}
  Given a \CTLs formula~$\phi$ over~$\AtP$,
  we~can construct an \AAPTA $\calA_\phi$ over $2^\AtP$  accepting exactly the
  trees satisfying~$\phi$. The automaton
  $\calA_\phi$ has size $\tuple{\EXP{1}{\size\phi},\EXP{1}{\size\phi},2,1,3}$.
\end{proposition}
\begin{proof}
The key point is the treatment of formulas of the form $\Ex \phi_p$
with $\phi_p \in \LTL$. In that case, we build a \nAPWA
$\calA_{\phi_p}$ in a standard way, whose size is exponential in
$\size\phi_p$. Then by applying~Proposition~\ref{prop-autEx}, we get
an \nAAPTA corresponding to the formula $\Ex \phi_p$ whose size is
$\tuple{\EXP{1}{\size\phi},\EXP{1}{\size\phi},1,1,2}$. In case of
$\All \phi_p$ formula, we just have to add a complementation step
(Theorem~\ref{thm-compl}) and we get an \AAPTA whose size is
$\tuple{\EXP{1}{\size\phi},\EXP{1}{\size\phi},2,1,3}$.

Now consider a \CTLs formula: we apply the previous construction to
every subformula $\Ex \phi_p$ starting with the innermost
subformulas. Finally we get an \AAPTA that \emph{combines} the
different automata and it provides an \AAPTA whose size is
$\tuple{\EXP{1}{\size\phi},\EXP{1}{\size\phi},2,1,3}$.
\end{proof}

We can now state the construction  for \QCTLs:

\begin{theorem}\label{thm-aut-qnctls}
  Given a \QCTLs[k] formula~$\phi$ over~$\AtP$ with $k>0$, we~can
  construct an \AAPTA $\calA_\phi$ over $2^\AtP$ accepting exactly the
  trees satisfying~$\phi$. The automaton $\calA_\phi$ has size
  $\tuple{\EXP{(k+1)}{\size\phi},\EXP{(k+1)}{\size\phi},\EXP{k}{\size\phi},1,\EXP{k}{\size\phi}}$.
  \end{theorem}

\begin{proof}
\begin{itemize}
\item Consider $\phi \in$ \QCTLs[1]. Assume that negations occur only
  before existential quantifications over
  propositions. Thus~$\phi$~can be seen as a formula 
  $\Phi[(\psi_i)_{1\leq i\leq m}]$ where $\Phi$ is a \CTLs formula and
  $(\psi_i)_{1\leq i\leq m}$ are \EQCTLs[1] formulas, with 
  $\psi_i = \exists
  p^i_1\ldots\exists p^i_{l_i}.\ \psi'_i$ with ${\psi'_i \in \CTLs}$.
 
  We first build an \AAPTA $\calA_{\psi'_i}$  as explained
  in Prop.~\ref{autctls}. Then we can transform each of these
  automata into a \kl{non-alternating} automaton~$\calN_{\psi'_i}$ whose size
  is in
  $\tuple{\EXP{2}{\size\psi_i},\EXP{2}{\size\psi_i},\EXP{1}{\size\psi_i},1,\EXP{1}{\size\psi_i}}$
  (see~Theorem~\ref{thm-simu}).  We~can then apply the \kl{projection}
  operation over these automata to deal with the
  quantification~$\exists p^i_1\ldots\exists p^i_{l_i}$. Then we get the (non-alternating) automata $\calB_{\psi_i}$. 
  
  A~\kl{complementation} procedure is possibly applied (when a
  negation precedes the corresponding existential quantification in~$\Phi$).
  In~that case, we obtain an alternating automaton~$\CAut_{\overline\psi_i}$
  whose size admits the same bounds as above.
  Finally it remains
  to consider the \CTLs context~$\Phi$, which corresponds to an \AAPTA of
  size at most $\tuple{\EXP{1}{\size\Phi},\EXP{1}{\size\Phi},2,1,3}$; 
  combined with automata~$\CAut_{\overline\psi_i}$s and $\calB_{\psi_i}$s, we~get an \AAPTA whose
  size is in
  $\tuple{\EXP{2}{\size\phi},\EXP{2}{\size\phi},\EXP{1}{\size\phi},1,
    \EXP{1}{\size\phi}}$.

  \item Consider $\phi \in \QCTLs[k+1]$. Here again we assume that
    negations occur only before existential quantifications over
    propositions. As in the construction for \QCTL formula, we show
    that for every $\phi$-subformula $\psi \in \QCTLs[k']$ with $1
    \leq k' \leq k$, the automaton \AAPTA $\Aut_\psi$ has size
    $\tuple{\EXP{(k'+1)}{\size\phi},\allowbreak\EXP{(k'+1)}{\size\phi},\EXP{k'}{\size\phi},
      1, \EXP{k'}{\size\phi}}$. We prove it by induction over
    $k'$. The result holds for $k'=1$ (similar to the previous
    case). Now assume $1 < k' \leq k$. Consider a $\phi$-subformula
    $\psi \in \QCTLs[k']$. Then $\psi$ is of the form
    $\Psi[(\psi_i)_{1\leq i\leq m}]$ where $\Psi$ is a \CTLs formula
    and every $(\psi_i)$ is of the form $\exists p^i_1\ldots\exists
    p^i_{l_i}.\ \psi'_i$ with ${\psi'_i \in \QCTLs[k'-1]}$.
  
     By~induction hypothesis, we~can build an \AAPTA~$\calA_{\psi'_i}$
     for each formula~$\psi_i$, whose size is bounded by
     $\tuple{\EXP{k'}{\size{\phi}},\allowbreak
       \EXP{k'}{\size{\phi}},\allowbreak
       \EXP{(k'-1)}{\size{\phi}},1,\EXP{(k'-1)}{\size{\phi}}}$.
     Applying the \kl{simulation} theorem (Theorem~\ref{thm-simu}),
     we~get an \nAAPTA whose size is at most
     $\tuple{\EXP{(k'+1)}{\size{\phi}},\allowbreak
       \EXP{(k'+1)}{\size{\phi}},\allowbreak
       \EXP{k'}{\size{\phi}},\allowbreak
       1,\allowbreak\EXP{k'}{\size{\phi}}}$ We can then apply the
     \kl{projection} operation, possibly followed by a complementation
     operation which provides an automaton whose size is still bounded
     by $\tuple{\EXP{(k'+1)}{\size{\phi}},\allowbreak
       \EXP{(k'+1)}{\size{\phi}},\allowbreak
       \EXP{k'}{\size{\phi}},1,\EXP{k'}{\size{\phi}}}$.  Finally,
     it~remains to incorporate the \CTLs context~$\Psi$, which we
     perform as in the base case; we~finally get an \AAPTA whose size
     is in
     $\tuple{\EXP{(k'+1)}{\size\phi},\EXP{(k'+1)}{\size\phi},\EXP{k'}{\size\phi},1,\EXP{k'}{\size\phi}}$.
     This~concludes
     the proof of the inductive step of the intermediary result. And
     we can deduce that the automaton for $\phi$ is therefore in
     $\tuple{\EXP{(k+1)}{\size\phi},\allowbreak
       \EXP{(k+1)}{\size\phi},\EXP{k}{\size\phi},1,\EXP{k}{\size\phi}}$
     which concludes the proof.  \qed
\end{itemize}
Note also that this construction provides a  non-alternating automaton if $\phi$ belongs  to \EQCTLs[k]. 
\let\qed\relax
\end{proof}

As a direct consequence, we get decision procedures for Model-checking
and satisfiability for \QCTLs based on \EU-automata; again, lower
complexity bounds are obtained from~\cite{LM14}:
\begin{theorem}
  The satisfiability problem for \QCTLs[k], \AQCTLs[k], and \EQCTLs[k+1] is \EXPTIME[(k+2)]-complete.
  
  The~model-checking problem for \QCTLs[k], \AQCTLs[k], and \EQCTLs[k] is \EXPTIME[(k+1)]-complete.
\end{theorem}

\subsection{From \AAPTA to \QCTL}
\label{sec-aut2qctl}

In~this section,
we~use \kl{\EU tree automata} to derive expressiveness results:
we~turn an \nAAPTA  $\Aut =
(\State,q_0,\Trans,\prio)$ over~$2^{\AtP}$
into an \kl(F){equivalent} (over all \Strees{2^{\AtP}})
\QCTL formula~$\Phi_\calA$. 
Remember that for~\nAAPTA,
the~transition function~$\delta(q,\sigma)$ is a disjunction of~\EUprs. 

\begin{theorem}
  For any~$\Aut$ be an~\nAAPTA over~$2^{\AtP}$, we~can build an
  \EQCTL[2] formula~$\Phi_\calA$ such that, for any
  $2^{\AtP}$-labelled tree~$\calT$, it~holds $\calT \in \calL(\calA)$
  if, and only~if, $\calT,\epsilon \models \Phi_\calA$.  The~size of
  the formula $\Phi_\calA$ is in~ $O(\size\State \cdot \priomax +
  \size\State\cdot 2^{\size\AtP}\cdot \size\AtP \cdot
                   {\sizeB\Trans}\cdot\sizeE\Trans\cdot(\sizeE\Trans+\sizeU\Trans))$
\end{theorem}

\begin{proof}
Let  $\State=\{q_i \mid 0 \leq i \leq n=\size{\State}-1\}$.
In~$\Phi_\calA$, we~use the set of fresh quantified atomic
propositions $\{b, q_0,\ldots,q_n,\penalty5000\relax
p_1,\ldots,p_{\sizeE\Trans },p\}$ in order to express the existence of
an \kl(ET){accepting} \kl{execution tree} of the automaton:
propositions in~$\{q_0,...,q_n\}$ will be used to label each node of
the input \kl{tree} with the name of the state visiting that node in
the \kl{execution tree}, while propositions
in~$\{p_1,...,p_{\sizeE\Trans}\}$ are used to distinguish between the
\kl{successors} involved in the verification of the $E$-part of
\EUconstrs; proposition~$b$ is used for expressing the \kl{acceptance}
condition and proposition~$p$ is used to ensure that no node has more
than one \kl{successor} labelled with the same $p_i$\footnote{For the
sake of clarity, we use two distinct propositions~$b$ and~$p$, but we
could have used the same one.}.
Our formula~$\Phi_{\calA}$ reads as follows:
\[
 \exists q_0\ldots \exists q_n.\
\exists p_1\ldots \exists p_{\sizeE\Trans }.\ \forall b.\ \forall p.\  (\Phi_p \et
\widetilde\Phi_{\calA}).
\]
In~this formula, $\Phi_p$ will be used to state that no nodes are
labelled with several~$p_i$s and no nodes have more than one
\kl{successor} labelled with~$p_i$ (for~any~$i$), while
$\widetilde\Phi_{\calA}$ will enforce that the labelled tree describes
an \kl(ET){accepting} \kl{execution tree} of~$\calA$ on~$\calT$.
Formally, $\Phi_p$~is defined as the following formula (remember that
$p$ is quantified universally):
\[
\Phi_p = \AG   \ET_{1\leq i \leq \sizeE\Trans}\Big[ (p_i \impl  \ET_{j\not=i} \non p_j) \et \Big( \AX (p_i \impl  p) \ou \AX (p_i \impl \non p)\Big) \Big]
\]
Formula~$\widetilde\Phi_{\calA}$ is defined as 
\[
\widetilde\Phi_{\calA} = 
  q_0 \et 
   \ET_{i=0}^{n} \AG \Bigl[q_i \Rightarrow \Big(    \neg\lambda_{Q\setminus\{q_i\}} \et
    \OU_{P \subseteq \AtP} \bigl(  
      \Gamma_{P}  
     \et \Psi_{\delta(q,P)}
     \bigr) \Bigr) \Bigr]  
  \et
  \Psi_{\prio},
\]
where, for any~set~$S$,
formula~$\lambda_S$~is the propositional formula $\OU_{q\in S} q$, and
for any $P\subseteq \AtP$, formula~$\Gamma_{P}$ is the propositional formula
$\ET_{p\in P} p \et \ET_{p'\in \AtP\setminus P} \non p'$
(note that the size of $\Gamma_{P}$ is in $O(\size\AtP)$).
We~now define formula~$\Psi_{\delta(q,P)}$, which encodes the
satisfaction of the transition function, and formula~$\Psi_{\prio}$,
which states that any infinite \kl{branch} (of the~\kl{tree})
labelled with~$b$ satisfies the \kl{parity condition}.

For the former, we write (remember that $\Aut$ is non-alternating):
  \[
  \Psi_{\delta(q,P)} =   \begin{cases} \top & \text{if}\; \delta(q,P)=\top \\
\bot & \text{if}\; \delta(q,P)=\bot \\
 {\displaystyle  \OU_{\EUpair(E;U) \in \delta(q,P)}  \!\!\!\!\!\! \Psi_{\EUpair(E;U)} } & \text{otherwise},\end{cases}
  \]  
  where $\Psi_{\EUpair(E;U)}$ encodes the constraint~$\EUpair(E;U)$
  of~$\calA$: writing~$E$ as the \kl{multiset} $\mset{
    E_1,...,E_{|E|}}$, where each $E_j$ belongs to~$\State$, we~let:
\[
\Psi_{\EUpair(E;U)} =  \ET_{j=1}^{\size{E}} \Bigl[ 
  \Ex \Next ( p_j \et E_j )
  \et  
  \AX \Big(
  (\bigwedge_{j=1}^{\size{E}} \neg p_j) \Rightarrow  \bigvee_{q\in U} q  
  \Big)  \Bigr]
\]
Remember that thanks to formula~$\Phi_p$, we~have ensured that no
nodes can be labelled with several~$p_i$s and no nodes can have
several \kl{successors} labelled with the same~$p_j$; thus
formula~$\Psi_{\EUpair(E;U)}$ ensures that all states in~$E$ label 
distinct successors, and nodes with no~$p_j$ are labelled with some
proposition~$q$ corresponding to a state in~$U$.

\smallskip
Formula $\Psi_{\prio}$ expresses the fact that in any infinite subtree
labelled with~$b$, there exists no infinite branches where the smallest
priority appearing infinitely many times is~odd.
This can be characterised as follows:
\begin{multline*}
\Psi_{\prio} =     \Big( b \et \AG (b \Rightarrow \EX b) \;\et\; \AG ((\non b) \Rightarrow \AX (\non b)) \Big)  \;\impl\;  \\
\non\!\!\OU_{\stackrel{0\leq d \leq \priomax}{\text{\tiny s.t. $d$ odd}}} \Bigl[ \AG \; \AF  \Big( b \Rightarrow (
\alpha_{=d}
) \Big) \et
\EF (\EG (b \et  \alpha_{\geq d})) 
\Bigr]\Big),
\end{multline*}
where $\alpha_{=d}$ is the formula $\bigvee_{\prio(q_i)=d} q_i$
characterising (atomic propositions corresponding~to) states having
\kl{priority}~$d$, and $\alpha_{\geq d}$ is the formula
$\bigvee_{\prio(q_i)\geq d} q_i$, identifying states with
\kl{priorities} greater than (or~equal~to)~$d$.  Note that the formula
to the left of the implication holds true if, and only~if,
proposition~$b$ labels exactly an infinite subtree from the current
node (subformula~$(b \impl \EX b)$ ensures infiniteness, and
subformula~$(\non b \impl \AX \non b)$ ensures that every $b$-node
is reachable from the current node via a $b$-path).
We will show below why $\Psi_{\prio}$ ensures the
satisfaction of the \kl{parity condition}.

The size of $\Psi_{\delta(q,P)}$ is in
$O(\sizeB\Trans\cdot\sizeE\Trans\cdot(\sizeE\Trans+\sizeU\Trans)$. The
size of $\Psi_{{\prio}}$ is in $O(\size\State \cdot\priomax)$.  The
size of $\Phi_p$ is in $O({\sizeE\Trans}^2)$. In~the~end, we~get that
the size of $\Phi_{\calA}$ is in $O(\priomax\cdot\size\State +
\size\State\cdot 2^{\size\AtP}\cdot \size\AtP \cdot
{\sizeB\Trans}\cdot\sizeE\Trans\cdot(\sizeE\Trans+\sizeU\Trans))$.
Finally, it~is easily seen that $\Phi_{\calA}$~belongs to~\EQCTL[2].

We~now prove:
\begin{lemma}
Let $\Aut = (\State,q_0,\Trans,\prio)$ be an \nAAPTA and
$\Phi_\calA$ be the \EQCTL[2]
formula defined above. For~any \Stree{2^{\AtP}}~$\calT$,
it~holds:
$\calT \in \calL(\calA)$ if, and only~if, $\calT,\troot \models
\Phi_\calA$.
\end{lemma}

\begin{proof}
Consider a \Stree{2^{\AtP}}~$\calT=\tuple{\tree,\lab}$ and
assume $\calT \in \calL(\calA)$. As~$\Aut$~is non-alternating, there
exists an \kl(ET){accepting}
\kl{execution tree} $\ExTree=\tuple{\tree,\exlab}$ of~$\Aut$
with the same structure~$t$ as~$\calT$.
Then any node~$n$ of~$\ExTree$ is such that $\exlab(n)=(n,q)$ for some
$q \in \State$.

We aim at showing that $\calT,\troot \models \Phi_\Aut$.  First,
we~can label~$\calT$ with~$q_i$s exactly as it is done in~$\ExTree$
with~$\exlab$.  For~each proposition~$p_j$, we~proceed as follows:
consider a node~$n$ labelled with~$(n,q)$; the~transition
function~$\delta(q,\lab(n))$ applies successfully over the subtree
rooted at~$n$ (as~$\calT \in \calL(\calA)$).
If~$\delta(q,\lab(n))=\top$, then $\Psi_{\delta(q,\lab(n))}$
is trivially satisfied. Otherwise there is some \EUpr
$\EUpair(E;U) \in \Trans(q,\lab(n))$ that is satisfied from~$n$
(since the~automaton is non-alternating) and there exist $\size{E}$
\kl{successors} of~$n$ that satisfy~$E$: writing~$E$ as the multiset
$\mset{E_1,\ldots,E_{\size{E}}}$, we~can associate a fixed \kl{successor} with
every~$E_j$. This provides the labelling for proposition~$p_j$ at
this level (we~know that $\size{E} \leq \sizeE\Trans$).  All~the
\kl{successors} that are not labelled with some~$p_j$ with $1\leq
j \leq \size{E}$, have to be accepted by a state in~$U$.
Note also that $\Phi_p$ is satisfied by $\calT,\varepsilon$.

It remains to verify that $\Psi_{{\prio}}$ is satisfied.  As $\ExTree$
is an \kl(ET){accepting} \kl{execution tree}, every infinite
\kl{branch} satisfies the \kl{parity condition}. Consider an
infinite subtree labelled with~$b$, and assume that there exists some
odd \kl{priority}~$d$ that appears infinitely many times along every
\kl{branch} of the $b$-subtree, and such that
along one of these branches,
eventually all \kl{priorities} are greater~than or equal to~$d$
(that is, $d$~is the least \kl{priority} along that \kl{branch});
this clearly
implies that this \kl{branch} violates the \kl{parity condition}, which
contradicts our initial assumption.

In conclusion, the chosen labelling makes the formula
$\widetilde{\Phi_\Aut}$ hold true.

\medskip
Conversely, assume $\calT,\troot \sat \Phi_\Aut$. Consider a labelling
for propositions~$q_i$s and~$p_j$s such that $\widetilde\Phi_{\calA}
\et \Phi_p$ holds true at the root for any valuation of~$p$ and~$b$.
This labelling associates exactly one state of the~automaton with
every node (thanks~to subformulas~$\lambda_{-}$).  Moreover, for every
node~$x$ labelled with proposition~$q$, at~least one
subformula~$\Psi_{\EUpair(E;U)}$ allowing to satisfy
$\Trans(q,\lab(x))$ is fulfilled (or~$\delta(q,\lab(x))=\top$, and the
result is ensured).  When a formula $\Psi_{\EUpair(E;U)}$ holds true
at a node~$x$, then there exist $\size{E}$ distinct \kl{successors}
of~$x$ that are labelled with the states in~$E$ (they are distinct
thanks to subformula~$\Phi_p$), and any other \kl{successor} is
labelled with some state in~$U$.  Therefore the satisfaction of
transitions of the automaton is locally ensured.

Finally, if~$\Psi_{{\prio}}$ is satisfied at the root of~$\calT$, then
for any infinite \kl{branch}, it is possible to label it with~$b$ and then
the formula on the right-hand side of the implication states that for
any odd \kl{priority}~$d$, either it appears a finite number of times along
the selected \kl{branch}, or it is not the smallest \kl{priority} along the
\kl{branch}: this ensures that the smallest infinitely-repeated \kl{priority} is
even, and then the \kl{branch} satisfies the \kl{parity condition}.  Therefore
all \kl{branches} are accepting, and the \kl{tree}~$\calT$ belongs
to~$\calL(\calA)$.  \qed[2]\let\qed\relax
\end{proof}
\let\qed\relax
\end{proof}

Combined with the \kl{simulation} theorem (Theorem~\ref{thm-simu}),
the~previous result provides the following corollary for alternating
automata:

\begin{corollary}
  For any~$\Aut$ be an~\AAPTA, we~can build an exponential-size  \EQCTL[2]
  formula~$\Phi_\calA$ such that, for any
  \Stree{2^{\AtP}}~$\calT$, it~holds $\calT \in \calL(\calA)$ if, and only~if,
$\calT,\troot \models \Phi_\calA$.
\end{corollary}

This result combined with the automata construction of
Section~\ref{sec-qctl2aut} allows us to prove important properties
about the expressive power of~\QCTL.

\subsection{Results about \QCTL\ expressiveness}

\label{sec-qctl-expr}

A~logic~$\calL$ is said to be \intro{at least as expressive} as a
logic~$\calL'$ over a class~$\calM$ of models, which we denote by
$\calL\mathrel{\reintro*\alae{\calM}} \calL'$ (omitting to mention~$\calM$ if it
is clear from the context), whenever for any
formula~$\phi'\in \calL'$, there is a $\phi\in\calL$ such that $\phi$
and $\phi'$ are equivalent over~$\calM$. Both logics~$\calL$
and~$\calL'$ are \intro{equally expressive},
denoted~$\calL \mathrel{\reintro*\eqex{\calM}} \calL'$, when $\calL\succeq\calL'$ and
$\calL'\succeq \calL$; finally, $\calL$~is \intro{strictly more expressive}
than~$\calL'$, written $\calL\mathrel{\reintro*\smex\calM} \calL'$, if
$\calL\succeq\calL'$ and $\calL'\not\succeq \calL$. 
We use $\calL \caplog \calL'$ to denote the fragment of~$\calL\cup\calL'$
containing formulas for which there are equivalent formulas in both $\calL$ and~$\calL'$. 

Combining the construction of Section~\ref{sec-qctl2aut}, turning a \QCTL
formula into an equivalent \AAPTA, and the construction of the previous section,
turning an \AAPTA into an equivalent \EQCTL[2] formula,
we~get that,
in~terms of expressive power, the~hierarchy \QCTL[k] collapses at level~$2$:
\begin{theorem}
  \QCTLs, \QCTL,  \EQCTL[2] and \AQCTL[2] are \kl{equally expressive}.
  Any~formula in~\QCTL[k] can be translated into an equivalent formula
  in~\EQCTL[2] whose size is bounded by~$\EXP{(k+1)}{\size\phi}$.
\end{theorem}

\begin{proof}
Given a \QCTLs formula~$\Phi$, one can build an \AAPTA $\Aut_\Phi$
which recognises the
\Strees{2^{\AtP_\Phi}}
satisfying~$\Phi$, where $\AtP_\Phi$ denotes the set of atomic propositions occurring in~$\Phi$.
 This~automaton~$\Aut_\Phi$ can
then be transformed into a \kl{non-alternating} \nAAPTA~$\calN_\Phi$, from
which we can build a formula~$\Phi_\calN$ belonging
to~\EQCTL[2]. By~construction, we~have $\Phi \equivF \Phi_\calN$ (over any
\Stree{2^{\AtP}}).
The~same holds for~$\non\Phi$, and the negation of the resulting
\EQCTL[2] formula belongs to \AQCTL[2] and is \kl(F){equivalent} to~$\Phi$.

\smallskip

The size of automaton~$\Aut_\Phi$ is in
$\tuple{\EXP{k}{\size\phi},\EXP{k}{\size\phi},\EXP{({k-1})}{\size\phi},
  1,\allowbreak \EXP{({k-1})}{\size\phi}}$.  The~non-alternating
automaton~$\calN_\Phi$ then has size at most
$\tuple{\EXP{({k+1})}{\size\phi},\EXP{({k+1})}{\size\phi},\EXP{k}{\size\phi},
  1, \EXP{k}{\size\phi}}$,
and the \EQCTL[2] formula~$\Phi_\calN$ is in~$\EXP{({k+1})}{\size\phi}$.
\end{proof}

Note that our complexity results about the satisfiability of \EQCTL[2]
and \QCTL[k] entail that any translation procedure to get such an
\EQCTL[2] formula has time complexity at least
$\EXP{(k-1)}{\size\Phi}$.

\smallskip
We~then have $\QCTLs \eqex{} \QCTL \eqex{} \EQCTL[2] \caplog \AQCTL[2]$. 
But there is a difference between \QCTL[2], \QCTL[1] and \CTL.
The following theorem summarises our expressiveness results:

\begin{theorem}
  In terms of their relative expressiveness, the fragments of~\QCTL
  satisfy the following relations:
  \par\bgroup\centering
    \begin{tikzpicture}[yscale=.6,xscale=1.2]
      \path (1.4,0) node (QCTL) {\QCTLs};
      \path (-1,0) node (Q2cap) {$\EQCTL[2]\caplog\AQCTL[2]$};
      \path (-3.5,0) node (Q1) {\QCTL[1]};
      \path (-5.2,.7) node (EQ1) {\EQCTL[1]};
      \path (-5.2,-.7) node (AQ1) {\AQCTL[1]};
      \path (-7.5,0) node (Q1cap) {$\EQCTL[1]\caplog\AQCTL[1]$};
      \path (-9.7,0) node (CTL) {\CTL};
      \path (Q2cap) -- (QCTL) node[midway,sloped] {$\eqex{}$};
      \path (Q2cap) -- (Q1) node[midway,sloped] {$\prec$};
      \path (Q1) -- (EQ1) node[midway,sloped] {$\prec$};
      \path (Q1) -- (AQ1) node[midway,sloped] {$\prec$};
      \path (Q1cap) -- (EQ1) node[midway,sloped] {$\prec$};
      \path (Q1cap) -- (AQ1) node[midway,sloped] {$\prec$};
      \path (Q1cap) -- (CTL) node[midway,sloped] {$\prec$};
    \end{tikzpicture}
    \par\egroup
\end{theorem}

\begin{proof}

We first prove that \EQCTL[2] is \kl{strictly more expressive} than~\QCTL[1].
By~duality, this extends to~\AQCTL[2].
We~already proved that \QCTL, hence also~\QCTL[1], can be translated
in~\EQCTL[2] and in~\AQCTL[2].
We~exhibit an \EQCTL[2] formula that \QCTL[1] cannot express,
namely:
\[
\lambda = \exists p.\ \forall q.\ [
  \EX(p \et (\AX q \ou \AX\non q) 
  \et
  \EX(\non p \et (\AX q \ou \AX\non q) 
].
\]
It~specifies that there exist at least two (immediate) \kl{successors}
whose arity is~$1$.
Consider the \kl{trees}~$\calT_k$ and~$\calT_k'$ depicted at Fig.~\ref{fig-expr}.
We~prove that~$\lambda$ holds in~$\calT'_k$, but fails to hold
in~$\calT_k$:
in~$\calT'_k$, take the $p$-labelling where only~$r'_1$ is labelled
with~$p$: then for any $q$-labelling, $r'_1$ satisfies $p \et (\AX
q \ou \AX\non q)$ and $r'_2$ satisfies $\non p \et (\AX q \ou \AX\non
q)$, so that $\lambda$ holds in~$\calT'_k$. Now, take~any
$p$-labelling of~$\calT_k$, and the $q$-labelling in which exactly one
of the \kl{successors} of each node~$t_i$ is labelled with~$q$. Then none
of the states~$t_i$ can be used to satisfy any of the two conjuncts
of~$\lambda$, and $r$~alone can't satisfy both. Hence $\calT_k$ does
not satisfy~$\lambda$.

\begin{figure}[tb]
\centering
\begin{tikzpicture}[inner sep=0pt,yscale=.8,xscale=.9]
\pgfarrowsdeclarecombine{latex2'}{latex2'}%
{latex'}{latex'}{latex'}{latex'}

\begin{scope}[yshift=4cm]
\begin{scope}
\draw(0,1.3) node[moyrond,vert] (s) {} node[left=15mm]
  {$\calK_{k}$} node[left=2mm]{$s_0$}; 
\draw(-1.5,0) node[moyrond,vert] (r) {} node[left=3mm] {$r$};
\draw(1.5,0) node[moyrond,vert] (t1) {} node[right=3mm] {$t$};
\draw (0,-1.3) node[moyrond,vert] (aux1) {};
\everymath{\scriptstyle}
\draw[-latex'] (s) -- (r);
\foreach \i in {-34,-20,-8,20} 
  {\draw (s) edge[-latex',bend left=\i] (t1);}
\draw[thin] ($(s)+(-80:6mm)$) arc[thin,start angle=-80,end angle=10,radius=6mm] node[above right] {$k$};
\draw (s) edge[-latex',dashed,bend left=8] (t1);
\draw[-latex'] (r) -- (aux1);
\draw[-latex'] (t1) edge[-latex',bend left=10] (aux1);
\draw[-latex'] (t1) edge[-latex',bend left=-10] (aux1);
\draw[-latex'] (aux1) edge[out=-70,in=-110, looseness=7] (aux1);
\end{scope}

\begin{scope}[xshift=8cm]
\draw(0,1.3) node[moyrond,vert] (s) {} node[left=15mm]
  {$\calK'_{k}$} node[left=2mm]{$s'_0$}; 
\draw(-1.5,0) node[moyrond,vert] (r) {} node[left=3mm] {$r'$};
\draw(1.5,0) node[moyrond,vert] (t1) {} node[right=3mm] {$t'$};
\draw (0,-1.3) node[moyrond,vert] (aux1) {};
\everymath{\scriptstyle}
\draw[-latex'] (s) edge[-latex',bend left=10] (r);
\draw[-latex'] (s) edge[-latex',bend left=-10] (r);
\foreach \i in {-34,-20,-8,20}
  {\draw (s) edge[-latex',bend left=\i] (t1);}
\draw (s) edge[-latex',dashed,bend left=8] (t1);
\draw[thin] ($(s)+(-80:6mm)$) arc[thin,start angle=-80,end angle=10,radius=6mm] node[above right] {$k$};
\draw[-latex'] (r) -- (aux1);
\draw[-latex'] (t1) edge[-latex',bend left=10] (aux1);
\draw[-latex'] (t1) edge[-latex',bend left=-10] (aux1);
\draw[-latex'] (aux1) edge[out=-70,in=-110, looseness=7] (aux1);
\end{scope}
\end{scope}

\draw(0,1.3) node[moyrond,vert] (s) {} node[left=10mm] {$\calT_{k}$}; %
\draw(-1.5,0) node[moyrond,vert] (r) {} node[left=2mm] {$r$};
\draw(0,0) node[moyrond,vert] (t1) {} node[left=2mm] {$t_1$};
\draw(3,0) node[moyrond,vert] (tp) {} node[left=2mm] {$t_k$};
\draw (-1.5,-1) node[moyrond,vert] (a1) {};
\draw (-0.6,-1) node[moyrond,vert] (a2) {};
\draw (0.6,-1) node[moyrond,vert] (a3) {};
\draw (2.4,-1) node[moyrond,vert] (a4) {};
\draw (3.6,-1) node[moyrond,vert] (a5) {};
\draw[-latex'] (s) -- (r);
\draw[-latex'] (s) -- (t1);
\draw[-latex'] (s) -- (tp);
\draw[-latex'] (r) -- (a1);
\draw[-latex'] (t1) -- (a2);
\draw[-latex'] (t1) -- (a3);
\draw[-latex'] (tp) -- (a4);
\draw[-latex'] (tp) -- (a5);
\draw (-1.5,-2) node[moyrond,vert] (aux1) {};
\draw (-0.6,-2) node[moyrond,vert] (aux2) {};
\draw (0.6,-2) node[moyrond,vert] (aux3) {};
\draw (2.4,-2) node[moyrond,vert] (aux4) {};
\draw (3.6,-2) node[moyrond,vert] (aux5) {};
\draw[-latex'] (a1) -- (aux1);
\draw[-latex'] (a2) -- (aux2);
\draw[-latex'] (a3) -- (aux3);
\draw[-latex'] (a4) -- (aux4);
\draw[-latex'] (a5) -- (aux5);
\draw (1.3,0) node(points){$\ldots$};
\foreach \n in {aux1,aux2,aux3,aux4,aux5}
  {\draw[dashed] (\n.-90) -- +(0,-.4);}

\begin{scope}[xshift=8cm]
\draw(0,1.3) node[moyrond,vert] (s) {} node[left=10mm] {$\calT'_{k}$};%
\draw(-2.7,0) node[moyrond,vert] (r) {} node[left=2mm] {$r_1'$};
\draw(-1.5,0) node[moyrond,vert] (u) {} node[left=2mm] {$r_2'$};
\draw(0,0) node[moyrond,vert] (t1) {} node[left=2mm] {$t'_1$};
\draw(3,0) node[moyrond,vert] (tp) {} node[left=2mm] {$t'_k$};
\draw (-2.7,-1) node[moyrond,vert] (a0) {};
\draw (-1.5,-1) node[moyrond,vert] (a1) {};
\draw (-0.6,-1) node[moyrond,vert] (a2) {};
\draw (0.6,-1) node[moyrond,vert] (a3) {};
\draw (2.4,-1) node[moyrond,vert] (a4) {};
\draw (3.6,-1) node[moyrond,vert] (a5) {};
\draw[-latex'] (s) -- (r);
\draw[-latex'] (s) -- (u);
\draw[-latex'] (s) -- (t1);
\draw[-latex'] (s) -- (tp);
\draw[-latex'] (r) -- (a0);
\draw[-latex'] (u) -- (a1);
\draw[-latex'] (t1) -- (a2);
\draw[-latex'] (t1) -- (a3);
\draw[-latex'] (tp) -- (a4);
\draw[-latex'] (tp) -- (a5);
\draw (-2.7,-2) node[moyrond,vert] (aux0) {};
\draw (-1.5,-2) node[moyrond,vert] (aux1) {};
\draw (-0.6,-2) node[moyrond,vert] (aux2) {};
\draw (0.6,-2) node[moyrond,vert] (aux3) {};
\draw (2.4,-2) node[moyrond,vert] (aux4) {};
\draw (3.6,-2) node[moyrond,vert] (aux5) {};
\draw[-latex'] (a0) -- (aux0);
\draw[-latex'] (a1) -- (aux1);
\draw[-latex'] (a2) -- (aux2);
\draw[-latex'] (a3) -- (aux3);
\draw[-latex'] (a4) -- (aux4);
\draw[-latex'] (a5) -- (aux5);
\draw (1.3,0) node(points){$\ldots$};
\foreach \n in {aux0,aux1,aux2,aux3,aux4,aux5}
  {\draw[dashed] (\n.-90) -- +(0,-.4);}
\end{scope}
\end{tikzpicture}
\caption{Two (families of) Kripke structures and their computation trees}
\label{fig-expr}
\end{figure}

We now prove that $\calT_k$ and~$\calT'_k$ satisfy the same \QCTL[1]
formulas of size at most~$k$. For convenience, we~define a binary
relation~$\calR_k$ between states of~$\calT_k$ and states
of~$\calT'_k$, by letting
$\calR_k=\{(s_0,s'_0),(r,r'_1),(r,r'_2)\} \cup
\{(t_i,t'_j)\mid 1\leq i,j\leq k\}$.
The~proof then proceeds in two steps:
\begin{itemize}
\item we~first prove that all states in relation by~$\calR_k$ satisfy
  the same \EQCTL[1] formulas of size at most~$k$;
\item we~then prove that $s_0$ and~$s'_0$ satisfy the same \QCTL[1]
  formulas of size at most~$k$.
\end{itemize}

For the first step: the~result is straightforward (as~the~subtrees are
the~same) for all pairs of states in~$\calR_k$ but~$(s_0,s'_0)$.  Take
a formula~$\phi$ in~\EQCTL[1]. If~$\phi$~is in~\CTL, the~result is
clear as the subtrees are bisimilar. We~thus consider the case where
$\phi=\exists (p_i)_i.\ \psi$, where $\psi$ can be written as a
boolean combination of at~most~$k$ \CTL formulas of the
form~$\Ex\zeta_i$ or~$\All\zeta_i$.

Assume $\calT_k,s_0\models\phi$, and consider a labelling~$\ell$
of~$\calT_k$ with atomic propositions~$(p_i)_i$ witnessing this
fact. Consider the labelling~$\ell'$ of~$\calT'_k$ where the subtree
under each~$t'_j$ is labelled in the same way as~$\ell$ labels the
subtree under the corresponding~$t_j$, and the subtrees under~$r'_1$
and~$r'_2$ are labelled in the same way as the subtree
under~$r$. The~labelled trees~$\calT_k$ and~$\calT'_k$ are then
bisimilar; since~$\psi$~holds in~$\calT_k$ with labelling~$\ell$,
it~also holds in~$\calT'_k$ with labelling~$\ell'$.

Conversely, assume that $\calT'_k,s_0\models\phi$, and take a
labelling~$\ell'$ witnessing this. Consider a first labelling~$\ell$
of~$\calT_k$ in which the subtree under each~$t_j$ is labelled in the
same way as the subtree under~$t'_j$, and the subtree under~$r$ is
labelled in the same way as the subtree under~$r_1$. All subformulas
of~$\psi$ of the form~$\All\zeta_i$ that hold true at~$s'_0$
in~$\calT_k$ labelled with~$\ell'$ also hold true at~$s_0$
in~$\calT_k$ with~$\ell$, since the paths in the latter are paths in
the former. Similarly, all subformulas of~$\psi$ of the
form~$\Ex\zeta_i$ that hold true at~$s'_0$ in~$\calT'_k$ under~$\ell'$
also hold true at~$s_0$ in~$\calT_k$ under~$\ell$, except for those
that are witnessed by the path through~$r'_2$, which is the only path
that has no counterpart in~$\calT_k$ under labelling~$\ell$. However,
since there are at most~$k$ such subformulas, at least one path
in~$\calT_k$, say one going through~$t_k$, is not used to fulfill any
of the $\zeta_i$ subformulas. We~then update the labelling~$\ell$ by
labelling \emph{both} \kl{branches} under~$t_k$ in the same way as~$\ell'$
labels the subtree under~$r'_2$. This~way, the subtrees under~$r'_2$
and under~$t_k$ are bisimilar, hence $\ell$ now fulfill all the
required subformulas, so that~$\phi$ also holds in~$s_0$.

\smallskip

The second step of the proof is easy: take a formula~$\phi$
in~$\QCTL[1]$ of size at most~$k$. By~definition of~$\QCTL[1]$, it~can
be written as $\phi[(\psi_i)_i]$, where $\psi_i$ are $\EQCTL[1]$
formulas. We~label the nodes of~$\calT_k$ and~$\calT'_k$ with new
atomic propositions~$(p_i)_i$, in such a way that node~$n$ is labelled
with~$p_i$ if, and only~if, it~satisfies~$\psi_i$. Since all $\psi_i$
have size at most~$k$, thanks to the result of the first step, the
resulting labelled \kl{trees} are bisimilar.
Hence they satisfy the same
\CTL formulas, in particular they both do or both don't
satisfy~$\phi$, which concludes our proof.

\medskip

It~remains to settle the relative expressiveness of~\EQCTL[1],
\AQCTL[1], and \QCTL[1]. We~first show $\EQCTL[1] \caplog \AQCTL[1]
\prec \EQCTL[1]$. For this, it is sufficient to provide an \EQCTL[1]
formula~$\phi$ such that no \AQCTL[1] formulas are \kl(F){equivalent}
to~$\vfi$. Consider $\vfi = \exists p.\ (\EX p \et \EX \non p)$, which
characterises all nodes having at least two \kl{successors}.  Now,~consider
an \AQCTL[1] formula $\psi = \forall p_1\ldots p_n.\ \widetilde{\psi}$
with $\widetilde{\psi}\in \CTL$, and assume that $\psi$~is \kl(F){equivalent}
to~$\vfi$. Let~$\calK_1 = (\{s_0,s_1\},\{(s_0,s_1),(s_1,s_1)\},\emptyset)$ and
$\calK_2 =
(\{s_0,s_1,s_2\},\{(s_0,s_1),(s_0,s_2),\allowbreak (s_1,s_1),\allowbreak (s_2,s_2)\},\emptyset)$ be
two \kl{Kripke structures} such that $s_0$ in $\calK_1$ (resp.\ $s_0$ in
$\calK_2$) has one \kl{successor} (resp.~two \kl{successors}). Therefore
${\calK_1,s_0 \not\sat \vfi}$ and ${\calK_2,s_0 \sat \vfi}$.

If $\psi$ is \kl(F){equivalent} to~$\vfi$, then $\calK_2,s_0 \sat \forall
p_1\ldots p_n.\ \widetilde{\psi}$. Therefore, for any labelling of the
tree~$\calT_{\calK_2,s_0}$ with propositions~$p_1$ to~$p_n$, the~\CTL
formula~$\widetilde{\psi}$ is satisfied. This~is in particular true of
the labellings that label both \kl{branches} of~$\calT_{\calK_2,s_0}$ in the
same way; since \CTL cannot distinguish between bisimilar structures,
we~deduce that $\calK_1,s_0 \sat \forall p_1\ldots
p_n.\ \widetilde{\psi}$; this~contradicts the hypothesis that ${\vfi \equivF \psi}$.
Since $\EQCTL[1] \caplog \AQCTL[1]$ is closed under negation, we~also get 
$\EQCTL[1] \caplog \AQCTL[1] \prec \AQCTL[1]$.

Finally, we~notice that $\EQCTL[1]\caplog\AQCTL[1]$ strictly
contains~\CTL: indeed, consider the property~$\textsf{even}(p)$, which
characterises all trees in which all nodes at even depth are labelled
with some atomic proposition~$p$ (and~in which all nodes at odd depth
may or may not be labelled with~$p$). It~is well-known that such a
property cannot be expressed in~\CTL~\cite{Wol83}. We~now express~it
in both~\EQCTL[1] and~\AQCTL[1]:
\begin{itemize}
\item in~\EQCTL[1], we~first label all nodes at even depth with a new
  atomic proposition~$q$, and require that all nodes labelled with~$q$
  must also be labelled with~$p$:
  \[
  \exists q.\ (q \et \AG(q \Leftrightarrow \AX\non q)) \et \AG
  (q\Rightarrow p);
  \]
\item in~\AQCTL[1], we~write that any labelling that labels exactly all
  nodes at even depth with~$q$ (there~is a unique such labelling) 
  is such that all
  nodes labelled with~$q$ are also labelled with~$p$:
  \[
  \forall q.\ (q \et \AG(q \Leftrightarrow \AX\non q)) \Rightarrow
  \AG(q\Rightarrow p).\tag*{\qedsymbol}
  \]
\end{itemize}
\let\qed\relax
\end{proof}

\section{Application to \MSO}

\label{sec-mso}

We~first briefly review \intro{Monadic Second-Order Logic} (\reintro*\MSO)
over finite or infinite \kl{trees}.  We~use constant monadic
predicates~$\Pred{a}$ for ${a\in \AtP}$ and a relation~$\Edg$
for the immediate successor relation in a \Stree{2^\AtP}
$\Tree=\tuple{\tree,\lab}$.

\MSO is built with first-order (or individual) variables for vertices
(denoted with lowercase letters~$x,y,...$), and monadic second-order
variables for sets of vertices (denoted with uppercase
letters~$X,Y,...$).  Atomic formulas are of the form $x=y$,
$\Edg(x,y)$ (to represent the immediate successor relation), $x<y$ (the transitive closure of $\Edg$), $x\in X$, and~$\Pred{a}(x)$. General \MSO formulas are
constructed from those atomic formulas using the boolean connectives and the
first- and second-order quantifiers~$\exists^1$ and~$\exists^2$, which
we both denote with~$\exists$ in the sequel as long as this is not
ambiguous.  We~write $\phi(x_1,...,x_n,X_1,...,X_k)$ to state that
$x_1,...,x_n$ and $X_1,...,X_k$ may appear free (\ie,~not within the
scope of a quantifier) in~$\phi$. A~closed formula is a formula that contains no free
variables.

We~use the standard semantics for~\MSO over trees: given a tree~$\Tree$, a
sequence of nodes~$s_1$ to~$s_n$, and a sequence of sets of
nodes~$S_1$ to~$S_k$, we~write $\Tree,s_1,...,s_n,S_1,...,S_k \models
\phi(x_1,...,x_n,X_1,...,X_k)$ to indicate that $\phi$~holds on~$\Tree$
when variables~$x_1$ to~$x_n$ in~$\phi$ are replaced with~$s_1$ to~$s_n$, 
and variables~$X_1$ to~$X_k$ are replaced with~$S_1$ to~$S_k$.
As~an example, the closed formula
\[
\forall x.\ \bigl(P_a(x) \Rightarrow
\bigl[\exists X.\  (x\in X \et
\forall y.\ (y\in X \Rightarrow \exists z.\ (z\in X \et \Edg(y,z))))
\bigr]
\bigr)
\]
holds true in a tree whenever any node labelled with~$a$ belongs to an
infinite branch.
In~the~sequel, we~may group quantifiers into \emph{blocks of
quantifiers}, writing e.g. $\exists \{x_1,x_2,X_1,x_3\}. \phi$ in
place of~$\exists x_1. \exists x_2. \exists X_1.  \exists x_3.\ \phi$.
We~may also restrict first-order quantification, writing e.g.
$\exists x\in X.\ \phi(x)$ in place of~$\exists x.\ (x\in X \et \phi(x))$,
and
$\forall y\in Y.\ \phi(y)$ in place of $\forall y.\ (y\in Y\Rightarrow \phi(y))$.

\medskip

In the following we consider an \MSO formula $\Phi$ in negated-normal
form (\ie,~containing no universal quantifiers, and~in~which negations
may only be followed by an existential quantifier or an atomic formula
of~the form $\Edg(x,y)$, $x<y$, $x=y$, $\Pred{a}(x)$ or $x \in X$).
Using De~Morgan's laws and the fact that $\forall x.\ \phi$ is
equivalent to~$\neg\exists x.\ \neg\phi$, any \MSO formula can be
turned into that form.  Even if it means renaming variables, we~assume
that every quantifier bounds a different variable.

\medskip

In the next section, for any given \MSO formula~$\Phi$, we~present the
construction of an \AAPTA~$\Aut_\Phi$accepting exactly the models
of~$\Phi$. As~for~\QCTL, existential quantification in the formula
will be encoded using projection, hence it requires alternation
removal. Since this operation involves an exponential blow-up, it~will
be crucial to group existential quantifiers together as much as
possible. This will be possible as long as no negation appears
inbetween.

In~order to formalise this, we~define three families of sets of \MSO formulas,
namely $(\Sigma_k)_{k\in\bbN}$, $(\Pi_k)_{k\in\bbN}$ and
$(\Delta_k)_{k\in\bbN}$, which we will use to define a notion of
\emph{quantifier alternation}, mixing first- and second-order. This~will be
a crucial parameter for our construction.
All the sets are closed under positive boolean combinations, and:
\begin{itemize}
\item at level~$0$, the sets are defined as follows:
  \begin{itemize}
  \item $\Sigma_0$ uses atomic formulas and
    there negations as base formulas, and is closed under non-negated existential
    quantification and positive boolean combinations;
  \item $\Pi_0$ uses negations of formulas in~$\Sigma_0$ as base
    formulas, and is closed under positive boolean combinations;
  \item $\Delta_0$ contains $\Sigma_0\cup\Pi_0$ as base formulas, and
    is closed under positive boolean combinations.
  \end{itemize}
\item the sets at level~$k+1$ are defined from those at level~$k$ as follows:
  \begin{itemize}
  \item $\Sigma_{k+1}$ contains all formulas in~$\Delta_k$ as base
    formulas, and is closed under non-negated existential
    quantification and positive boolean combinations;
  \item $\Pi_{k+1}$ contains negations of formulas in~$\Sigma_{k+1}$ as base formulas,
    and is closed under positive boolean combinations;
  \item $\Delta_{k+1}$ contains $\Sigma_{k+1}\cup\Pi_{k+1}$ as base formulas, and
    is closed under positive boolean combinations.
  \end{itemize}
\end{itemize}
Any \MSO formula in negated-normal form is (syntactically)
in~$\Delta_k$ for some~$k$; we~call \emph{quantifier alternation}
of~$\phi$ the smallest~$k$ such that $\phi\in\Delta_k$. Using
De~Morgan's laws and the four equivalences given below, we~easily
obtain the following lemma:
\begin{lemma}
Any formula~$\phi$ in~$\Delta_k$ can be rewritten as a positive
boolean combination of subformulas in prenex-normal
form,~\ie,~of~subformulas of the form $\exists \calQ_1.\neg \exists
\calQ_2.\neg\exists\calQ_3\ldots \neg\exists\calQ_p.\ \phi$ (and~their
negations) where each~$\calQ_i$ is a block mixing first- and
second-order variables, and where $p\leq k+1$.
\end{lemma}

The~equivalences needed to rewrite formulas in prenex form are:
\begin{xalignat*}2 
(\vfi \et \exists \calQ.\psi)  & \equiv  \exists\calQ.(\vfi \et \psi) &
(\vfi \et \non  \exists \calQ.\psi)  & \equiv  \non \exists\calQ.(\non \vfi \ou \psi) \\
(\vfi \ou \exists \calQ.\psi)  & \equiv  \exists\calQ.(\vfi \ou \psi)  & 
(\vfi \ou \non \exists \calQ.\psi)  & \equiv  \non \exists\calQ.(\non \vfi \et \psi) 
\end{xalignat*}
Notice that the size of the resulting formula is linear in the size of
the original formula.

For example, we have:
\begin{xalignat*}1
\hbox to 1.2cm{$\exists X_1.\ \Big[ (\exists x_1 .\ \psi_1) \et \Big(\exists X_2. (\exists X_3.\ \psi_2 \et \non \exists x_2. (\psi_3 \et \exists X_4.\ \psi_4))\Big) \Big] \et \non \exists x_5.\ \psi_5$}\\
{} \equiv {}& \exists \{X_1,x_1,X_2,X_3\}.\ \Big[ \psi_1 \et   \psi_2 \et \non \exists \{x_2,X_4\}.\  \Big(\psi_3 \et \psi_4 \Big) \Big] \et \non \exists x_5.\ \psi_5\\
{} \equiv{} & \exists \{X_1,x_1,X_2,X_3\}.\  \non \exists \{x_2,X_4\} .\ \Big[ \non \Big(\psi_1 \et  \psi_2\Big) \ou  \Big(\psi_3 \et \psi_4 \Big) \Big] \et \non \exists x_5.\ \psi_5 .
\end{xalignat*}
In this example, the resulting formula contains two nested existential
blocks separated with a negation, so it is in~$\Delta_1$ (its~quantifier alternation is
one).

\medskip
More details about \MSO can be found e.g.~in~\cite{Tho97b}.
In~\cite{LM14}, it is proved that 
\MSO and \QCTL are \kl{equally expressive} over trees.
This could be used to define translations between \MSO
and \kl{\EU-automata}, but we prefer direct, more efficient
constructions, which we develop below.

\subsection{From \MSO to \AAPTA}

In this section, given a closed formula $\Phi \in \MSO$ in negated normal form and where the existential quantifications have been grouped together as explained above, we~build
an \AATA $\Aut_\Phi$ such that $\Lang(\Aut_\Phi)$
is the set of all trees satisfying~$\Phi$.
Actually, for any (non-closed) \MSO formula $\phi(x_1,...,x_n,X_1,...,X_k)$,
we~build an automaton~$\Aut_\phi$ such that,
for any nodes~$s_1$ to~$s_n$ and any sets~$S_1$ to~$S_k$,
it~holds
$\Tree,s_1,...,s_n,S_1,...,S_k \models
\phi(x_1,...,x_n,X_1,...,X_k)$ if,
and only~if,
the~\Stree{2^{\AtP \cup \{x_1,...,x_n,X_1,...,X_k\}}}~$\Tree'$,
obtained from~$\Tree$ by labelling any node~$t$ with~$x_i$ if~$t=s_i$, and with~$X_j$ if~$t\in S_j$, 
belongs to $\Lang(\Aut_\phi)$.

\bigskip

The automaton $\Aut_\phi$ is built inductively on the structure
of~$\phi$.  Handling boolean connectives $\et$ and $\ou$ is achieved with
the corresponding operations of~\AATA. For~quantifications~$\exists
X.\ \psi$ or $\exists x\ .\psi$, we~use the \kl{projection} operation
(after having built an equivalent \kl{non-alternating} automaton
for~$\Aut_\psi$) exactly as for the \QCTL formulas~$\exists p.\ \psi$.
Before applying the simulation construction, we~add a verification step to ensure that every proposition
corresponding to a first-order variable labels exactly one node in
the tree; we will describe this construction below.  First we describe
several types of automata to deal with atomic \MSO formulas and their
negations. In~each~case, we~use a state~$\state_\top$ that accepts any
tree. We~have:
\begin{itemize} 
\item If $\phi$ is  $\Edg(x,y)$: 
We define the automaton $\Aut_{E}=\tuple{\{\state_E,\state'_E\}, \state_E,\delta_E,\prio_E}$  as follows:
\[
\delta_E(\state_E,\sigma) = \begin{cases}  \EUpair(
  \mset{\state'_E}
  ;\{\state_\top\}) & \mbox{if } x \in \sigma  \\
\EUpair(
  \mset{\state_E}
  ;\{\state_\top\}) & \mbox{otherwise}  \end{cases}
\quad \quad 
\delta_E(\state'_E,\sigma) = \begin{cases}  \top & \mbox{if } y \in \sigma  \\
\bot & \mbox{otherwise}  \end{cases}
\]
with $\prio_E(\state_E)=\prio_E(\state'_E)=1$. 
\item If $\phi$ is  $\non\Edg(x,y)$: We define the automaton $\Aut_{\bar{E}}=\tuple{\{\state_{\bar{E}},\state'_{\bar{E}}\}, \state_{\bar{E}},\delta_{\bar{E}},\prio_{\bar{E}}}$  as follows:
\[
\delta_{\bar{E}}(\state_{\bar{E}},\sigma) = \begin{cases}  \EUpair(\emptyset;\{\state'_{\bar{E}}\}) & \mbox{if } x \in \sigma  \\
\EUpair(%
  \mset{\state_{\bar E}};\{\state_\top\}) & \mbox{otherwise}  \end{cases}
\quad \quad 
\delta_{\bar{E}}(\state'_{\bar{E}},\sigma) = \begin{cases}  \bot & \mbox{if } y \in \sigma  \\
\top & \mbox{otherwise}  \end{cases}
\]
with $\prio_{\bar{E}}(\state_{\bar{E}})=\prio_{\bar{E}}(\state'_{\bar{E}})=1$.
\item If $\phi$ is  $x<y$: 
We define the automaton $\Aut_{<}=\tuple{\{\state_<,\state'_<\}, \state_<,\delta_<,\prio_<}$  as follows:
\[
\delta_<(\state_<,\sigma) = \begin{cases}  \EUpair(%
  \mset{\state'_<};\{\state_\top\}) & \mbox{if } x \in \sigma  \\
\EUpair(%
  \mset{\state_<};\{\state_\top\}) & \mbox{otherwise}  \end{cases}
\quad \quad 
\delta_E(\state'_<,\sigma) = \begin{cases}  \top & \mbox{if } y \in \sigma\\
\hbox to 6mm{$\EUpair(%
\mset{\state'_<};\{\state_\top\})$} \\ & \mbox{otherwise}
\end{cases}
\]
with $\prio_<(\state_<)=\prio_<(\state'_<)=1$. 
\item If $\phi$ is  $\non(x<y)$: We define the automaton $\Aut_{\not{<}}=\tuple{\{\state_{\not{<}},\state'_{\not{<}}\}, \state_{\not{<}},\delta_{\not{<}},\prio_{\not{<}}}$  as follows:
\[
\delta_{\not{<}}(\state_{\not{<}},\sigma) = \begin{cases}  \EUpair(\emptyset;\{\state'_{\not{<}}\}) & \mbox{if } x \in \sigma  \\
\EUpair(%
  \mset{\state_{\not<}};\{\state_\top\}) & \mbox{otherwise}  \end{cases}
\quad \quad 
\delta_{\not{<}}(\state'_{\not{<}},\sigma) = \begin{cases}  \bot & \mbox{if } y \in \sigma  \\
\hbox to 6mm{$\EUpair(\emptyset;\{\state'_{\not{<}}\})$} \\ & \mbox{otherwise}  \end{cases}
\]
with $\prio_{\not{<}}(\state_{\not{<}})=1$ and $\prio_{\not{<}}(\state'_{\not{<}})=0$.

\item If $\phi$ is $x=y$: 
We define the automaton $\Aut_{=}=\tuple{\{\state_=\}, \state_=,\delta_=,\prio_=}$  as follows:
\[
\delta_=(\state_=,\sigma)  = \begin{cases}  \top &  \mbox{if}\: x,y \in \sigma  \\
\bot  & \mbox{if}\: (x \in \sigma \et y \not\in \sigma) \ou (x \not\in \sigma \et y \in \sigma) \\
\EUpair(%
\mset{\state_=};\{\state_\top\}) &  \mbox{otherwise}  \end{cases}
\]
with $\prio_=(\state_=)=1$.
\item If $\phi$ is $\non(x=y)$:  
We define the automaton $\Aut_{\neq}=\tuple{\{\state_{\neq}\}, \state_{\neq},\delta_{\neq},\prio_{\neq}}$  as follows:
\[
\delta_{\neq}(\state_{\neq},\sigma)  = \begin{cases}  \bot &  \mbox{if}\: x,y \in \sigma  \\
\top  & \mbox{if}\: (x \in \sigma \et y \not\in \sigma) \ou (x \not\in \sigma \et y \in \sigma) \\
\EUpair(%
\mset{\state_{\neq}};\{\state_\top\}) &  \mbox{otherwise}  \end{cases}
\]
with $\prio_{\neq}(\state_{\neq})=1$.
\item if  $\phi$ is $\Pred{a}(x)$: We define the automaton $\Aut_{a}=\tuple{\{\state_a\}, \state_a,\delta_a,\prio_a}$  as follows:
\[
\delta_a(\state_a,\sigma)  = \begin{cases}  \top & \mbox{if}\: x,a \in \sigma \\
\bot & \mbox{if}\: x \in \sigma \et a \not\in \sigma \\
\EUpair(%
\mset{\state_a};\{\state_\top\}) &  \mbox{otherwise}  \end{cases}
\]
with $\prio_a(\state_a)=1$. 
\item if  
$\phi$ is $\non\Pred{a}(x)$: We define the automaton $\Aut_{\bar{a}}=\tuple{\{\state_{\bar{a}}\}, \state_{\bar{a}},\delta_{\bar{a}},\prio_{\bar{a}}}$  as follows:
\[
\delta_{\bar{a}}(\state_{\bar{a}},\sigma)  = \begin{cases}  \top & \mbox{if}\: x \in \sigma  \et a \not\in \sigma\\
\bot & \mbox{if}\: x,a \in \sigma  \\
\EUpair(%
\mset{\state_{\bar a}};\{\state_\top\}) &  \mbox{otherwise}  \end{cases}
\]
with $\prio_{\bar{a}}(\state_{\bar{a}})=1$. 
\end{itemize}

The correctness of the constructions for $\Edg$ is stated as follows:
given a $2^{\AtP\cup\{x,y\}}$-labeled tree $\calT=(t,l)$ such that
there exists exactly one node $n\in t$ (resp.\ $n'\in t$) such that
$x\in l(n)$ (resp.~${y \in l(n')}$), we~have $\calT \in \calL(\Aut_E)$
if, and only~if, $\calT,n,n'\sat \Edg(x,y)$.  We~proceed in a similar
way for the other cases. The~correctness proofs are straightforward.

Now we can follow exactly the same steps as for \QCTL formulas. Let
$\Phi$ be an \MSO formula without any (first-order or second-order)
quantifier.  The previous automata constructions can then be composed
with union and intersection operations in order to get an \AAPTA whose
size is bounded by $\tuple{O(\size\Phi),O(\size\Phi),O(1),1,O(1)}$.

Consider a formula $\Phi=\exists \calV.\phi$ where $\calV$ is a set of
variables $\{x_1,\ldots x_m\} \cup \{X_1 \ldots X_p\}$ (where every
$x_j$ is a first-order variable and every $X_j$ is a second-order
variable) and $\phi$ is an \MSO formula without any quantifier. As
explained above, one can build an \AAPTA $\Aut_\phi$ for $\phi$. We
can also combine $\Aut_\phi$ with (a conjunction~of) automata
$\Aut_{x_j}$, for every $1 \leq j \leq m$, to~ensure that the letter
$x_j$ labels exactly one node. Such an automaton $\Aut_{x}$ is then
defined as $\Aut_{x}=\tuple{\{\state_{x}, \state_{\bar{x}}\},
  \state_{x},\delta_{x},\prio_{x}}$ with:
\[
\delta_x(\state_x,\sigma)  = \begin{cases}  \EUpair(\emptyset;\{\state_{\bar{x}}\}) &  \mbox{if}\: x \in \sigma  \\
\EUpair(\state_x \mapsto 1;\{\state_{\bar{x}} \}) & \mbox{otherwise} \end{cases}
\quad 
\delta_x(\state_{\bar{x}},\sigma)  = \begin{cases} \bot & \mbox{if}\: x \in \sigma  \\
 \EUpair(\emptyset;\{\state_{\bar{x}}\}) &   \mbox{otherwise} \end{cases}
\]
with $\prio_x(\state_x)=1$ and $\prio_x(\state_{\bar{x}})=0$.

This provides an \AAPTA $\Aut_{\phi'}$ whose size is bounded by
$\tuple{O(\size\phi),O(\size\phi),1,1,2}$ and which recognises
precisely the trees satisfying $\phi$ and where every first-order
variable $x_i$ labels exactly one node in the tree.
Applying the \kl{simulation} theorem (Theorem~\ref{thm-simu}), we~get
an \nAAPTA whose size is bounded by
$\tuple{\EXP{1}{\size{\phi}},\allowbreak
  \EXP{1}{\size{\phi}},\allowbreak \EXP{0}{\size{\phi}},\allowbreak
  1,\allowbreak\EXP{0}{\size{\phi}}}$. It~remains to use the
projection to get a (non-alternating) automaton which recognises
precisely the infinite trees satisfying the formula $\Phi=\exists
\calV.\phi$.
Note that dealing with the case $\Phi=\non\exists
\calV.\phi$ requires the addition of a complement operation, but this preserves the bound on the size of the final automaton, which still is $\tuple{\EXP{1}{\size{\phi}},\allowbreak
  \EXP{1}{\size{\phi}},\allowbreak \EXP{0}{\size{\phi}},\allowbreak
  1,\allowbreak\EXP{0}{\size{\phi}}}$.
    
 As for \QCTL and \QCTLs, one can extend the previous construction for
 any \MSO formula $\phi$. 
 An~important point is that we do not distinguish
 between first-order and second-order quantifiers: both quantifiers
 are treated in the same way, via the projection operation; in~both
 cases, each quantifier alternation
 induces an exponential blow-up,  due to the simulation step. 
 By~proceeding exactly as for~\QCTL, and with the specific treatment
 of first-order quantifiers as explained above, we~get:
\begin{theorem}
\label{th-mso-aut}
Given a formula in~$\Delta_k$ for some~$k\geq 0$, 
we~can construct an \mbox{\AAPTA}~$\Aut_\phi$ over~$2^\AtP$  accepting exactly the
  trees satisfying~$\phi$. 
  The~automaton~$\Aut_\phi$ has size $\tuple{\EXP{(k+1)}{\size\phi},\allowbreak\EXP{(k+1)}{\size\phi},\allowbreak\EXP{k}{\size\phi}, 1, \EXP{k}{\size\phi}}$.
\end{theorem}
 As a corollary, we get the following results about decision procedures for \MSO via \AAPTA construction:

\begin{corollary}
\label{mso-compl}
Let~$\phi$ be a formula in~$\Delta_k$.
The~satisfiability problem for~$\phi$ is in~\EXPTIME[(k+2)]. 
The~model-checking problem  for~$\phi$ over finite Kripke structures is in~\EXPTIME[(k+1)].
\end{corollary}

\subsection{From \AAPTA to \MSO}
Expressing acceptance of some tree~$\Tree$ by some~\nAAPTA~$\Aut$
as an \MSO formula
is based on the same techniques as the ones we used for~\QCTL:
an~existential (second-order) quantification is used to label every
node of~$\Tree$ with a unique state of~$\Aut$; the~rest of the formula
checks that the root is labeled with the initial state, that  the (non-alternating) transition function is fulfilled
locally at any node, and that for every infinite \kl{branch}  (which we
encode using second-order quantification), the parity condition is fulfilled.

Consider an \nAAPTA $\Aut=\tuple{\State,\state_0,\Trans,\prio}$ over
$\Alp=2 ^\AtP$.  Let~$Q$ be $\{\state_0,\ldots,\state_{n}\}$ and
let~$D$ be the set of \kl{priorities} in~$\prio$ (in~the following,
we~use~$D_{\geq k}$ to denote the subset of priorities greater than or
equal to~$k$).  Formally, we~define~$\Phi_\Aut$ as:
\begin{multline*}
\exists Q_0 \ldots Q_n.\ \exists x_\varepsilon.\
\Bigl( 
\Phi_{\text{\tiny root}}(x_\varepsilon) \et 
(x_\varepsilon \in Q_{0}) \et \Phi_{\Trans} \: \et \: \forall B . \
\bigl[ \Br(B)  \Rightarrow {}  \Phi_{\text{\tiny parity}}(B) \bigr] \Bigr)
\end{multline*}
 
In this formula, quantification over~$Q_0$ to~$Q_n$ is used to label
the nodes of~$\Tree$ with states of~$\Aut$, and quantification
over~$x_\varepsilon$ is used to characterise the \kl{root} of~$\Tree$;
subformula~$\Phi_{\Trans}$~ensures that each node is labelled with
exactly one state of~$\Aut$, and that the transition function is
satisfied; the last part states that every infinite branch satisfies the parity condition. 
We have $\Phi_{\text{\tiny root}}(x)$ defined as  $\non (\exists y.\ \Edg(y,x))$.

Consistency w.r.t.~the transition function is expressed as
follows:
\begin{multline*}
\Phi_{\Trans} = \forall x.\ \Bigl[ \ET_{0\leq i \leq n} \Bigl(x \in Q_i \:\impl\:  \Big(
   \ET_{j\not=i} x \not\in Q_j \et {} \\
  \OU_{R \in 2^\AtP}
    \bigl[\Phi_{R}(x) \et \Trans(\state_i,R)\stackrel{\text{\tiny?}}{\not=}\bot \et
  \bigl(\Trans(\state_i,R)\stackrel{\text{\tiny?}}=\top \ou{} \\
  \OU_{\EUpair(E_j;U_j)\in\Trans(\state_i,R)}
  \Phi_{\EUpair(E_j;U_j)}(x)\bigr)\bigr] \Big) \Bigr) \Bigr]
\end{multline*}
where the subformulas $\Trans(\state_i,R)\stackrel{\text{\tiny?}}{\not=}\bot$ and
$\Trans(\state_i,R)\stackrel{\text{\tiny?}}=\top$ on the second 
line are just replaced by~$\top$ for all~$R$ for which the
(in)equality holds, and with~$\bot$ otherwise.
For a subset $R \subseteq\AtP$, the formula $\Phi_{R}(x)$
specifies that $x$ is labelled exactly with the propositions
in~$R$: $\Phi_{R}(x)
= \ET_{p\in R} \Pred{p}(x) \et \ET_{p \in\AtP\setminus R} \non \Pred{p}(x)$.

The formula $ \Phi_{\EUpair(E;U)}(x)$ requires that the successors of
the node labelled with~$x$ satisfy the \EUprs $\EUpair(E;U)$ and it is
defined as follows:
\begin{multline*}
\Phi_{\EUpair(\mset{r_1,\ldots,r_k};\{s_1,\ldots,s_m\})}(x) = 
\exists x_1 \ldots x_k. \\
\Bigl(\ET_{1\leq i \leq k} \bigl(\Edg(x,x_i) \et \ET_{\substack{1\leq j\leq k\\ j\not=i}} x_i\not=x_j
\et x_i \in Q_{r_i} 
\bigr)
\et \\
 \forall z.\ \Bigl[ \bigl(\Edg(x,z) \et  \ET_{1\leq i \leq k} z \neq x_i \bigr) \impl \bigl( \OU_{1\leq i \leq m} z \in Q_{s_i} %
 \bigr)\Bigr]\Bigr)
\end{multline*}
where $Q_{r_i}$ (resp.\ $Q_{s_i}$) denotes the second-order variable associated with the state $r_i$ (resp.\ $s_i$). 

The formula~$\Br(B)$, stating that the set of nodes labelled with~$B$
forms an infinite \kl{branch} from the root $x_\varepsilon$, 
can be written as
\begin{multline*}
  \Br(B) =  (x_\varepsilon \in B) \et \forall x \in B.\ \Bigl( \exists y \in B. \Edg(x,y) \Bigr) \et {}\\
  \forall x,y,z \in B.\
  \Big( (\Edg(x,y) \et 
  \Edg(x,z)) \:\impl\: y=z \Big)  \et {}\\
  \forall x,y. \Big(x \not\in B \et \Edg(x,y) \:\impl\: y \not\in B\Big)
\end{multline*}
Finally, $\Phi_{\text{\tiny parity}}(B)$, stating that the least priority occurring infinitely many times along $B$ is even, is defined as follows:
\begin{multline*}
  \Phi_{\text{\tiny parity}}(B) =  \exists y \in B. \Big( \ET_{0\leq i \leq n} y \not\in Q_i \Big) \ou \Bigl( 
  \OU_{\stackrel{d \in D}{\mbox{\tiny $d$ even}}} \: \OU_{ \stackrel{0 \leq i \leq n}{\mbox{\tiny $\prio(q_i)=d$}}} y\in Q_i   \: \et \\
    \et 
   \forall z \in B. \Bigl[ z > y  \:\impl\:  
  \Bigl( \!\!\! \OU_{ \stackrel{0 \leq j \leq n}{\mbox{\tiny $\prio(q_i)\geq d$}}} \!\!\!\!\! z\in Q_j \Bigr) \et 
  \Bigl( \exists t \in B. t > z \:\et\: \!\!\! \OU_{ \stackrel{0 \leq k \leq n}{\mbox{\tiny $\prio(q_k)=d$}}}\!\!\!\!\! t \in Q_k\Bigr)\Bigr]
  \Bigr) 
   \end{multline*}
In this formula, $y$ represents a position labelled with a state such that
(a)~it~is~not labelled with any~$Q_i$ (hence~it~belongs to an accepted subtree whose root has a transition function equals to $\top$)
or (b)~it~is labelled with some $q_i$ with an even priority~$d$ such that (1)~every successor (along~$B$) is labelled with a state whose priority is  a greater than (or~equal~to)~$d$ and (2)~there~are infinitely many states along~$B$ with priority~$d$.

\medskip

Finally we can observe that the size of $\Phi_\Aut$ is in
$O(\size\State\cdot(\size\State+2^{\size\AtP}\cdot(\size\AtP+\sizeB\Trans\cdot(\sizeE\Trans+\sizeU\Trans)))+\priomax\cdot\size\State)$,
from which we can deduce that $\size\Phi_\Aut$ is in
$O({\size\State}^2\cdot2^{\size\AtP}\cdot(\size\AtP+\sizeB\Trans\cdot(\sizeE\Trans+\sizeU\Trans)))$.
Formula~$\Phi_\Aut$ contains four alternations of (first-~or
second-order) quantifiers. Note also that it contains two blocks of
second-order quantifiers, and there is only one alternation of second-order
quantifiers.
It~follows:
\begin{theorem}
  Any  closed
  formula in~$\Delta_k$ (with~$k>0$)
  can be translated into an equivalent
  closed formula in~$\Delta_4$ with at most
  one
  alternation of second-order quantifiers. The~size
  of the resulting formula can be bounded by~$\EXP{(k+2)}{\size\phi}$.
\end{theorem}

\begin{proof}
From $\phi$, we can build an \AAPTA $\Aut_\phi$ over $\Sigma$ whose
size is in
$\tuple{\EXP{({k+1})}{\size\phi},\allowbreak\EXP{(k+1)}{\size\phi},\EXP{k}{\size\phi},
  1, \EXP{k}{\size\phi}}$. Here the alphabet $\Sigma$ is
$2^{\AtP_\phi}$ where $\AtP_\phi$ is the set of monadic predicates
occurring in $\phi$ (and then $\size{\AtP_\phi} \leq \size\phi$).

Applying the simulation theorem provides us with an
\nAAPTA~$\NAut_{\phi}$ whose size is
$\tuple{\EXP{(k+2)}{\size\phi},\EXP{(k+2)}{\size\phi},\EXP{(k+1)}{\size\phi},
  1, \EXP{(k+1)}{\size\phi}}$. It~remains to build~$\Phi_{\NAut_{\phi}}$
as above to get the result.
\end{proof}

\section{Conclusion}

We have introduced a new class of symmetric tree automata (\AAPTA) for
 trees of arbitrary branching degrees. We showed that these automata
 have exactly the same expressive power as the temporal logics~\QCTL
 and~\QCTLs, and as the logic~\MSO: given a formula $\Phi$ in those
 formalisms, the set of infinite trees satisfying~$\Phi$ can be
 defined as the language of some automaton~$\calA_\Phi$, and
 conversely for any \AAPTA $\calA$ one can build a formula
 $\Phi_\calA$ whose models are precisely~$\calL(\calA)$.

In order to prove those results, we have developed 
algorithms for manipulating our \AAPTA, and have carefully studied
their complexities. This has allowed us to obtain decision procedures
for satisfiability and model checking for \QCTLs and its fragments
whose complexities match the lower-bound established in previous
papers~\cite{LM14}. It~also allowed us to obtain an effective translation
from \QCTL to \EQCTL[2], and similarly,
from \MSO to its fragment with only two second-order-quantifier alternations.

\newcommand{\etalchar}[1]{$^{#1}$}

\end{document}